\documentclass{article}
\usepackage[a4paper,margin=1in]{geometry}
\usepackage{graphicx}
\usepackage[colorlinks=true, linkcolor=blue, citecolor=red, urlcolor=black]{hyperref}
\usepackage{soul}
\usepackage{amsfonts}
\usepackage{amssymb}
\usepackage{amsmath}
    \numberwithin{equation}{section}
\usepackage{amsthm}
\usepackage{graphicx}
\usepackage{appendix}
\usepackage{bbold}
\usepackage{dsfont}
\usepackage{enumitem}
\usepackage{xcolor}
\usepackage{float}
\usepackage{indentfirst}
\usepackage{mathtools}
\usepackage{bm}
\usepackage{adjustbox}
\usepackage{cite}

\newcommand{\tp}{{\rm T}}

\newcommand{\Tr}{{\rm Tr}}

\newcommand{\lb}{\left [}
\newcommand{\rb}{\right ]}
\newcommand{\lp}{\left (}
\newcommand{\rp}{\right )}
\newcommand{\bb}{\mathbb}
\newcommand{\Arg}{{\rm Arg}}

\newcommand{\GL}{{\rm GL}}
\newcommand{\Hm}{{\rm H}}
\newcommand{\Sy}{{\rm S}}
\newcommand{\Cp}{{\bb C}}

\usepackage{subcaption}

\newtheorem{theorem}{Theorem}[section]
\newtheorem{corollary}{Corollary}[theorem]

\newtheorem{proposition}[theorem]{Proposition}

\renewcommand{\[}{\begin{equation}}
\renewcommand{\]}{\end{equation}}
\let\ophi\phi
\renewcommand{\phi}{\varphi}

\title{\textbf{An equivalence in random matrix and tensor models via a dually weighted intermediate field representation}}
\author{Juan Abranches$^{a}$\footnote{Email: \href{mailto:juan.abranches@oist.jp}{juan.abranches@oist.jp}} \qquad 
Alicia Castro$^{b}$\footnote{Email: \href{mailto:A.Castro@thphys.uni-heidelberg.de}{A.Castro@thphys.uni-heidelberg.de}} \qquad 
Reiko Toriumi$^{a}$\footnote{Email: \href{mailto:reiko.toriumi@oist.jp}
{reiko.toriumi@oist.jp}}
\\[3mm]
{\small $^a$ Okinawa Institute of Science and Technology Graduate University,} \\ {\small 1919-1, Tancha, Onna, Kunigami District, Okinawa 904-0495, Japan. }\\
{\small $^b$ Institute for Theoretical Physics, Heidelberg University, Heidelberg, Germany.}}
\date{ }

\begin{document}

\maketitle
\begin{abstract}
\noindent We present novel equivalences in random matrix and tensor models between complex and self-adjoint theories with nontrivial quadratic terms in the action, established through an intermediate field representation. More precisely, we show that the partition functions of certain self-adjoint models and their complex counterparts are different integral representations of the exact same function. A special case of these equivalences takes a form of newly found dually weighted intermediate field representations, which generalize the standard intermediate field representation. We also find indications of an equivalence between real tensor models and self-transpose tensor models.
\end{abstract}

\tableofcontents

\section{Introduction}

\noindent Random matrix models have long played a central role in two-dimensional (2D) Euclidean quantum gravity, where they provide a direct combinatorial and analytic handle on the sum over random surfaces.
In particular, \emph{self-adjoint random matrix models}~\cite{DiFrancesco:1993,DiFrancesco:2004qj} provide a direct handle on the sum over random oriented geometries. Through the dual-graph correspondence, the partition function of this matrix model was shown to be equal to that of \emph{Euclidean Dynamical Triangulations} (EDT)~\cite{Ambjorn:1985az}. 

EDT is a state-sum formulation of quantum gravity (QG). In this framework, the geometry of spacetime is discretized by a tessellation into intrinsically flat simplices, and the Einstein--Hilbert action is approximated by its Regge version~\cite{Regge:1961}. In contrast to quantum field theory (QFT) on a fixed background, in quantum gravity the lattice itself is dynamical, as its structure represents the gravitational degrees of freedom. The partition function is therefore obtained as the continuum limit of a state sum over all tessellations. In the matrix model formulation, the continuum limit is reached by sending the matrix size $N \to \infty$ and tuning the face weights to criticality. Due to the availability of powerful combinatorial and analytic tools~\cite{Budd2023} as well as known techniques from high-energy physics, these discrete models become considerably more tractable and their continuum limits can be studied analytically and numerically. For example, Functional Renormalization Group (FRG) techniques provide a systematic way to analyze this limit through the large-$N$ expansion~\cite{Eichhorn:2013isa}.
Altogether, this motivates the search for analytic frameworks \emph{beyond 2D}, capable of capturing higher-dimensional generalizations of EDT.

However, in higher dimensions, extensive numerical evidence shows that the continuum limit of EDT does not yield a geometric phase compatible with our universe. Instead, the path integral is dominated by a ``crumpled'' phase and a tree-like ``branched polymer'' phase~\cite{Ambjorn:1985az}. This shortcoming motivated the development of discrete approaches that enforce a causal structure at the level of the building blocks themselves, most notably \emph{Causal Set Theory}~\cite{Sorkin1987,Surya:2019ndm} and \emph{Causal Dynamical Triangulations} (CDT)~\cite{Ambjorn:1998,Ambjorn:2012jv,Loll:2019rdj}. In $D$-dimensional CDT, simplicial building blocks carry both spacelike and timelike $(D-1)$-faces, and gluing rules forbid configurations incompatible with a global time foliation. Crucially, CDT exhibits strong numerical evidence for the appearance of a genuine \emph{geometric phase} in the continuum limit~\cite{Loll:2019rdj}. This motivated the implementation of such causal structures in other QG state-sum formulations, including random matrix models. 

In the 2D case, \cite{Benedetti:2008hc} introduced a causal structure using dually weighted matrix models (DWMMs)~\cite{das_1990,Kazakov:1995gm,Kazakov:1996zm}. In these, the causal structure is imposed at the level of the Feynman rules by assigning different weights to faces of the ribbon graph. Concretely, the action contains interaction terms dressed with fixed ``weight matrices’’ that distinguish timelike from spacelike identifications of edges. In the dual triangulation, these weights control how different types of simplicial building blocks can meet, and thereby implement a combinatorial remnant of causality already at the level of the matrix integral. While this framework successfully reproduces features of CDT in two dimensions, the resulting models are analytically challenging: the presence of nontrivial face weights obscures the usual large-$N$ techniques and complicates the extraction of continuum limits.

By looking for a tool to regain analytic control over dually weighted matrix models, we encounter the \emph{intermediate field representation}. This technique introduces an auxiliary field that allows to rewrite quartic (or higher) matrix and tensor interactions in simpler forms. This effectively replaces complicated multitrace couplings by a Gaussian integrals over the auxiliary field. Despite its success in several matrix and tensor models~\cite{Bonzom2015ColoredTO,Lionni:2016ush,Rivasseau:2010ke,Gurau:2014lua,Nguyen:2014mga,Lionni2016NoteOT,Erbin:2019zug, bardy2025largenlimiton3invariantgeneral}, its application to theories
with dually-weighted interactions has remained largely unexplored.

Building on this foundation,
and motivated by the causal matrix models studied in~\cite{Benedetti:2008hc,Castro:2020dzt,abranches2025} as well as higher-dimensiomal analogues in the form of dually-weighted tensor models \cite{Benedetti:2011nn}, we consider certain complex matrix and tensor models that include an external rigidity matrix in the 
propagator. We find that the partition and correlation functions of these complex matrix/tensor models
are \emph{exactly equivalent} to certain self-adjoint two--matrix/tensor models via a dually-weighted intermediate field representations. Interestingly, examples of this type of equivalence have been observed in~\cite{Szabo:1996fj, difrancesco1992generatingfunctionfatgraphs}.
In this work, we formalize this equivalence for more general models by constructing an intermediate field representation for matrix/tensor models with nontrivial propagators and potentials.

The paper is organized as follows. In Section~\ref{sec:notations}, we begin with a summary of the notation and basic identities used throughout the paper.
In Section~\ref{sec:causalmm}, we review DWMMs, as well as the causal matrix model introduced by Benedetti and Henson \cite{Benedetti:2008hc} and motivate the introduction of an analogous complex matrix model. 
In Section \ref{sec:summarymain}, we summarize our main results.

In Section~\ref{sec:equivmm}, we introduce and prove equivalences between certain random matrix models of complex and self-adjoint matrices using the intermediate field representation. This section includes several examples illustrating these equivalences in random matrix theories, including the causal matrix model. It also contains ancillary results which, while not essential for our main computations, may be relevant in other contexts

In Section~\ref{sec:equivtensor}, we extend our results of equivalences to complex and self-adjoint random tensor models. This includes an exact intermediate-field equivalence between complex random tensor models of order $D$ and self-adjoint tensor models of order $2D$. 
Furthermore, we show that when the interaction potential exhibits additional symmetries, the self-adjoint formulation reduces to an effective tensor model of lower order $2(D-1)$.
This section also includes several illustrative examples. We conclude the section by presenting an example of an equivalence between real and self-transpose random tensor models, and comment on its potential as a toy model for CDT in three and higher dimensions.

Finally, in Section~\ref{sec:concl}, we summarize our results and give a general overview of possible future directions.

\section{Notations and essentials}\label{sec:notations}

\paragraph{Matrix multi-traces.}

\begin{itemize}

\item \textbf{Definition.}  
For $\sigma\in \Sy_n$ and $M$ a complex $N\times N$ matrix,
$
\Tr_{[\sigma]}(M)
= \sum_{k_1,\dots,k_n=1}^N
  \prod_{i=1}^n M_{k_ik_{\sigma(i)}}$, 
where $[\sigma]$ denotes the equivalence class of $\sigma$.

\item \textbf{Cycle-type dependence.}  
Set 
$\Tr_\lambda(M):=\prod_{j=1}^{n}\Tr(M^{\lambda_j})$, where $\lambda\vdash n$ is the partition corresponding to the equivalence class $[\sigma]$. Then $ \Tr_\lambda(M)=\Tr_{[\sigma]}(M).$

\item \textbf{Homogeneity.} For $z\in\mathbb{C}$
and $\lambda\vdash n$, 
$\Tr_\lambda(zM)=z^n\Tr_\lambda(M)$.

\item \textbf{Examples.} $
\Tr_{[n^1]}(M)=\Tr(M^n),\qquad
\Tr_{[1^n]}(M)=\Tr(M)^n,\qquad
\Tr_{[k^d]}(M)=\Tr(M^k)^d.$

\end{itemize}
\paragraph{Representations and characters.}

\begin{itemize}
\item \textbf{Symmetric group $\Sy_n$.} 
A partition $\lambda\vdash n$ labels an irreducible representation 
$\rho_\lambda:\Sy_n\to\GL(d_\lambda)$ with character
$\chi_\lambda(\sigma)=\Tr(\rho_\lambda(\sigma))$ and the dimension
of the representation is given by $d_\lambda=\chi_\lambda(e)$.
\item \textbf{Orthogonality and completeness (discrete).} 
For $\mu,\nu\vdash n$, and $\alpha, \beta \in S_n$,
\[
\frac{1}{n!}\sum_{\sigma\in \Sy_n} \chi_\mu(\sigma)^* \chi_\nu(\sigma)
     = 
     \delta_{\mu\nu}
     ,\qquad
\sum_{\lambda\vdash n} \chi_\lambda(\alpha)^* \chi_\lambda(\beta)
     = \frac{n!}{|[\alpha]|}\delta_{[\alpha][\beta]}.
\]

\item \textbf{Double-class identity (discrete convolution).}
$
\frac{1}{n!}\sum_{\sigma\in \Sy_n}
\chi_\lambda(\alpha\sigma\beta\sigma^{-1})
= \frac{1}{d_\lambda}\chi_\lambda(\alpha)\chi_\lambda(\beta).$

\item \textbf{Matrix characters of $\GL(N)$ and $U(N)$.}  
A partition $\lambda\vdash n$ labels an irreducible representation  
$\rho_\lambda:\GL(N)\to\GL(D_\lambda)$ with character
$
\chi_\lambda(M)=\Tr(\rho_\lambda(M))$ and dimension $D_\lambda=\chi_\lambda(\bb 1)$.
If $M$ has eigenvalues $m_1,\dots,m_N$, then
$
\chi_\lambda(M)=s_\lambda(m_i),$
the Schur polynomial.

\item \textbf{Unitary orthogonality.} For $U\in U(N)$
\[
\int_{U(N)} dU \;\chi_\mu(U)^* \chi_\nu(U)
= \delta_{\mu\nu}.
\]

\item \textbf{Character factorization over $U(N)$ 
(separation of $A$ and $B$).}
\[
\int_{U(N)} dU\;\chi_\lambda(A U B U^\dagger)
= \frac{1}{d_\lambda}\,\chi_\lambda(A)\chi_\lambda(B).
\]

\item \textbf{Schur--Weyl duality.}
\[
\Tr_{[\sigma]}(M)
= \sum_{\lambda\vdash n} \chi_\lambda(\sigma)\chi_\lambda(M),
\qquad
\chi_\lambda(M)
= \frac{1}{n!}\sum_{\sigma\in \Sy_n}
\chi_{\lambda}(\sigma)\Tr_{[\sigma]}(M).
\]

\end{itemize}
\paragraph{Index conventions.}

\begin{itemize}

\item \textbf{Multi-indices.}  
For a fixed $D\ge 1$, a multi-index is an element
$
\bm a = (a_1,\dots,a_D)\in\{1,\dots,N\}^{D}.$

\item \textbf{Deletion of an entry.}  
For $c\in\{1,\dots,D\}$,
$
\bm a_{\hat c} := (a_1,\dots,a_{c-1},a_{c+1},\dots,a_D).$

\item \textbf{Substitution.}  
For $\bm a,\bm b\in\{1,\dots,N\}^{D}$ and $c\in\{1,\dots,D\}$,
$
\bm a_{\hat c}\bm b := (a_1,\dots,a_{c-1},b_c,a_{c+1},\dots,a_D).$

\end{itemize}
\paragraph{Tensors.}

\begin{itemize}

\item \textbf{Complex tensors.}  
An order-$D$ complex tensor
$\phi\in\mathbb C_N^{\otimes D}$
has components
$\phi^{\bm a}\in\mathbb C.$

\item \textbf{Conjugate tensors.}  
The conjugate tensor $\phi^\dagger\in\mathbb C_N^{*\otimes D}$ has components
$
\phi^\dagger_{\bm a} = (\phi^{\bm a})^*.$

\item \textbf{Complex mixed tensors.}  
A tensor
$
P \in \mathbb C_N^{\otimes D}\otimes\mathbb C_N^{*\otimes D}
$
has components 
$P^{\bm a}_{\bm b}\in \Cp.$

\item \textbf{Conjugate mixed tensors.}  
The tensor
$
P^\dagger \in \mathbb C_N^{\otimes D}\otimes\mathbb C_N^{*\otimes D}
$
has components 
$[P^\dagger]^{\bm a}_{\bm b}=(P_{\bm a}^{\bm b})^*.$

\item \textbf{Self-adjoint tensors.} 
A self-adjoint tensor
$\Phi\in\mathrm H(\mathbb C_N^{\otimes D})
\subset
\mathbb C_N^{\otimes D}\otimes\mathbb C_N^{*\otimes D}$ is characterized by
$
\Phi^\dagger = \Phi.$

\item \textbf{Identity and transposition operators.} 
For $\bm a,\bm b,\bm c,\bm d\in\{1,\dots,N\}^{D}$, $[{\bb 1}^{\otimes D}]^{\bm a}_{\bm b}=\delta^{\bm a}_{\bm b}$ and
we define
$\Sigma^{\bm a \bm c}_{\bm b \bm d}:=\delta^{\bm a}_{\bm d}\delta^{\bm c}_{\bm b}$\;.
\end{itemize}
\paragraph{Tensor operations.}

\begin{itemize}

\item \textbf{Index contraction.}
Repeated indices are summed over. 
If $\phi,\psi\in \mathbb C_N^{\otimes D}$, then
\[\psi^\dagger_{\bm a}\phi^{\bm a}=\sum_{{\bm a}}\psi^\dagger_{\bm a}\phi^{\bm a},\qquad  \psi^\dagger_{\bm a}\phi^{{\bm b}_{\hat c}\bm a}=\sum_{a_c}\psi^\dagger_{\bm a}\phi^{{\bm b}_{\hat c}\bm a}.\]

\item \textbf{Tensor product and contraction.}  
If $\phi,\psi\in \mathbb C_N^{\otimes D}$, then
$
\phi\psi^\dagger \in \mathbb C_N^{\otimes D}\otimes \mathbb C_N^{*\otimes D},
\;
[\phi\psi^\dagger]^{\bm a}_{\bm b} = \phi^{\bm a}\psi_{\bm b}$
and
$
\psi^\dagger\phi \in \mathbb C,\; \psi^\dagger\phi = \psi^\dagger_{\bm a}\phi^{\bm a}.$

\item \textbf{Product of an operator and a tensor.}  
If $P\in \mathbb C_N^{\otimes D}\otimes \mathbb C_N^{*\otimes D}$ and  
$\phi\in \mathbb C_N^{\otimes D}$, then
$
P\phi \in \mathbb C_N^{\otimes D},
\;
[P\phi]^{\bm a} = P^{\bm a}_{\bm b}\phi^{\bm b},$
and
$
\phi^\dagger P \in \mathbb C_N^{*\otimes D},
\;
[\phi^\dagger P]_{\bm a} = \phi^\dagger_{\bm b}P^{\bm b}_{\bm a}.$

\item \textbf{Composition of operators.}  
If $P,Q\in \mathbb C_N^{\otimes D}\otimes \mathbb C_N^{*\otimes D}$, then
$
PQ \in \mathbb C_N^{\otimes D}\otimes \mathbb C_N^{*\otimes D},
\;
[PQ]^{\bm a}_{\bm b}
= P^{\bm a}_{\bm c}Q^{\bm c}_{\bm b}.$

\item \textbf{Trace.}  
For $P\in \mathbb C_N^{\otimes D}\otimes \mathbb C_N^{*\otimes D}$,
$\Tr(P)\in \Cp,\;
\Tr(P) = P^{\bm a}_{\bm a}.$

\item \textbf{Trace invariance.}  
If $U\in U(N^D)$, then $\Tr(UPU^\dagger)=
\Tr(P)$.

\item \textbf{Partial trace.}  
For $P\in \mathbb C_N^{\otimes D}\otimes \mathbb C_N^{*\otimes D}$,
$
{\rm Tr}_{(c)}(P)\in  \mathbb C_N^{\otimes (D-1)}\otimes \mathbb C_N^{*\otimes (D-1)},\; [{\rm Tr}_{(c)}(P)]^{\bm a_{\hat c}}_{\bm b_{\hat c}} = P^{\bm a_{\hat c}\bm d}_{\bm b_{\hat c}{\bm d}}.$

\item \textbf{Partial trace invariance.}  
If $U\in U(N)$ and $\check U=U\otimes {\bb 1}^{\otimes (D-1)}$, then ${\rm Tr}_{(1)}(\check UP\check U^\dagger)=
{\rm Tr}_{(1)}(P)$.

\item \textbf{From tensor product to contraction.} 
If $\phi,\psi\in \mathbb C_N^{\otimes D}$, then
$
\Tr(\phi\psi^\dagger)=\psi^\dagger\phi.$

\end{itemize}

\paragraph{Tensor and operator derivatives.}
\begin{itemize}
    \item \textbf{Tensor derivatives.} For $\phi\in \Cp_N^{\otimes D}$,
\begin{equation}
    \left[\frac{\partial}{\partial \phi}\right]_{\bm a}:=\frac{\partial}{\partial \phi^{\bm a}}\;.
\end{equation} 
    \item \textbf{Operator derivative identity.} For $\phi\in \Cp_N^{\otimes D}$, $\frac{\partial}{\partial \phi}\phi={\bb 1}$.
    \item \textbf{Operator derivatives.} For $\Phi\in \Cp_N^{\otimes D}\otimes\Cp_N^{*\otimes D}$,
\begin{equation}
    \left[\frac{\partial}{\partial \Phi}\right]^{\bm a}_{\bm b}:=\frac{\partial}{\partial \Phi^{\bm b}_{\bm a}}\;.
\end{equation} 
    \item \textbf{Operator derivative identity.} For $\Phi\in \Cp_N^{\otimes D}\otimes\Cp_N^{*\otimes D}$, $\frac{\partial}{\partial \Phi}\Phi=\Sigma$, where $\Sigma{}^{\bm a\bm c}_{\bm b \bm d}=\delta^{\bm a}_{\bm d}\delta^{\bm c}_{\bm b}$.
    
\end{itemize}
\paragraph{Trace invariants for tensors.}
In the literature 
and in our paper, 
 ``trace invariants for tensors" may be referred to as ``tensor invariants" or ``bubbles".

\begin{itemize}

\item \textbf{Multi-permutations.}  
For $n\ge 1$, a multi-permutation is a $D$-tuple $
\bm{\alpha} = (\alpha_1,\dots,\alpha_D) \in \Sy_n^{D},$
with each $\alpha_j\in \Sy_n$.  

\item \textbf{Scalar distribution.} For $\mu\in \Sy_n$ and $\bm{\alpha} \in \Sy_n^{D}$,
we define 
$
\mu\bm{\alpha} := (\mu\alpha_1,\dots,\mu\alpha_D)\quad \text{and}\quad \bm{\alpha}\mu := (\alpha_1\mu,\dots,\alpha_D\mu).$

\item \textbf{Action on index arrays.}  
Let $\bm{\alpha} 
\in \Sy_n^{D}
$, and let $\bm k = (k_{ij})$ be an $n\times D$ array of indices with  
$k_{ij}\in\{1,\dots,N\}$.  
The action of $\bm{\alpha} 
$ on $\bm k$ is defined componentwise by
$
(\bm{\alpha}_* \bm k)_{ij}
= k_{\sigma_j^{-1}(i)\, j}.$
That is, $\sigma_j$ permutes the $i$-index in the $j$-th column.

\item \textbf{Trace invariants of an operator.
}  
For $P\in\mathbb C_N^{\otimes D}\otimes\mathbb C_N^{*\otimes D}$,
\[
\Tr_{[\bm{\alpha}]}(P)
:= \sum_{\bm k}\prod_{i=1}^n
P^{(\bm{\alpha}_*\bm k)_i}_{\bm k_i},
\]
where ${\bm k_i}$ 
and 
$(\bm{\alpha}_*\bm k)_i$
respectively
denote
the $i$-th row of ${\bm k}$
and
the $i$-th row of
$(\bm{\alpha}_*\bm k)$.

\item \textbf{Permutation conjugation invariance.}  
For any $\mu\in \Sy_n$,
$
\Tr_{[\mu\bm{\alpha}\mu^{-1}]}(P)
= \Tr_{[\bm{\alpha}]}(P).
$

\item \textbf{Multi-unitary conjugation invariance.} For any $U\in U(N)^D$, 
$
\Tr_{[\bm{\alpha}]}(UPU^{-1})
= \Tr_{[\bm{\alpha}]}(P).
$

\item \textbf{Trace invariants of a tensor pair.
} 
For $\phi,\psi\in\mathbb C_N^{\otimes D}$,
\[
\Tr_{[\bm{\alpha}]}(\psi^\dagger,\phi)
:=\sum_{\bm k}\prod_{j=1}^n\psi^\dagger_{\bm k_j}
\phi^{(\bm{\alpha}_*\bm k)_j}= \Tr_{[\bm{\alpha}]}(\phi\psi^\dagger)
.
\]

\item \textbf{Left and right invariance.}  
For any $\mu,\nu\in \Sy_n$,
$
\Tr_{[\mu\bm{\alpha}\nu]}(\phi\psi^\dagger)
= \Tr_{[\bm{\alpha}]}(\phi\psi^\dagger)
.
$

\item \textbf{Gauge choice.} For $c\in\{1,...,D\}$, 
$\Tr_{[\bm{\alpha}]}(\phi\phi^\dagger)=\Tr_{[\alpha_c^{-1}\bm{\alpha}_{\hat c}]}(\Tr_{(c)}(\phi\phi^\dagger))$
\end{itemize}
\paragraph{Integration measures.}

\begin{itemize}

\item \textbf{Complex numbers.}
$
dz^*dz = dadb$ with $z=a+ib.$

\item \textbf{Complex matrices.}
$
\mathcal DM^\dagger\mathcal DM
=\prod_{ij} dM_{ij}.$

\item \textbf{Self-adjoint matrices.}
$
\mathcal DA
=\prod_{i\leq j} dA_{ij}.$

\item \textbf{Complex tensors.}
$
\mathcal D\phi^\dagger\mathcal D\phi
=\prod_{\bm a} d\phi^{\bm a}.$

\item \textbf{Self-adjoint tensors.}
$
\mathcal D\Phi
= \prod_{\bm a\le\bm b} d\Phi^{\bm a}_{\bm b},$
using lexicographic order.

\end{itemize}

\section{Motivation}\label{sec:causalmm}
\noindent In this section, we review three types of random matrix models: dually weighted, self--adjoint multi--matrix,
and complex matrix models.

\subsection{Dually weighted matrix model}\label{subsec:DWMM}
\noindent Dually weighted matrix models 
(DWMM)
extend the self-adjoint one--matrix models by allowing independent control of the combinatorial weights of vertices and faces in the ribbon graph expansion
amplitudes.
This is implemented using two fixed ``background'' matrices
which control the weights.
In this subsection, we give a brief summary of this.

\paragraph{External matrices and couplings.}

Let $X,Y \in M_N(\mathbb{C})$ be two fixed matrices.  
They generate the 
sequences
\begin{equation}\label{eq:DWMM_COUPLINGS}
    x^{(k)} = \Tr(X^k)\;, \qquad
    y^{(k)} = \Tr(Y^k)\;, \qquad k \ge 1\;.
\end{equation}  
No specific assumptions on $X$ or $Y$ are required beyond the existence of these traces. For practical purposes one may take them diagonalizable, but their algebraic role in the model depends solely on the above traces.

\paragraph{Dually weighted potential.}
For a self-adjoint matrix $A\in{\rm H}(N)$,
specify
\begin{equation}\label{eq:DW_potential_log}
    V(A)
    = -\Tr\ln \left(\mathbb{1}^{\otimes 2} - i\,X\otimes(YA)\right)\; .
\end{equation}
Expanding the logarithm results in
\begin{equation}\label{eq:DW_potential_series}
    V(A)
    = \sum_{k=1}^{\infty} \frac{i^{k}}{k}\,
      \Tr(X^{k})\Tr \left((YA)^{k}\right)\; .
\end{equation}
\paragraph{Partition function.}

The 
one-matrix
DWMM is defined by
\begin{equation}\label{eq:Z_DY}
\mathcal{Z}_{\rm DWMM}
    [V]
    = \int_{{\rm H}(N)} \mathcal{D}A\;
      e^{-\frac{N}{2}\Tr(A^2) -\Tr\ln (\mathbb{1}^{\otimes 2} - i\,X\otimes(YA))}\; .
      \end{equation}
The ribbon graph expansion of this partition function describes discretized surfaces in which both vertex degrees and face lengths can be adjusted independently. This comes from looking at the role of the coupling sequences \eqref{eq:DWMM_COUPLINGS}\footnote{The original DWMM literature \cite{Kazakov_1996,das_1990} works with the series \eqref{eq:DW_potential_series} (without the $i$). The factor of $i$ in \eqref{eq:DW_potential_log} is introduced here only because it arises naturally in the Fourier/intermediate-field representations used later; it does not change the combinatorics.} in the potential \eqref{eq:DW_potential_series}: the factor $\Tr(X^{k})$ provides a weight for $k$--valent vertices, and the terms $\Tr((YA)^{k})$ produce a cycles of length $k$, so that each face of length $k$ is weighted by $\Tr(Y^{k})$. Thus, in the ribbon graph expansion of this matrix model:
\begin{itemize}
    \item a vertex of degree $k$ contributes $\Tr(X^{k})=x^{(k)}$,  
    \item a face of length $k$ contributes $\Tr(Y^{k})=y^{(k)}$\,,
\end{itemize}
to the amplitudes.
Therefore, the matrices $X$ and $Y$, independently, determine the vertex and face
weights
of the model. This flexibility makes DWMM suitable for encoding local geometric constraints, like the causal matrix model presented in the next subsection \ref{subsec:CDT}.

\subsection{Causal matrix model}
\label{subsec:CDT}
\noindent The causal matrix model~\cite{Benedetti:2008hc} provides a matrix integral representation of two--dimensional CDT by translating global causal rules into local constraints on ribbon graphs.

\paragraph{Two--matrix formulation.}
Let $A,B\in{\rm H}(N)$.  
The causal matrix model is defined by the partition function
\begin{equation}\label{cdtmm}
    \mathcal{Z}_{\rm CDT}(g)
    = \int dA\,dB\;
      e^{-N\Tr \left[
          \frac{1}{2}A^2
        + \frac{1}{2}(C_2^{-1}B)^2
        + g\,A^2 B
      \right]}\; ,
\end{equation}
where $C_2$ is the special case of the matrices $C_k$ satisfying
\begin{equation}\label{cnprop}
    \Tr(C_k^p) = N\,\delta_{p,k}\;, \qquad p=1,\dots,N\; .
\end{equation}

\noindent In the Feynman expansion of this model, the propagator of $A$ \footnote{This is standard physics jargon for $\langle A_{ij}A_{kl}\rangle(g=0)$, in this case. Along this work, we use this term when refering to Feynman expansions.} corresponds to spacelike ribbon-graph edges, while that of $B$ (involving $C_2$) corresponds to timelike ribbon-graph edges.
In this model, properties of dually-weighted matrix models are incorporated through the matrix $C_2$ (see \cite{Benedetti:2008hc} for more details).
Ribbon-graph faces containing $B$--lines acquire factors $\Tr(C_2^m)$, hence, by \eqref{cnprop}, contain exactly two $B$--lines, reproducing the CDT rigidity condition: each triangle has only two timelike ribbon-graph edges, and their gluing imposes that no spatial topology change is allowed. In the dual triangulation, these rules imply that each dual face contains either two or zero timelike ribbon-graph edges. These local rules enforce the global foliation structure.
\begin{figure}[h]
\centering
    \begin{subfigure}[H]{0.3\textwidth}
    \centering
    \includegraphics[width=\linewidth]{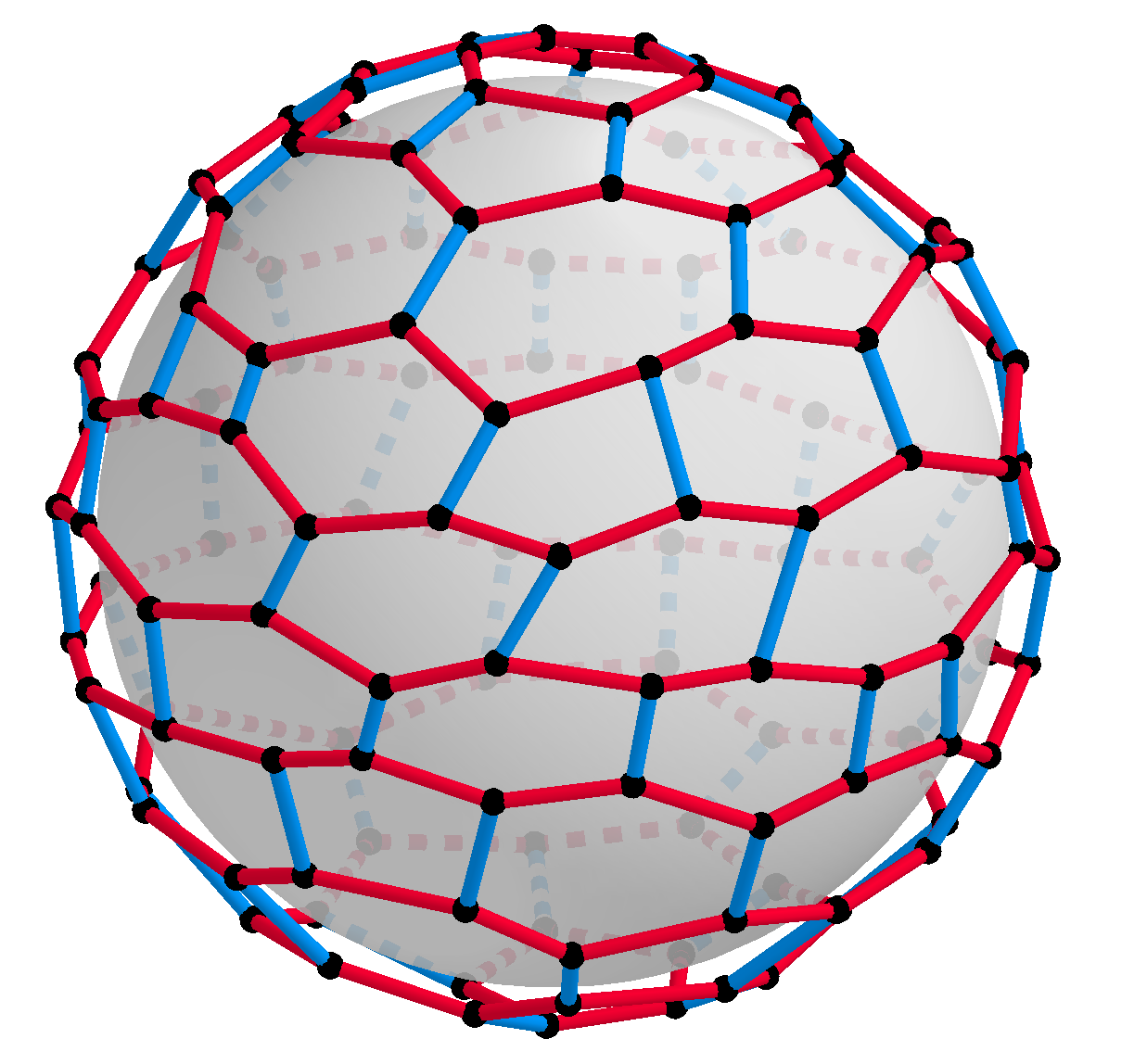}
    \caption{}
    \label{causrib}
    \end{subfigure}
    \;\;
    \begin{subfigure}[H]{0.3\textwidth}
    \centering
    \includegraphics[width=\linewidth]{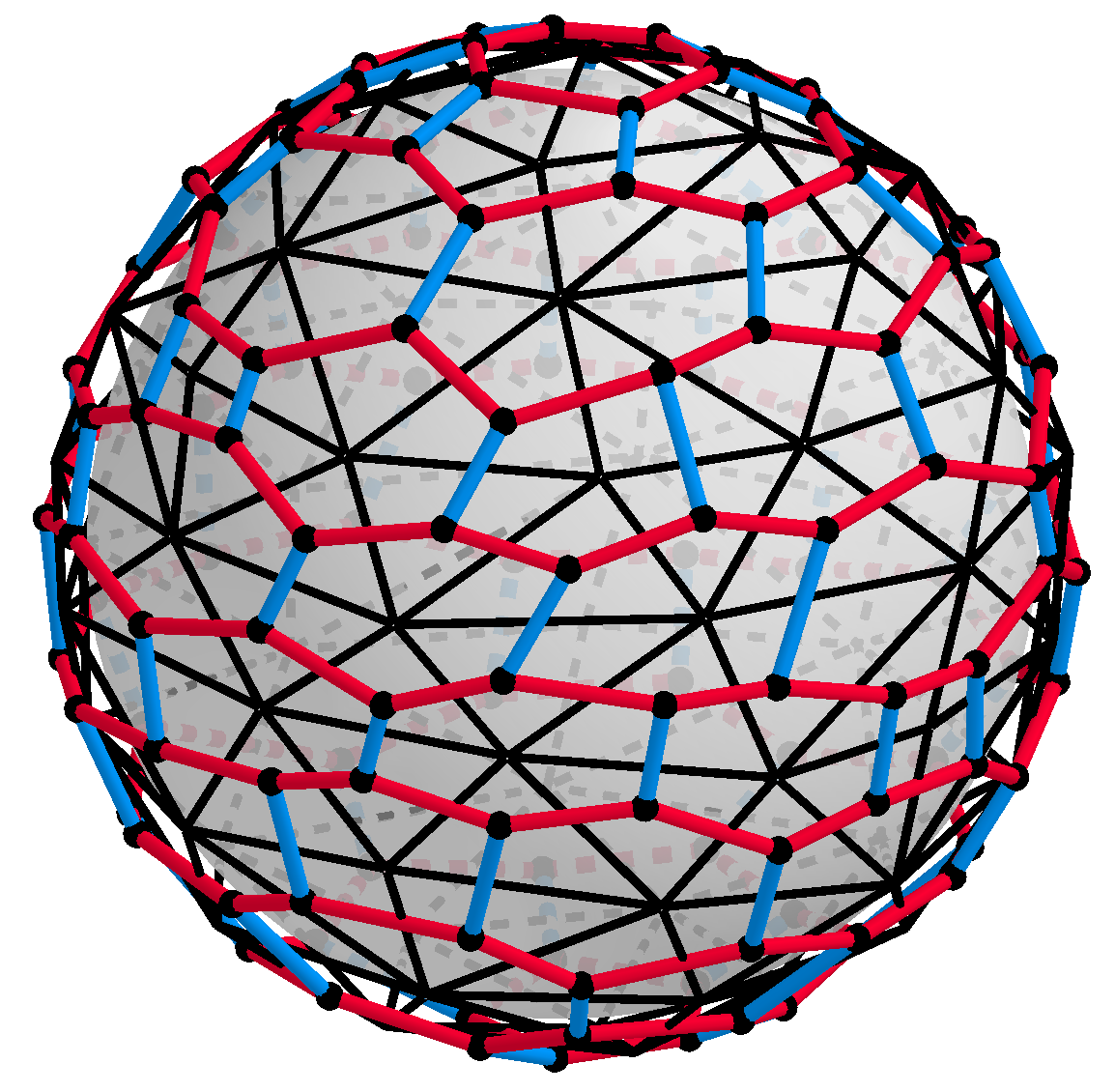}
    \caption{}
    \label{causribtocdtgr}
    \end{subfigure}
    \;\;
    \begin{subfigure}[H]{0.3\textwidth}
    \centering
    \includegraphics[width=\linewidth]{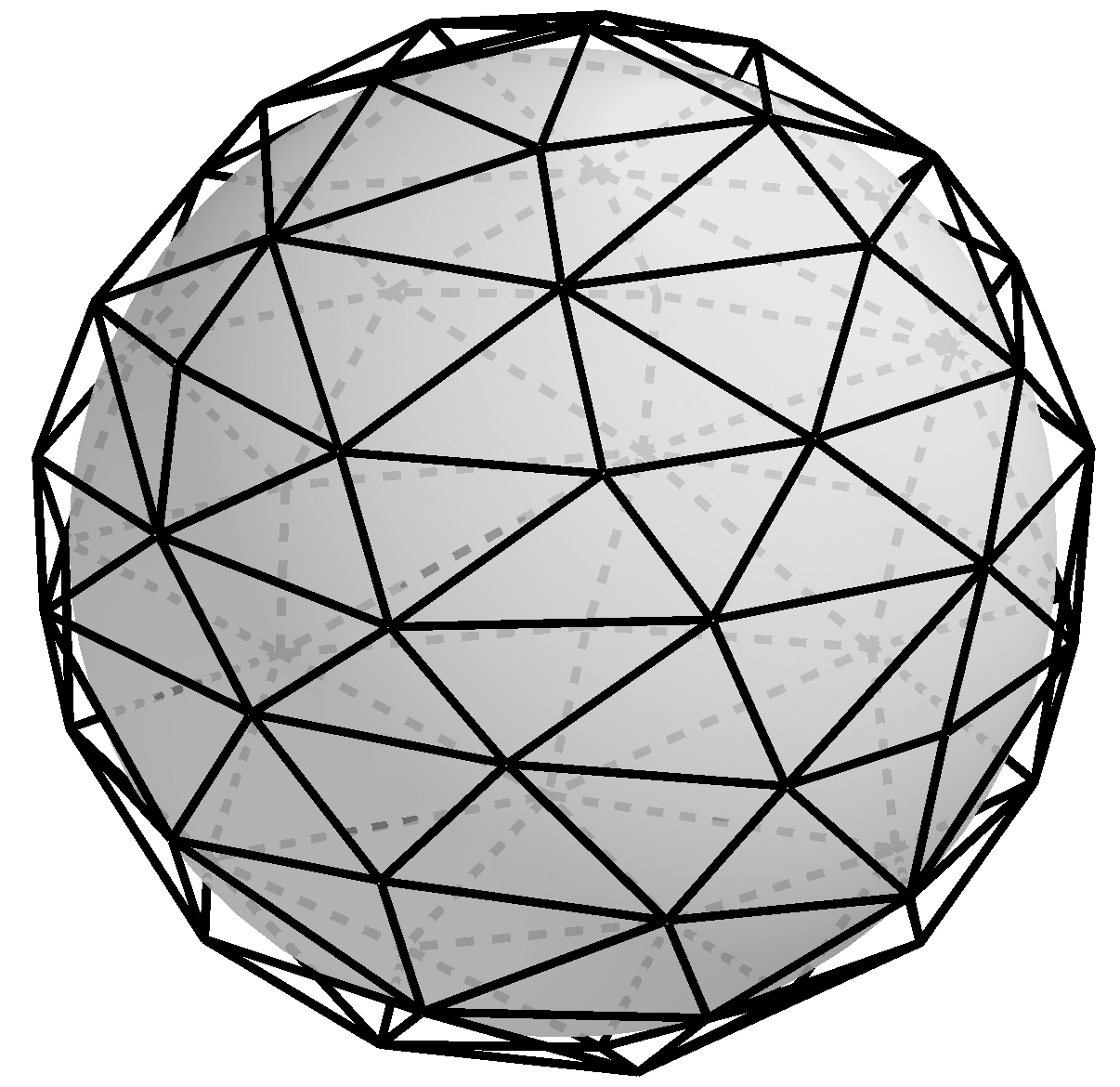}
    \caption{}
    \label{cdtgr}
    \end{subfigure}
    \caption{(a) 
    A ribbon graph generated by the causal matrix model \eqref{cdtmm}. The spacelike 
     ribbon-graph
    edges are drawn in red and the timelike
    ribbon-graph
    edges in blue.
    (b) Its dual graph corresponds to a CDT triangulation.
    (c) A CDT triangulation. The parallel horizontal lines correspond to the foliation structure.}
    \label{cdtgraphs}
\end{figure}

\paragraph{Integrating out the timelike matrix.}
Since the $B$--integral in \eqref{cdtmm} is Gaussian, it can be evaluated explicitly. This results on the self-adjoint one--matrix model
\begin{equation}\label{eq:cdmm_onematrix}
    \mathcal{Z}_{\rm CDT}(g)
    = \kappa_N \det(C_2)^N
      \int dA\;
      e^{-N\Tr \left[
          \frac{1}{2}A^2
        - \frac{g^2}{2}(C_2 A^2)^2
      \right]}\; ,
\end{equation}
with 
\[
\kappa_N
= \int dB\; e^{-\frac{N}{2}\Tr B^2}
= N^{-\frac{N^2}{2}} 2^{\frac{N}{2}} \pi^{\frac{N^2}{2}} \; .
\]
The quartic interaction $(C_2A^2)^2$ reflects the CDT rigidity encoded in the original propagator of $B$. 
\medskip

Even though \eqref{eq:cdmm_onematrix} successfully reproduces a sum over causal triangulations, this success has not translated into analytic control, since the partition function is not exactly computable with existing techniques. Motivated by FRG studies \cite{Eichhorn:2013isa,Castro:2020dzt}, we propose a change of variables that moves the constant matrix $C_2$ from the quartic to the quadratic part of the action in \eqref{eq:cdmm_onematrix}. Placing the rigidity matrix $C_2$ in the propagator is particularly natural, as it modifies the quadratic part of the action and is analogous to introducing a non-trivial dispersion relation in field theory. Somewhat unexpectedly, this change of variables forces us from working with self-adjoint to complex matrix ensembles, since the change of variables is not valid in the self-adjoint model, but is valid in the complex model. Moreover, it leads to the observation that the partition functions of these complex models coincide with those of certain random self-adjoint matrix and tensor models, revealing an equivalence between complex and self-adjoint ensembles.
\subsection{Motivating the equivalence with a complex matrix model}\label{subsec:motiv_glnmm}
\noindent Motivated by making a change of variables in the exponential of \eqref{eq:cdmm_onematrix} that takes the rigidity matrix $C_2$ from the quartic to the quadratic part of the action, we consider the complex matrix model
\begin{equation}\label{bc}
    \mathcal{Z}(g)
    = \det(C_2)^N
      \int dM^\dagger dM\;
      e^{-N\Tr \left[
        M^\dagger M
        +g(C_2 M M^\dagger)^2
      \right]}\; .
\end{equation}
Although the quartic term in this action resembles that of~\eqref{eq:cdmm_onematrix}, this model is \emph{not} equivalent to the causal matrix model, its ribbon graphs are far more restricted. This can be seen by making a change of variables in \eqref{bc} to bring the external matrix $C_2$ to the quadratic part of the action (See Appendix~\ref{change}). This results in
\begin{equation}\label{model1}
    \mathcal{Z}(g)
    = \int dM^\dagger dM\;
      e^{-N\Tr \left[
          M^\dagger C_2^{-1} M
        - \frac{g^2}{2}(M^\dagger M)^2
      \right]}\; .
\end{equation}
The Feynman rules for this model are:
\begin{itemize}
    \item the interaction $\Tr[(M^\dagger M)^2]$ is a single 4--valent vertex;
    \item the propagator is modified by $C_2^{-1}$;
    \item the propagator only allows edges of the form $M^\dagger$--$M$.
\end{itemize}

\paragraph{Restricted ribbon graphs and their duals.}
With only one vertex type present, the generated graphs form a necklace-lik sequence of loops of identical quartic vertices (See Fig. \ref{fig:korder}).  These diagrams lack the $AAB$-vertices and face constraints necessary to realize the CDT structure. Thus, \eqref{model1} does \emph{not} reproduce CDT.
\begin{figure}[H]
    \centering
    \includegraphics[width=0.60\linewidth]{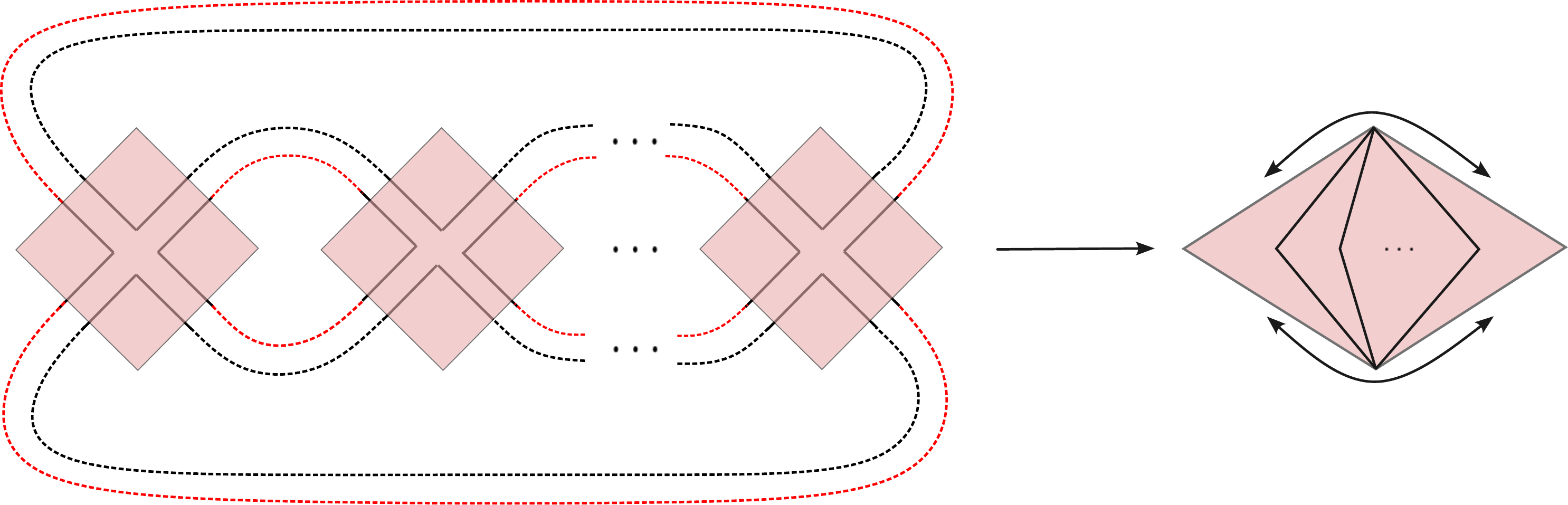}
    \caption{Ribbon graph and dual quadrangulation associated with the complex model~\eqref{model1}.  
    Its chain-like structure shows why the model does not reproduce CDT.}
    \label{fig:korder}
\end{figure}
\medskip
\noindent These observations motivate a general question:
\emph{What is the precise relationship between complex matrix models with general quadratic/kinetic term and self-adjoint dually-weighted multi--matrix models?} This paper provides the full answer. We show that a complex matrix model with covariance $(P,Q)$\footnote{This term is used to indicate that the pair of matrices $(P,Q)$ appears in the quadratic part of the action. It is not a formal statement about the probabilistic nature of these models.}  is \emph{exactly equivalent} to a self-adjoint two--matrix model via an intermediate field representation. In this, the intermediate field carries a logarithmic, dually weighted potential (characteristic of DWMM and the causal matrix model). 
Furthermore, the complex matrix model \eqref{model1} is simply a special case of the general mechanism proved in Section~\ref{sec:equivmm}.

\section{Summary of main results}
\label{sec:summarymain}

\noindent
In this section, we survey the main results of this paper.

\paragraph{Matrix models.}
Let $P,Q$ be fixed $N\times N$ complex matrices. Our complex random matrix model is defined by the partition function
\begin{equation}
\mathcal{Z}_{\rm CM}[V](P,Q)
  = \iint_{\mathbb C^{N\times N}} {\cal D}M^\dagger\,{\cal D}M\;
    e^{-N\Tr(P^{-1} M^\dagger Q^{-1} M) + V(M^\dagger M)} ,
\end{equation}
where $V$ is an arbitrary potential of  the self–adjoint matrix $M^\dagger M$. On the self–adjoint side, we consider a two–matrix model on ${\rm H}(N)\times{\rm H}(N)$ with partition function
\begin{equation}
\mathcal{Z}_{\rm HM}[V,
{Y}_{\ln}
]
  = \iint_{{\rm H}(N)^2} {\cal D}A\,{\cal D}B\;
      e^{- iN\Tr(AB) + 
{Y}_{\ln}
      [P,Q](B)+ V(A) },
\hspace{1mm}
{Y}_{\ln}
      [P,Q](B)
    = - \Tr\ln \left(\mathbb 1^{\otimes 2}
      - i\,Q\otimes(PB)\right).
\end{equation}
The {\it intermediate effective field} $A$ carries the effective potential $V(A)$, while the {\it dually-weighted intermediate field} $B$ carries a dually-weighted logarithmic potential $Y_{\ln}[P,Q](B)$.
\medskip

\noindent
Then, Theorem~\ref{matrixequiv_thm} implies that: for any potential $V$,
\begin{equation}
\label{ZCMconc}
\frac{\mathcal Z_{\rm CM}[V](P,Q)}{\mathcal Z_{\rm CM}[0](P,Q)}
   = 
\frac{\mathcal Z_{\rm HM}[V,
{Y}_{\ln}
[P,Q]]}{\mathcal Z_{\rm HM}[0,
{Y}_{\ln}
[P,Q]]}.
\end{equation}
In general, this equality holds at the level of formal power series and its convergence depends on the form of the potential $V$\footnote{The convergence criteria of $Z_{\rm CM}$ in terms of the spectra of $P$ and $Q$ is studied in Appendix~\ref{convcrit}.}.

Equivalently, any observable depending only on the self–adjoint combination $M^\dagger M$ can be computed in either model.
Corollary~\ref{cor:exp_val_mm} states that, for any potential $V$ and any function $f:{\rm H}(N)\to\mathbb C$,
\begin{equation}
\left\langle f(M^\dagger M)\right\rangle[V](P,Q)
    = \left\langle f(A)\right\rangle[V,
{Y}_{\ln}
    [P,Q]].
\end{equation}

So far, the equivalence has been stated algebraically. A complementary 
combinatorial
perspective on this equivalence 
follows
from Propositions~\ref{avM} and~\ref{avAB} (later contained in Corollary~\ref{cor:exp_val_mm}), which state that all multi-trace observables match in the two models.
In the complex matrix model, multi-trace observables correspond to ribbon graphs with two complementary families of faces, weighted by the matrices $P$ and $Q$, respectively (Proposition~\ref{avM}). In contrast, Proposition~\ref{avAB} shows that in the self-adjoint formulation this structure is reorganized into a dually weighted model, in which faces carry weights associated with $P$, while vertices are weighted by $Q$.
The equivalence can thus be viewed 
as a reorganization of the same combinatorial data, trading one set of face weights for vertex weights through the introduction of the intermediate field.

\paragraph{Tensor models.}
\noindent The equivalence between complex and self-adjoint theories persists also in tensor models. Let $\phi\in\Cp_N^{\otimes D}$ be a random complex order–$D$ tensor with covariance
\(
R\in\Cp_N^{\otimes D}\otimes\Cp_N^{*\otimes D}
\), and $V(\phi\phi^\dagger)$,  {\emph{a scalar-$U(1)$ invariant}} potential,
as elaborated in Subsection \ref{subsec:complex-tensor}.
As in \eqref{CTVscalar}, the partition function of the complex theory in question is 
\begin{equation}\label{eq:CTpf}
    {\cal Z}_{\rm CT}[V](R)
    =\iint {\cal D}\phi^\dagger {\cal D}\phi\;
      e^{-\phi^\dagger R^{-1}\phi+V(\phi{\phi^\dagger})}\;.
\end{equation}
\noindent
We consider random tensors $\Phi,\Psi\in\Hm(\Cp_N^{\otimes D})$, both serving as intermediate fields, distributed according to the partition function of the self-adjoint theory \eqref{HCT} with potential \eqref{eq:logpot}
\begin{equation}\label{eq:HTpf}
    {\cal Z}_{\rm HT}[V,
{Y}_{\ln}
    [R]]
      =\iint {\cal D}\Phi {\cal D}\Psi\;
        e^{-i \Tr(\Phi\Psi)+
{Y}_{\ln}
        [R](\Psi)+V(\Phi)}\,,
\qquad
{Y}_{\ln}
    [R](\Psi)=-\Tr\ln\lp {\bb 1}^{\otimes D}-i R\Psi\rp \,,
\end{equation}
where the potential $Y[R](\Psi)$ gives $\Psi$ the role of the dually-weighted intermediate field.
\noindent
Our Theorem~\ref{mainthm} states that, for any potential function $V$, 
\begin{equation}
    \frac{\mathcal{Z}_{\rm CT}[V](R)}{\mathcal{Z}_{\rm CT}[0](R)}
    =\frac{\mathcal{Z}_{\rm HT}[V,Y_{\ln}[R]]}{\mathcal{Z}_{\rm HT}[0,Y_{\ln}[R]]}\;.
\end{equation}
Just as in the matrix model case, this should be understood as an equality between the formal power series, since its convergence depends on the potential $V$.
Equivalently, as stated in Corollary~\ref{cor:exp_val_tm},  for any potential $V$ and any function $f:{\rm H}(\Cp_N^{\otimes D})\to\mathbb C$, 
\begin{equation}
\left\langle f(\phi\phi^\dagger)\right\rangle[V](R)
    = \left\langle f(\Phi)\right\rangle[V,
{Y}_{\ln}
    [R]].
\end{equation}
\vskip 5pt
Under an additional symmetry, i.e., {\emph{the partial-$U(N)$ invariance}} explained in Subsection \ref{subsec:complex-tensor}, the self–adjoint tensor model can be reduced to a model with tensors of order $2(D-1)$ \eqref{eq:partialUN_self-adjoint}, as
shown
in  Theorem~\ref{theorem:partialtensor}. 
By working with the partial trace ${\rm Tr}_{(1)}(\phi\phi^\dagger)$, this takes the form
\begin{equation}
    \frac{\mathcal{Z}_{\rm CT}[\hat V\circ{\rm Tr}_{(1)}](R)}{\mathcal{Z}_{\rm CT}[0](R)}=\frac{\mathcal{Z}_{\rm \hat HT}[\hat V,
\hat{Y}_{\ln}
    [R]]}{\mathcal{Z}_{\rm \hat HT}[0,
\hat{Y}_{\ln}
    [R]]}\;.
\end{equation}

\medskip 

\noindent 
We remark that when $D=2$ and the propagator factorizes as $R=Q\otimes P^{\tp}$, the tensor equivalence of Theorem~\ref{theorem:partialtensor} has the matrix equivalence in Theorem~\ref{matrixequiv_thm} as a direct corollary.

\medskip
In summary, our results presented in Theorems \ref{matrixequiv_thm}, \ref{mainthm}, and \ref{theorem:partialtensor} show that complex matrix
and tensor
models with nontrivial covariance and interactions depending on $M^\dagger M$ or $\phi\otimes\phi^\dagger$, respectively,
admit exact self–adjoint representations in which the covariance data is carried by a dually-weighted intermediate field 
$B$ or $\Psi$, respectively.

\noindent This equivalence provides both a conceptual reinterpretation of such models and a practical computational tool: any expectation values/observable built from the self–adjoint combination may be computed in whichever formulation (complex or self–adjoint).
In particular, self-adjoint models have lower order potential, so they tend to provide easier computation. In certain simple cases, the complex model in question may even correspond to a Gaussian self-adjoint models.
In this work, we give an example of this simplification in the case of the CDT-like matrix model in \eqref{eq:quadmatrix} and in Appendix~\ref{quaex}.

\paragraph{Relation to existing intermediate field representations.}

\noindent The identities established in this paper put a broad family of tensor and matrix intermediate field representations into a single unified framework. Several constructions in the literature appear as special cases of our main formula (Theorem \ref{mainthm}),
\begin{equation}
\label{eq:ifrep_dwmm}
        \mathcal Z_{\rm CT}[V](R)
        \!
    = \!\! \iint {\cal D}\phi^\dagger{\cal D}\phi\, 
      e^{-\phi^\dagger R^{-1}\phi + V(\phi\phi^\dagger)}
    = {\cal N} \!\! \iint {\cal D}\Phi\,{\cal D}\Psi\,
      e^{-i\Tr(\Psi\Phi)
        - \Tr\ln({\bb 1}- i R\Psi) + V(\Phi)}={\cal N}\mathcal Z_{\rm HT}[V](R),
\end{equation}
where ${\cal N}=(\mathcal Z_{\rm HT}[0,Y[R]])^{-1}(\mathcal Z_{\rm CT}[0](R))$, 
and, with an additional partial symmetry
(Theorem~\ref{theorem:partialtensor}),
\begin{equation}
\label{eq:lowerorderresult}
        \mathcal Z_{\rm CT}[V](R)
        \!
    = \!\! \int\!\!\!\!\!\int \!\! {\cal D}\phi^\dagger{\cal D}\phi\, 
      e^{-\phi^\dagger R^{-1}\phi + V(\phi\phi^\dagger)}
    = \hat{\cal N} \!\!\! \int\!\!\!\!\!\int \!\!{\cal D}\hat\Phi{\cal D}\hat\Psi\,
      e^{-i\Tr(\hat\Psi\hat\Phi)
        - \Tr\ln({\bb 1}- i R({\bb 1}\otimes \hat\Psi)) + \hat V(\hat\Phi)}=\hat{\cal N}\mathcal Z_{\rm \hat HT}[\hat V](R),
\end{equation}
for when there is a $\hat V$ such that $V(\phi\phi^\dagger)=\hat V(\Tr_{(1)}(\phi\phi^\dagger))$.
Previously known intermediate field representations arise from suitable specializations of this identity \eqref{eq:ifrep_dwmm}.

\begin{itemize}
    \item 
The tensor intermediate field representation of Bonzom--Lionni--Rivasseau.

In~\cite{Bonzom2015ColoredTO}, the intermediate field formula is developed 
for a \emph{single} trace-invariant tensor interaction $\Tr_B(T,\bar T)$ with
unit covariance and is proven as a formal series in the coupling constant $\lambda$. 
In our notation, their intermediate field formula
corresponds to choosing $R={\bb 1}$ and repeatedly introducing a new field $\Psi_i$ for each pair of the $k$ pairs of $\phi$ with $\phi^\dagger$ in a 2k-vertex bubble $B$, so $V(\Phi)=\frac{\lambda}{n} N^{s}\,V_{B,\Omega}(\{\Phi_i\})$ (Eq. (52) in \cite{Bonzom2015ColoredTO}). Making these substitutions in \eqref{eq:ifrep_dwmm}, we obtain the following representation (Eq. (53) in \cite{Bonzom2015ColoredTO})
\[
    Z_B(\lambda,N)
      = \iint \prod_{i=1}^k{\cal D}\Phi_i\,{\cal D}\Psi_i\;
          e^{-i\Tr(\Psi\Phi) 
            - \Tr\ln({\bb 1}^{\otimes D}-i\Psi)
            + \frac{\lambda}{n}N^{s} V_{B,\Omega}(\Phi)} .
\]
A transformation similar to a Wick rotation, turning the self-adjoint tensors $(\Psi,\Phi)$ into a complex adjoint pair $(M, M^\dagger)$ yields the formulation that appears in~\cite{Bonzom2015ColoredTO}. However, a conceptual difference is that our equality holds exactly, while the derivation in~\cite{Bonzom2015ColoredTO} is performed at the level of formal power series.

\item
Hubbard--Stratonovich representation.

The classical intermediate field representation for the quartic one-matrix model,\;see e.g.~\cite{gurau2014analyticityresultscumulantsrandom}, follows from choosing $D=2$,
$R=\sqrt{\frac{\lambda}{N}}{\bb 1^{\otimes 2}}$, 
and $\hat V(\hat\Phi)= -\frac12\Tr(\hat\Phi^2)$.
Our identity \eqref{eq:lowerorderresult}, 
after integrating out $\Phi$,
yields
\[
  Z = \int {\cal D}\hat\Psi\;
         e^{
            -\frac12\Tr(\hat\Psi^2)- N\,\Tr\ln({\bb 1}- i\sqrt{\frac{\lambda}{N}}\hat\Psi)},
\]
which is the well-known one-matrix formula. Thus the quartic Hubbard--Stratonovich representation is also contained as a direct corollary of our results.
\end{itemize}

\medskip

\noindent In summary, the various intermediate field constructions developed in the matrix and tensor literature arise as specific choices of the covariance $R$ and the potential $V$ in the universal identity
\eqref{eq:ifrep_dwmm} that we newly discovered in this paper. 
This provides a unified analytic framework that extends these previous results to completely general potentials and covariances.

\section{Equivalence in random matrix models}\label{sec:equivmm}
\noindent This section establishes an exact correspondence between two seemingly different matrix models: (i) a complex matrix model with covariance $(P,Q)$ and potential $V(M^\dagger M)$, and (ii) a self-adjoint two-matrix model whose kinetic term is the mixed bilinear $\Tr(AB)$ and whose {\it dually-weighted} intermediate $B$-field carries a logarithmic potential.
The equivalence between the partition functions of these models, as well as all observables depending only on $M^\dagger M$ (or equivalently on the intermediate effective field $A$), holds for arbitrary potentials $V$. In general, this equivalence is to be understood at the level of formal power series, while for suitable choices of $V$ it also holds at the level of convergent integrals.
In summary, this result shows that a complex matrix integral with covariance of a general form can be reformulated as a self-adjoint two-matrix model whose intermediate field $B$ carries a dually weighted logarithmic interaction.

\subsection{The complex matrix model}
\label{subsec:GLmodel}
\noindent Let the partition function\footnote{Convergence criteria for~\eqref{eq:ZCM} in terms of the spectra of $P$ and $Q$ are given in Appendix~\ref{convcrit}.}
\begin{equation}
\label{eq:ZCM}
\mathcal{Z}_{\rm CM}[V](P,Q)
  = \iint_{\mathbb C^{N\times N}} {\cal D}M^\dagger\,{\cal D}M\;
    e^{-N\Tr(P^{-1} M^\dagger Q^{-1} M) + V(M^\dagger M)} ,
\end{equation}
define a complex matrix model with arbitrary potential $V$,
usually taken analytic in the
self-adjoint matrix $M^\dagger M$. 
Expectation values in this model
can be computed
 by
\[
\langle f(M^\dagger M)\rangle[V](P,Q)
  = \frac{
      \iint {\cal D}M^\dagger{\cal D}M\;
      f(M^\dagger M)\;e^{-N\Tr(P^{-1}M^\dagger Q^{-1}M)+V(M^\dagger M)}
    }{
      \iint {\cal D}M^\dagger{\cal D}M\;
      e^{-N\Tr(P^{-1}M^\dagger Q^{-1}M)+V(M^\dagger M)}
    }\;.
\]

\paragraph{Propagator.}
In the Gaussian theory ($V=0$), Feynamn rules only  allow pairwise Wick contractions of $M^\dagger$ with $M$. Analytically, this takes the form
\begin{equation}\label{eq:Mprop}
    \langle M^\dagger_{ij} M_{kl}\rangle[0](P,Q)
      = N^{-1} P_{il}\,Q_{kj}\; ,\qquad
    \langle M_{ij} M_{kl}\rangle[0](P,Q)=0\; ,\qquad
    \langle M^\dagger_{ij} M^\dagger_{kl}\rangle[0](P,Q)=0\; .
\end{equation}
Thus a contraction `carries' the matrix $Q$ along the first index and $P$ along the second one, giving rise to two complementary families of faces in the ribbon-graph expansion.
\begin{figure}[H]
    \centering
    \includegraphics[width=0.3
    \linewidth]{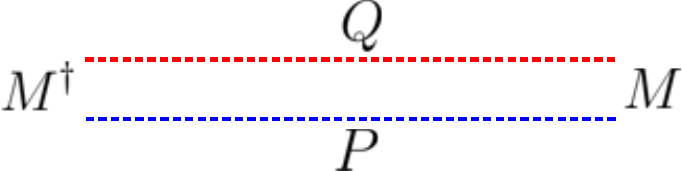}
    \caption{Ribbon graph representation of the propagator of the complex matrix model defined by the partition function \eqref{eq:ZCM} and given in \eqref{eq:Mprop}. The blue and red strands respectively represent the presence of $P$ and $Q$.
    }
    \label{fig:mprop}
\end{figure}

\begin{figure}[H]
\begin{subfigure}[c]{0.33\textwidth}
\centering
\begin{minipage}[c][4cm][c]{\linewidth}
    \centering
    \includegraphics[width=0.5\linewidth]{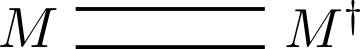}
\end{minipage}
\caption{$\Tr(M^\dagger M)$ vertex.}
\label{fig:mvert2}
\end{subfigure}
\begin{subfigure}[c]{0.33\textwidth}
\centering
\begin{minipage}[c][4cm][c]{\linewidth}
    \centering
    \includegraphics[width=0.5\linewidth]{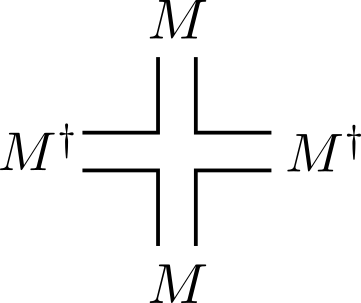}
\end{minipage}
\caption{$\Tr(M^\dagger M)^2$ vertex.}
\label{fig:mvert4}
\end{subfigure}
\begin{subfigure}[c]{0.33\textwidth}
\centering
\begin{minipage}[c][4cm][c]{\linewidth}
    \centering
    \includegraphics[width=0.5\linewidth]{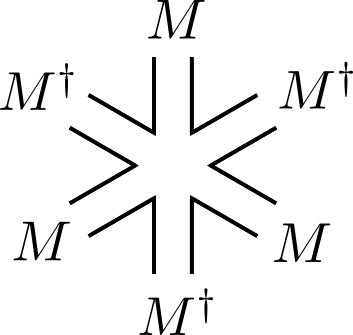}
\end{minipage}
\caption{$\Tr(M^\dagger M)^3$ vertex.}
\label{fig:mvert6}
\end{subfigure}

\caption{Examples of interaction vertices that can be considered in a theory of the form given in \eqref{eq:ZCM} via ribbon graph representation.}
\label{fig:mvert}
\end{figure}

\begin{figure}[H]
    \begin{subfigure}[c]{0.33\textwidth}
    \centering
    \begin{minipage}[c][4cm][c]{\linewidth}
    \centering
    \includegraphics[width=
    0.3\linewidth]{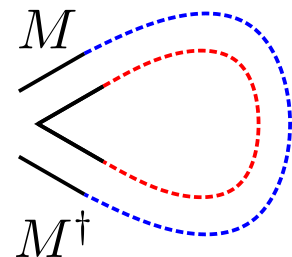}
    \end{minipage}
    \caption{Length-1 face.}
    \label{fig:mface1}
    \end{subfigure}
    \begin{subfigure}[c]{0.33\textwidth}
    \centering
    \begin{minipage}[c][4cm][c]{\linewidth}
    \centering
    \includegraphics[width=
    0.5\linewidth]{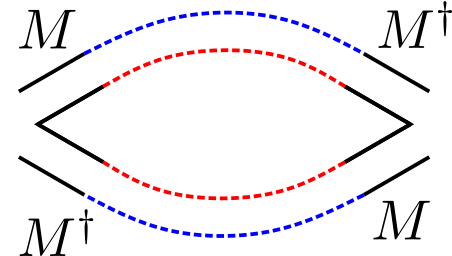}
    \end{minipage}
    \caption{Length-2 face.}
    \label{fig:mface2}
    \end{subfigure}
    \begin{subfigure}[c]{0.33\textwidth}
    \centering
    \begin{minipage}[c][4cm][c]{\linewidth}
    \centering
    \includegraphics[width=
    0.5\linewidth]{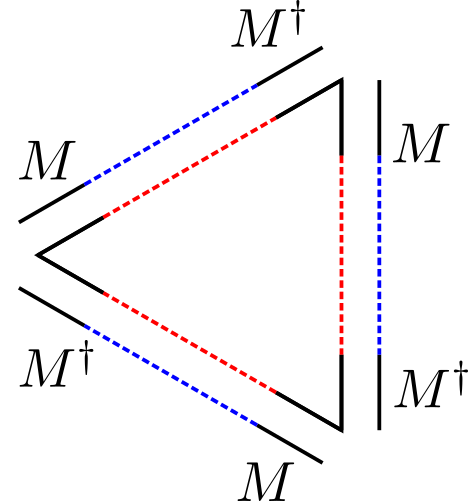}
    \end{minipage}
    \caption{Length-3 face.}
    \label{fig:mface3}
    \end{subfigure}
    \caption{Examples of faces that can be generated in \eqref{eq:ZCM} when vertices (see Fig. \ref{fig:mvert}) are connected by propagators (see Fig. \ref{fig:mprop}).}
    \label{fig:mface}
\end{figure}

\paragraph{Multi-trace invariants.}
A natural next step is to examine observables built from multi-traces, since any class-function potential can be expanded in that basis and the Feynman rules become fully transparent in terms of these invariants. The following result gives the Gaussian expectation of all multi-traces.
 \begin{proposition}\label{avM}
Let $M\in\mathbb C^{N\times N}$ be distributed according to $\mathcal Z_{\rm CM}[0](P,Q)$.
Then, for any $\sigma\in S_n$,
\begin{equation}
\label{eq:propcomplexmm}
        \left\langle \Tr_{[\sigma]}(M^\dagger M)\right\rangle[0]
(P,Q)
        = N^{-n} \sum_{\mu\in S_n}
          \Tr_{[\mu]}(Q)\,
          \Tr_{[\mu\sigma]}(P)\; .
\end{equation}
\end{proposition}
\begin{proof}
The proof is given in Appendix~\ref{app:proof}.
\end{proof}

\noindent The formula of Proposition~\ref{avM} has a natural geometric interpretation in terms of the ribbon graphs generated by the complex matrix model. Each propagator of $M$ in~\eqref{eq:Mprop} carries two matrix structures: $Q$ is coupled to the \emph{first} index of $M$ (equivalently, the second index of $M^\dagger$), while $P$ is coupled to the \emph{second} index of $M$.

Because every propagator is represented by a ribbon with two strands—one following the first index and one following the second—the Feynman expansion automatically produces two independent families of faces:
\begin{itemize}
  \item faces traced by the first index (the “$Q$–faces”),
  \item faces traced by the second index (the “$P$–faces”).
\end{itemize}
A closed cycle of the first index of length $k$ appears in the ribbon graph precisely as a face of length $k$, and contributes the weight $\Tr(Q^{k})$. Similarly, a closed cycle of the second index of length $k$, represented as a face of length $k$, contributes $\Tr(P^{k})$. Thus the amplitude of a ribbon graph factorizes into a product of a $Q$–dependent contribution from first-index faces and a $P$–dependent contribution from second-index faces.

This geometric picture is exactly what is captured algebraically in Proposition~\ref{avM}. The permutation $\mu$ specifies the cycles of the first index and thus the $Q$–faces, while $\mu\sigma$ specifies the cycles of the second index and thus the $P$–faces. Each cycle contributes a trace of the corresponding matrix, so the expression naturally splits into a product of two independent trace invariants.

A schematic depiction of the two–stranded propagator and the two families of faces appears in Figure~\ref{fig:mface} locally for the “$Q$–faces” and Figure~\ref{fig:equivex1m} in a given example of a ribbon graph. 

\paragraph{Graph expansion.}
If $V(A)$ is analytic and a class-invariant, then its exponential has an expansion
\[
e^{V(A)}=\sum_{n=0}^\infty\;\sum_{\sigma\in S_n}
   \frac{1}{n!}\,w_{[\sigma]}\Tr_{[\sigma]}(A)\;,
\]
where $w$ is a complex parameter for a given equivalence class of an symmetric group element.
Substituting this expression into~\eqref{eq:ZCM} yields
\begin{equation}
    \frac{\mathcal Z_{\rm CM}[V](P,Q)}{\mathcal Z_{\rm CM}[0](P,Q)}
      = \sum_{n=0}^\infty \sum_{\sigma\in S_n}\sum_{\gamma\in S_n}
        \frac{1}{n!}\, w_{[\gamma\sigma]}\,
        N^{-n}\Tr_{[\gamma]}(Q)\Tr_{[\sigma]}(P)\;.
\end{equation}
Each term in the triple sum represents a ribbon graph with two families of faces weighted respectively by $Q$ and $P$ as illustrated in Figure \ref{fig:equivex1m}.

\paragraph{Structural link to DWMM.} 
The DWMM \eqref{eq:Z_DY} weights vertices by $\Tr(X^k)$ and faces by $\Tr(Y^k)$. On the other hand, in the present complex matrix model \eqref{eq:ZCM}, the two matrices $(P,Q)$ weight two complementary sets of faces. Therefore, the complex matrix model $\mathcal{Z}_{CM}$ \eqref{eq:ZCM} can be seen as a DWMM of the type which has two different weights $P$ and $Q$ for ribbon-graph faces, while the potential $V$ weights the ribbon-graph vertices.

\subsection{The self-adjoint two--matrix model}
\label{subsec:HMmodel}
\noindent In a parallel manner to the complex matrix model
presented in Section \ref{subsec:GLmodel},
 we consider a self-adjoint two--matrix
model defined on ${\rm H}(N)\times{\rm H}(N)$.
Let $A,B\in{\rm H}(N)$ be self-adjoint matrices distributed according to the partition function
\begin{equation}
\label{eq:ZH0}
\mathcal{Z}_{\rm HM}[V,Y]
  = \iint_{{\rm H}(N)^2} {\cal D}A\,{\cal D}B\;
      e^{- iN\Tr(AB) + V(A) + Y(B) }.
\end{equation}
The interaction $-iN\Tr(AB)$ couples the two matrices  $A$ and $B$ linearly, while $V$ and $Y$ are arbitrary one--matrix potentials. Expectation values of any function $f(A,B)$ in this theory are given by
\begin{equation}\label{eq:avH}
\langle f(A,B)\rangle{[V,Y]}
  := \frac{
      \iint_{{\rm H}(N)^2} {\cal D}A\,{\cal D}B\;
      f(A,B)\,e^{- iN\Tr(AB)+V(A)+Y(B)}
    }{
      \iint_{{\rm H}(N)^2} {\cal D}A\,{\cal D}B\;
      e^{- iN\Tr(AB)+V(A)+Y(B)}
    }\,.
\end{equation}

\paragraph{Gaussian case: $V=Y=0$.}
When both potentials are zero, the model is purely bilinear and the Wick contractions are fixed entirely by the $AB$ term. In particular,
\begin{equation}\label{HM00}
\langle A_{ij}B_{kl}\rangle[0,0]
    = - i N^{-1}\delta_{il}\delta_{kj},\qquad
\langle A_{ij}A_{kl}\rangle[0,0]=0,\qquad
\langle B_{ij}B_{kl}\rangle[0,0]=0.
\end{equation}
Thus the only edges of the ribbon graphs appearing in the Feynman expansion are $AB$-propagators. This propagator is illustarted in Figure~\ref{fig:abprop}.
\begin{figure}[H]
    \centering
    \includegraphics[width=0.3
    \linewidth]{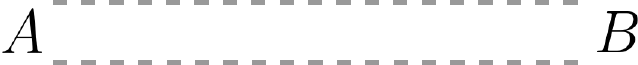}
    \caption{Ribbon graph representation of the propagator \eqref{HM00} of the self-adjoint matrix model defined by the partition function \eqref{eq:ZH}.
    }
    \label{fig:abprop}
\end{figure}

\paragraph{Turning on $Y$ but keeping $V=0$.}
When $Y\neq 0$, the $AB$ and $BB$ propagators are unchanged compared to the first and third expression in \eqref{HM00}. However, the $AA$ propagator receives a  modification that depends only on second derivatives of $Y$ at $B=0$. Namely,
\begin{equation}\label{HM0Y}
\langle A_{ij}A_{kl}\rangle[0,Y]
    = -N^{-2}
      \left.
        \frac{\partial^{2}e^{Y (B)}}
             {\partial B_{ji}\,\partial B_{lk}}
      \right|_{B=0}.
\end{equation}
Remark that no higher derivatives of $Y$ enter the propagator; its structure is entirely determined by the quadratic part of the potential at $B=0$.

A simple example is the Gaussian choice
$Y(B)=Y_{2}[P](B):=-\tfrac12\,N\,\Tr((PB)^{2})$, 
for which \eqref{HM0Y} becomes
\[\label{eq:aprop}
\langle A_{ij}A_{kl}\rangle[0,Y_2[P]]
  = N^{-1}P_{il}P_{kj}.
\]
In this case, the integral over
the field
$B$ can be performed explicitly, giving
\[\label{eq:w2pf}
\mathcal Z_{\rm HM}[V,Y_{2}]
    = {\cal N}
      \int_{{\rm H}(N)} {\cal D}A\;
      e^{-\frac12 N\Tr((P^{-1}A)^{2}) + V(A)},
\]
with
\(
{\cal N}
  = \int_{{\rm H}(N)} {\cal D}B\,e^{-\frac12 N\Tr((PB)^{2})}.
\)
Thus, a Gaussian potential in $B$ reduces the model to an ordinary self-adjoint one--matrix model in~$A$.

\paragraph{Multi-trace observables.}
Since invariant potentials of the form of any analytic function decompose into multi-trace invariants, it is natural to examine the Gaussian averages of these invariants when one of the potentials is zero. For $V=0$ and any $Y$, one has the compact formula
\begin{equation}
\langle \Tr_{[\sigma]}(A)\rangle[0,Y]
    = (iN)^{-n}\,
      \Tr_{[\sigma]} \left(\frac{\partial}{\partial B}\right)
      e^{Y (B)}\Big|_{B=0},
\qquad \sigma\in S_n.
\end{equation}
This identity is useful because it expresses Gaussian multi-trace averages entirely in terms of derivatives of the potential $Y$ without requiring any explicit diagrammatic expansions.

\begin{figure}[H]
    \begin{subfigure}{0.33\textwidth}
    \centering
    \begin{minipage}[c][2cm][c]{\linewidth}
    \centering
    \includegraphics[width=
    0.25\linewidth]{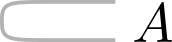}
    \end{minipage}
    \caption{$\Tr\,A$ vertex.}
    \label{fig:avert1}
    \end{subfigure}
    \begin{subfigure}{0.33\textwidth}
    \centering
    \begin{minipage}[c][2cm][c]{\linewidth}
    \centering
    \includegraphics[width=
    0.4\linewidth]{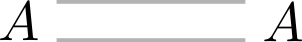}
    \end{minipage}
    \caption{$\Tr\,A^2$ vertex.}
    \label{fig:avert2}
    \end{subfigure}
    \begin{subfigure}{0.33\textwidth}
    \centering
    \begin{minipage}[c][2cm][c]{\linewidth}
    \centering
    \includegraphics[width=
    0.3\linewidth]{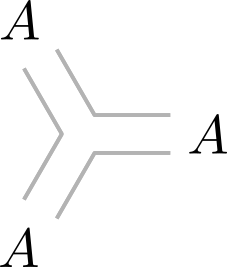}
    \end{minipage}
    \caption{$\Tr\,A^3$ vertex.}
    \label{fig:avert3}
    \end{subfigure}
     \caption{
    Ribbon graph representation of various examples of interactions of the field $A$ in the partition function \eqref{eq:ZH}.
    }
    \label{fig:avert}
\end{figure}

\begin{figure}[H]
    \begin{subfigure}{0.33\textwidth}
    \begin{minipage}[c][2cm][c]{\linewidth}
    \centering
    \includegraphics[width=
    0.25\linewidth]{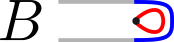}
    \end{minipage}
    \caption{$\Tr Q\;\Tr(PB)$ vertex.}
    \label{fig:bvert1}
    \end{subfigure}
    \begin{subfigure}{0.33\textwidth}
    \begin{minipage}[c][2cm][c]{\linewidth}
    \centering
    \includegraphics[width=
    0.4\linewidth]{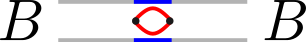}
    \end{minipage}
    \caption{$\Tr Q^2\;\Tr(PB)^2$ vertex.}
    \label{fig:bvert2}
    \end{subfigure}
    \begin{subfigure}{0.33\textwidth}
    \begin{minipage}[c][2cm][c]{\linewidth}
    \centering
    \includegraphics[width=
    0.3\linewidth]{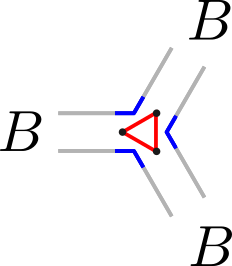}
    \end{minipage}
    \caption{$\Tr Q^3\;\Tr(PB)^3$ vertex.}
    \label{fig:bvert3}
    \end{subfigure}
    \caption{
    Ribbon graph representation of various examples of interactions of the dually-weighted field $B$ in the partition function \eqref{eq:ZH}. The blue corners represent a factor of $P$ in $(PB)^n$, and the red loops represent a factor of $\Tr\,Q^n$.
    }
    \label{fig:bvert}
\end{figure}

\paragraph{The logarithmic (dually weighted) potential.}
The key choice for the equivalence theorem is the dually weighted potential
\[\label{eq:YPQtrln}
Y_{\ln}
[P,Q](B)
    = - \Tr\ln \left(\mathbb 1^{\otimes 2}
      - i\,Q\otimes(PB)\right)=\sum_{k=1}^\infty\frac{i^k}{k}\Tr Q^k \, \Tr(PB)^k ,
\]
leading to the partition function
\begin{equation}\label{eq:ZH}
\mathcal Z_{\rm HM}[V,
Y_{\ln}
[P,Q]]
 = 
    \iint_{{\rm H}(N)^2} {\cal D}A\,{\cal D}B\;
    e^{
         -iN\Tr(AB)
         + V(A)
         - \Tr\ln(\mathbb 1^{\otimes 2}
                 - i\,Q\otimes(PB))
      }.
\end{equation}
Figures \ref{fig:avert} and \ref{fig:bvert} illustrate several interactions that can be considered in this model.
The Gaussian expectations reproduce exactly the two--face structure of the complex matrix model:

\begin{proposition}\label{avAB}
Let $A\in{\rm H}(N)$ be distributed according to
$\mathcal{Z}_{\rm HM}[0,
Y_{\ln}
[P,Q]]$ 
\eqref{eq:ZH0}.  
For any $\sigma\in S_n$,
\begin{equation}\label{eq:prop5.2}
\left\langle \Tr_{[\sigma]}(A)\right\rangle
[0,
Y_{\ln}
[P,Q]]
    = N^{-n} \sum_{\gamma\in S_n}
      \Tr_{[\gamma]}(Q)\,\Tr_{[\gamma\sigma]}(P).
\end{equation}
\end{proposition}
\begin{proof}
    The proof is given in Appendix~\ref{app:proof}.
\end{proof}
\noindent It is straightforward to note that the right-hand side of \eqref{eq:prop5.2} 
in Proposition~\ref{avAB}
is identical to the expression obtained for the complex matrix model in Proposition~\ref{avM}. Thus, the two theories \eqref{eq:ZCM} and \eqref{eq:ZH0} assign the same weights to every multi-trace observable in the Gaussian sector
($ V = 0$).
This strongly suggests a deeper correspondence that goes beyond expectation values. We show that this is not merely a formal similarity: the two partition functions are \emph{exactly} equivalent. The matrix $B$ functions as the intermediate field whose logarithmic potential encodes the covariance of the complex matrix model \eqref{eq:ZCM}, and the identification of both theories follows from integrating out the resuling Gaussian distribution of $M$ from \eqref{eq:ZCM}.

\subsection{Exact matrix equivalence via intermediate fields}
\label{sec:matrixequiv}
\noindent 
 The previous subsections \ref{subsec:GLmodel} and \ref{subsec:HMmodel} showed that for $V=0$, the complex and self-adjoint models match at the level of all multi-trace observables. We show that this agreement is an implication of a more general ($V\neq 0$) result:
the partition function for complex matrices \eqref{eq:ZCM} and the partition function for self-adjoint matrices \eqref{eq:ZH0} with logarithmic potential \eqref{eq:YPQtrln} coincide.
This property has been observed in \cite{difrancesco1992generatingfunctionfatgraphs} for a quadratic potential. Furthermore, the bridge between the two theories is provided by an intermediate--field representation, in which the self-adjoint matrix $B$ 
(i.e., the dually-weighted intermediate field)
of \eqref{eq:ZH} plays the role of a Fourier variable enforcing the constraint $A = M^\dagger M$. Once this constraint is made explicit, the logarithmic potential $Y_{\ln}[P,Q](B)$ \eqref{eq:YPQtrln} naturally emerges from integrating out $M$ from \eqref{eq:ZCM}.
We remark that, depending on the potentials $V$ and $Y$, the partition functions \eqref{eq:ZCM} and \eqref{eq:ZH} may not be convergent. In that case, one may only be interested in their representation as a formal power series, or as a generating functions of discrete geometries.

\medskip

\begin{theorem}
\label{matrixequiv_thm}
Let $\mathcal Z_{\rm CM}[V](P,Q)$ \eqref{eq:ZCM} be the partition function for the complex matrix model, and $\mathcal Z_{\rm HM}[V,Y]$ \eqref{eq:ZH0} be the partition function for the self-adjoint two--matrix model. Then, for $Y=Y_{\ln}[P,Q]$, as given in  \eqref{eq:YPQtrln}, one has that
\[
\frac{\mathcal Z_{\rm CM}[V](P,Q)}{\mathcal Z_{\rm CM}[0](P,Q)}
   = 
\frac{\mathcal Z_{\rm HM}[V,
Y_{\ln}
[P,Q]]}{\mathcal Z_{\rm HM}[0,
Y_{\ln}
[P,Q]]}\;
\]
as formal power series.
\end{theorem}

\begin{proof}
For any $M\in\mathbb C^{N\times N}$, the matrix $M^\dagger M$ is self-adjoint. We insert a delta function enforcing $M^\dagger M=A\in \Hm(N)$,
\begin{equation}
\mathcal Z_{\rm CM}[V](P,Q)
 = \iint_{\mathbb C^{N\times N}} {\cal D}M^\dagger {\cal D}M
   \int_{{\rm H}(N)} {\cal D}A\,
   \delta(M^\dagger M - A)\,
   e^{-N\Tr(P^{-1} M^\dagger Q^{-1} M)+V(A)}.
\end{equation}
Representing the delta function as a Fourier integral over
${\rm H}(N)$, we get that 
\begin{equation}
\label{eq:deltaAM}
\delta(M^\dagger M - A)
  = {\cal N}_{\rm HM}^{-1}
     \int_{{\rm H}(N)} {\cal D}B\;
       e^{\,iN\Tr(B(M^\dagger M - A))},
\end{equation}
where $B\in \Hm(N)$ plays a role of an intermediate field, and 
\[
{\cal N}_{\rm HM}
   = \iint_{{\rm H}(N)^2} {\cal D}A\,{\cal D}B\,
     e^{-iN\Tr(AB)}
   = \mathcal Z_{\rm HM}[0,0]\,.
\]
Substituting this identity \eqref{eq:deltaAM} yields
\begin{equation}\label{eq:Zint}
\mathcal Z_{\rm CM}[V](P,Q)
   = {\cal N}_{\rm HM}^{-1}
      \iint_{\mathbb C^{N\times N}}\hspace{-7mm} {\cal D}M^\dagger {\cal D}M
      \iint_{{\rm H}(N)^2} \hspace{-7mm} {\cal D}A{\cal D}B\;
      e^{
         -N\Tr(P^{-1}M^\dagger Q^{-1}M)
         + iN\Tr(B M^\dagger M)
         - iN\Tr(BA)
         + V(A)
      }.
\end{equation}
For sufficiently well-behaved potentials $V$, for which the integrals are
absolutely convergent, one may exchange the order of integration and perform
the integral over $M$ first. More generally, for arbitrary potential $V$, the same manipulation can be
carried out at the level of formal power series. In either case, the integral
over $M$ is Gaussian, as is most clearly seen in tensor notation,
\[
\Tr \left(P^{-1} M^\dagger Q^{-1} M - B(M^\dagger M)\right)
  = M^*
     \left(
        Q^{-1}\otimes (P^{-1})^{ \tp}
        - i\,\mathbb 1\otimes B^{ \tp}
     \right)
     M.
\]
Thus,
\begin{equation}
\iint_{\mathbb C^{N\times N}} {\cal D}M^\dagger {\cal D}M\;
   e^{-N\Tr(P^{-1}M^\dagger Q^{-1}M)+iN\Tr(BM^\dagger M)}
 = {\cal N}_{\rm CM}\,
   \det(\mathbb 1^{\otimes 2} - i\,Q\otimes (P B))^{-1},
\end{equation}
where
\[
{\cal N}_{\rm CM}
 = \iint_{\mathbb C^{N\times N}} {\cal D}M^\dagger {\cal D}M\,
    e^{-N\Tr(P^{-1}M^\dagger Q^{-1}M)}
 = \mathcal Z_{\rm CM}[0](P,Q).
\]
Using $\det = e^{\Tr\ln}$, we obtain
\begin{equation}
\mathcal Z_{\rm CM}[V](P,Q)
 = {\cal N}_{\rm HM}^{-1} {\cal N}_{\rm CM}
    \iint_{{\rm H}(N)^2} {\cal D}A\,{\cal D}B\;
    e^{
         -iN\Tr(AB)
         + V(A)
         - \Tr\ln(\mathbb 1^{\otimes 2}
                 - i\,Q\otimes(PB))
      }.
\end{equation}
The remaining integral over $A$ and $B$ is precisely $\mathcal Z_{\rm HM}[V,
Y_{\ln}
[P,Q]]$, where the intermediate field $B$ has a dually weighted potential. The constant ${\cal N}_{\rm HM}^{-1} {\cal N}_{\rm CM}$ does not depend on $V$, thus the proof is complete.
\end{proof}

\noindent 
Theorem~\ref{sec:matrixequiv} shows that, after integrating out the complex matrix $M$, the self-adjoint intermediate field $B$ acquires a logarithmic potential of the dually weighted type. The intermediate field thus provides a link between the two matrix models, transforming the two–index-cycle structure of the complex one-matrix model into a self-adjoint two-matrix model.

\emph{Remark.}
The result is obtained at the level of formal power series since, depending on the choice of potential $V$, the partition functions involved may 
not be
absolutely convergent or may only be defined in a distributional sense. In such situations, the equivalence is to be interpreted as an identity between formal expansions. 
However,
for suitable choices of $V$
for which the integrals are well-defined
(e.g.\ when the action is sufficiently bounded so that the integrals are
absolutely convergent, or admit a well-defined Fourier transform), the
equivalence holds at the level of convergent integrals. In this case, the
derivation can be justified by exchanging the order of integration and
performing the Gaussian integral over $M$ explicitly.

In summary, in the \emph{complex matrix model} $\mathcal{Z}_{\rm CM}$, the Gaussian measure couples $Q$ to the first index of $M$ and $P$ to the second. As a consequence, the ribbon graphs generated by the theory naturally carry two independent families of faces: one tracking the index cycle acted on by $Q$, and one tracking the index cycle acted on by $P$. These two families of faces coexist within a single connected Feynman graph and are the geometric origin of the weights appearing in Proposition~\ref{avM}. In the \emph{self-adjoint matrix model}
$\mathcal{Z}_{\rm HM}$,
the interaction produced by the intermediate field $B$ reproduces exactly the same structure: the propagator determines one family of faces, while the dually weighted potential $Y_{\ln}[P,Q]$ supplies the weights for the other. By convention, in our presentation, the $Q$–dependent contribution appears as a ribbon-graph vertex weight and the $P$–dependent contribution as a ribbon-graph face weight, but this is a matter of choice: the roles of $P$ and $Q$ may be exchanged without changing the underlying combinatorics. We can visualize how elementary ribbon-graph elements (edges, corners, and faces) are mapped between the two-matrix model formulations with aid of Figures~\ref{fig:mprop}, ~\ref{fig:mvert}, ~\ref{fig:mface}, ~\ref{fig:abprop}, ~\ref{fig:avert}, ~\ref{fig:bvert} and ~\ref{fig:mbvert}. Vertices of $M$ (Figure~\ref{fig:mvert}) are equivalent to vertices of
the intermediate effective field
$A$ (Figure~\ref{fig:avert}), and faces of $M$ (Figure~\ref{fig:mface}) are equivalent to vertices of 
the dually-weighted intermediate field
$B$ (Figure~\ref{fig:bvert}). Figure~\ref{fig:DWMMex} illustrates this correspondence by displaying an example of the correspondence in a full ribbon graph, where the complex model is recovered from the self-adjoint theory with the dually weighted potential.
\begin{figure}[H]
    \centering
    \includegraphics[width=0.15
    \linewidth]{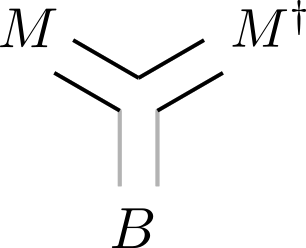}
    \caption{The interaction vertex $\Tr (B M^\dagger M)$, which appears in \eqref{eq:Zint}.
    }
    \label{fig:mbvert}
\end{figure}

\begin{figure}[H]
    \begin{subfigure}{0.3\textwidth}
    \centering
    \includegraphics[width=
    0.9\linewidth]{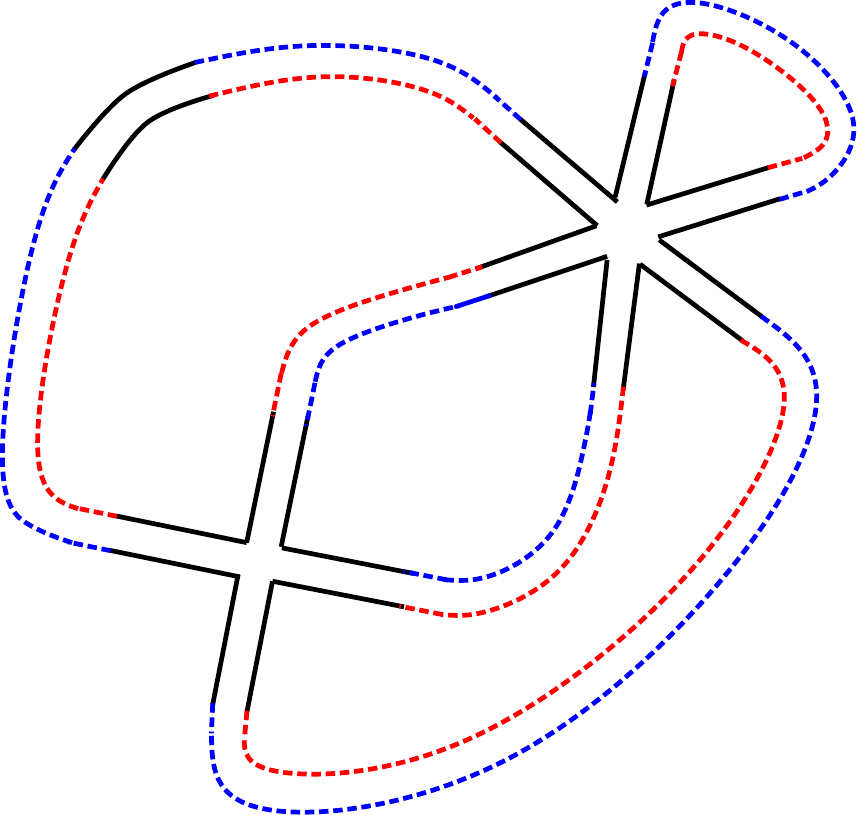}
    \caption{}
    \label{fig:equivex1m}
    \end{subfigure}
    \begin{subfigure}{0.39\textwidth}
    \centering
    \includegraphics[width=
    \linewidth]{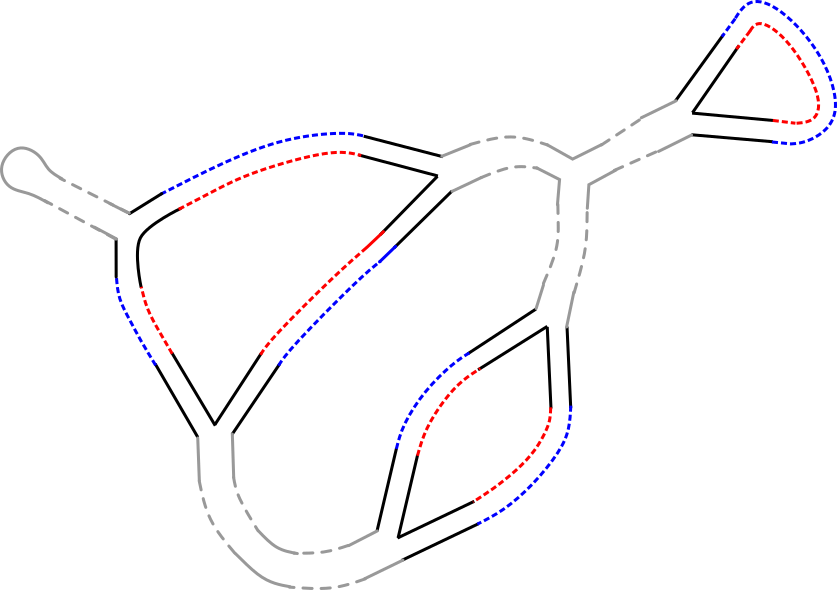}
    \caption{}
    \label{fig:equivex1abm}
    \end{subfigure}
    \begin{subfigure}{0.3\textwidth}
    \centering
    \includegraphics[width=
    \linewidth]{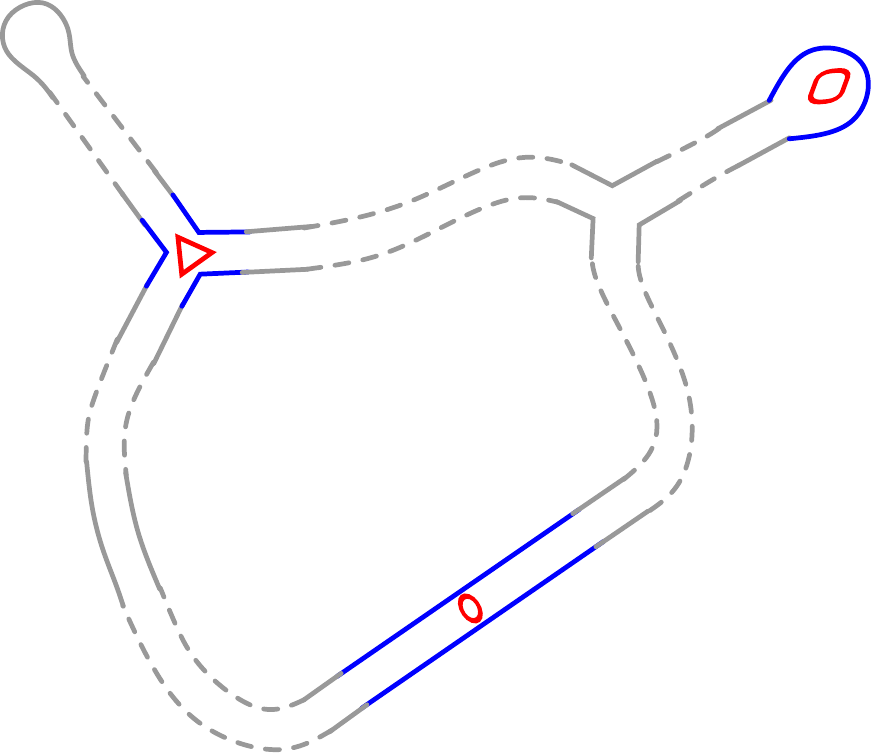}
    \caption{}
    \label{fig:equivex1ab}
    \end{subfigure}
    \caption{An example of equivalence between a graph in the: (a) complex matrix theory ${\cal Z}_{\rm CM}[V](P,Q)$ \eqref{eq:ZCM} using the graphical representations of the propagator, interaction vertices and faces shown in Figures~\ref{fig:mprop},~\ref{fig:mvert}, and~\ref{fig:mface}; (b) intermediate theory \eqref{eq:Zint} involving the  propagators and interaction vertices shown in Figures~\ref{fig:mprop}, \ref{fig:mface}, \ref{fig:abprop}, \ref{fig:avert}, and~\ref{fig:mbvert}; (c)
     the Hermitian matrix theory $\mathcal{Z}_{\rm HM}[V,
Y_{\ln}
    [P,Q]]$ \eqref{eq:ZH} now composed of the propagator and interaction vertices shown in Figure~\ref{fig:abprop}, \ref{fig:avert}, and~\ref{fig:bvert}.
    Gray lines represent intermediate fields $A$ and $B$.
     }
    \label{fig:DWMMex}
\end{figure}

\paragraph{Expectation values.}
\noindent
The equivalence of the partition functions proved in Theorem \ref{matrixequiv_thm} has an immediate implication:
the expectation values of any trace invariant depending only on the self-adjoint combination $M^\dagger M$
can be computed in either model.  
To make this precise, recall that for any function
$f:{\rm H}(N)\rightarrow\mathbb C$, the expectation values in each of the two
models can be written as
\[
\left\langle f(M^\dagger M)\right\rangle[V](P,Q)
           = \frac{\mathcal{Z}_{\rm CM}[V+\ln f](P,Q)}
           {\mathcal{Z}_{\rm CM}[V](P,Q)},
           \qquad
\left\langle f(A)\right\rangle[V,Y]
    = \frac{\mathcal{Z}_{\rm HM}[V+\ln f,\,Y]}
           {\mathcal{Z}_{\rm HM}[V,Y]}.
\]
Applying Theorem~\ref{matrixequiv_thm} to these ratios gives an exact
identity between all corresponding correlators. This is stated in the following Corollary.

\begin{corollary}
\label{cor:exp_val_mm}
For any potential $V$ and any function $f:{\rm H}(N)\to\mathbb C$,
\[
\left\langle f(M^\dagger M)\right\rangle[V](P,Q)
    = \left\langle f(A)\right\rangle[V,Y_{\ln}[P,Q]].
\]
\end{corollary}
\noindent
A direct implication is that for $V=0$ and $f(A)=\Tr_{[\sigma]}(A)$,
\[\label{eq:traces_PQ}
 \left\langle \Tr_{[\sigma]}(M^\dagger M)\right\rangle[0](P,Q)
   =\left\langle \Tr_{[\sigma]}(A)\right\rangle[0,
Y_{\ln}
   [P,Q]]
    .
\]
Furthermore, using Proposition~\ref{avM} and 
Proposition~\ref{avAB}, we see explicitly that \eqref{eq:traces_PQ} holds.

\medskip

\paragraph{Discussion.}
The subsection \ref{sec:matrixequiv} is the core of the Section \ref{sec:equivmm}: the complex matrix model \eqref{eq:ZCM} and the self-adjoint matrix model \eqref{eq:ZH0} with the logarithmic intermediate potential  \eqref{eq:YPQtrln} are equivalent random geometry theories.
The field $B$ plays the role of an intermediate Fourier variable that enforces $A=M^\dagger M$, while its potential $
Y_{\ln}
[P,Q]$ 
encodes the $P$– and $Q$–dependence arising from integrating out $M$. Thus, the two models are dual descriptions of the same underlying matrix integral, and 
any
expectation value
depending on $M^\dagger M$ may be computed in whichever formulation.

\subsection{Examples}\label{examp}
\noindent The results in the previous subsections have notable implications once specific choices of the potential $V$ and matrices  $(P,Q)$ are made.  
This subsection collects several illustrative cases.  

\paragraph{Dually weighted complex matrix model.}

\noindent Consider the complex matrix model 
$\mathcal{Z}_{\rm CM}$ 
\eqref{eq:ZCM}
with covariance $(P,Q)$ and 
its potential $V$ to be dually weighted
\[
\label{eq:Vln}
V_{\ln}
[K,L](M^\dagger M)
= -\Tr\ln\left({\bb 1}^{\otimes 2}-i L\otimes(K M^\dagger M)\right)
\;
\]
for constant $N\times N$ matrices $K$ and $L$.
Its partition function reads
\[
\label{eq:ZCMVln}
\mathcal{Z}_{\rm CM}[
V_{\ln}
[K,L]](P,Q)
=\int_{\mathbb C^{N\times N}} {\cal D}M^\dagger{\cal D}M\;
e^{-N\Tr(P^{-1}M^\dagger Q^{-1}M)
-\Tr\ln({\bb 1}^{\otimes 2}-i L\otimes(K M^\dagger M))}
\;.
\]
As in Section~\ref{subsec:GLmodel}, contractions of the first index of the matrix
$M$ generate color--1 faces weighted by $\Tr(Q^n)$. Given the form of the dually weighted potential 
$
V_{\ln}
[K,L]$, 
each vertex carries a weight $\Tr(L^n)$, just as in the self-adjoint one-matrix model. However, in this case, the face structure is enriched by the presence of two independent families of cycles built into the complex matrix model. This is because contractions of the second index of the matrix $M$ generate color--2 faces weighted by $\Tr((P K)^n)$. In summary, for each $n$, one obtains
\begin{itemize}
    \item Vertices of degree $n$ with weight $\Tr(L^n)$.
    \item Color--1 faces of length $n$ with weight $\Tr(Q^n)$.
    \item Color--2 faces of length $n$ with weight $\Tr((P K)^n)$.
\end{itemize}
This reproduces the pattern of Feyman graphs characteristic of dually weighted models. However, in this case, the two families of faces arise from the two distinct indices of one complex matrix, $M$, instead of two distinst matrices $A$ and $B$. 
One may notice a certain redundancy in having both the constant matrices $P$ and $K$, as the model only depends on their product $PK$. This is directly connected to the fact that the DWMM can also be formulated by placing the constant matrix in the propagator instead of the potential.
\medskip

\noindent For random matrices $A,B\in{\rm H}(N)$, consider the self-adjoint model
\[
\mathcal{Z}_{\rm HM}[
V_{\ln}
[K,L],Y_{\ln}[P,Q]]
=\int dA dB\;e^{-iN\Tr(AB)
-\Tr\ln({\bb 1}^{\otimes 2}-i L\otimes(KA))
-\Tr\ln({\bb 1}^{\otimes 2}-i Q\otimes(PB))}
\;.
\]
In this case, ribbon graphs have the following properties:
\begin{itemize}
    \item $A$-vertices of degree $n$ 
each 
    with weight $\Tr(L^n)$.
    \item $B$-vertices of degree $n$ 
each 
    with weight $\Tr(Q^n)$.
    \item Faces of length $n$ 
each 
    with weight $\Tr((P K)^n)$.
\end{itemize}
This is a bipartite-graph version of
dually weighted matrix models \cite{Kazakov_1996, das_1990}, and according to Theorem~\ref{matrixequiv_thm}, we have
\[
\frac{\mathcal{Z}_{\rm CM}[
V_{\ln}
[K,L]](P,Q)}{\mathcal{Z}_{\rm CM}[0](P,Q)}
=
\frac{\mathcal{Z}_{\rm HM}[
V_{\ln}
[K,L],
Y_{\ln}
[P,Q]]}{\mathcal{Z}_{\rm HM}[0,
Y_{\ln}
[P,Q]]}
\;.
\]

\paragraph{$\mathbf{Q=C_2}$: The causally inspired model.}

An interesting case arises when $Q=C_2$ (See~\eqref{cnprop}).
We recall that the matrix $C_2$ is only well defined in the large $N$ limit \cite{Benedetti:2008hc}. Therefore,  consider the large-$N$ partition function\footnote{This is similar to the partition function seen in \eqref{eq:w2pf}, but it uses the matrix $C_2$, which is only well-defined in the large-N limit.}
\begin{equation}\label{abzc2}
    \mathcal{Z}_{\rm HM}[V,
Y_{\ln}
    [P,C_2]]
    =\int dA dB\;e^{-iN\Tr(AB)+\frac{1}{2}N\Tr((PB)^2)+V(A)}
    \;.
\end{equation}
The integration over 
the intermediate field $B$ is Gaussian and gives
\begin{equation}\label{c2pf}
    \tilde{\mathcal{Z}}_{\rm HM}
    ={\cal N}\int dA\;e^{-\frac{1}{2}N\Tr((P^{-1}A)^2)+V(A)}
    \;,
\end{equation}
where
${\cal N}
  = \int_{{\rm H}(N)} {\cal D}B\,e^{-\frac12 N\Tr((PB)^{2})}.$
Thus, up to a constant depending only on $N$, one recovers a self-adjoint one-matrix model.

Applying Proposition~\ref{avAB} with $Q=C_2$ yields
\[
    \left\langle\Tr_{[\sigma]}(A)\right\rangle[0,
Y_{\ln}
    [P, C_2]]
    = N^{-n}\sum_{\gamma\in S_n}
      \Tr_{[\gamma\sigma]}(P)\Tr_{[\gamma]}(C_2)
    \;.
\]
Using the defining property \eqref{cnprop}, this implies
\[
\left\langle\Tr_{[\sigma]}(A)\right\rangle[0,
Y_{\ln}
[P, C_2]]=0
\qquad\text{if $n$ is odd}\;,
\]
and
\[
\left\langle\Tr_{[\sigma]}(A)\right\rangle[0,
Y_{\ln}
[P, C_2]]
= N^{-m}\sum_{\gamma\in [2^m]}\Tr_{[\gamma\sigma]}(P)
\qquad\text{if $n=2m$ is even}\;.
\]

By Theorem~\ref{sec:matrixequiv}, the corresponding complex matrix model is
\begin{equation}
\label{ZCM}
{\cal Z}_{\rm CM}[V](P,C_2)
=\int dM^\dagger dM\;e^{-N\Tr(P^{-1}M^\dagger C_2^{-1}M)+V(M^\dagger M)}
\;.
\end{equation}
Corollary~\ref{cor:exp_val_mm} then gives
\[
\left\langle\Tr_{[\sigma]}(M^\dagger M)\right\rangle[0](P,C_2)
= N^{-m}\sum_{\gamma\in [2^m]}\Tr_{[\gamma\sigma]}(P)
\;,
\]
for $n=2m$ even and $1\le n\le N$.  
This correspondence is illustrated in Figure~\ref{fig:equivex2}.
\begin{figure}[H]
    \begin{subfigure}{0.24\textwidth}
    \centering
    \includegraphics[width=
    0.8\linewidth]{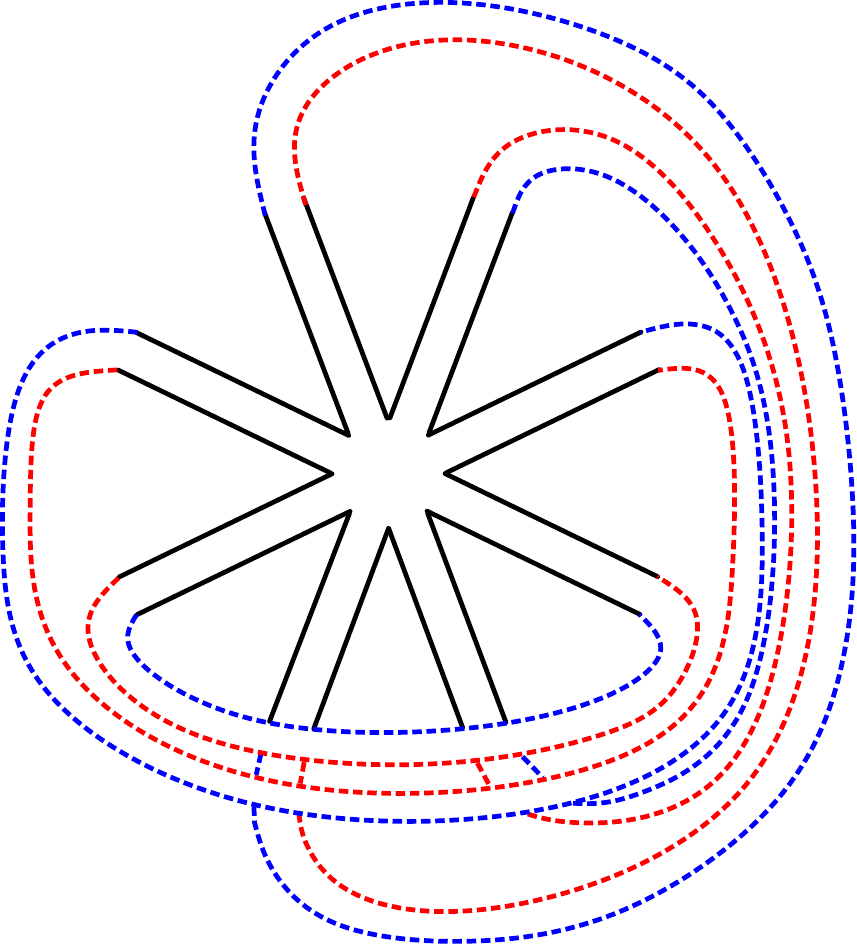}
    \caption{}
    \label{fig:equivex2m}
    \end{subfigure}
    \begin{subfigure}{0.24\textwidth}
    \centering
    \includegraphics[width=
    0.9\linewidth]{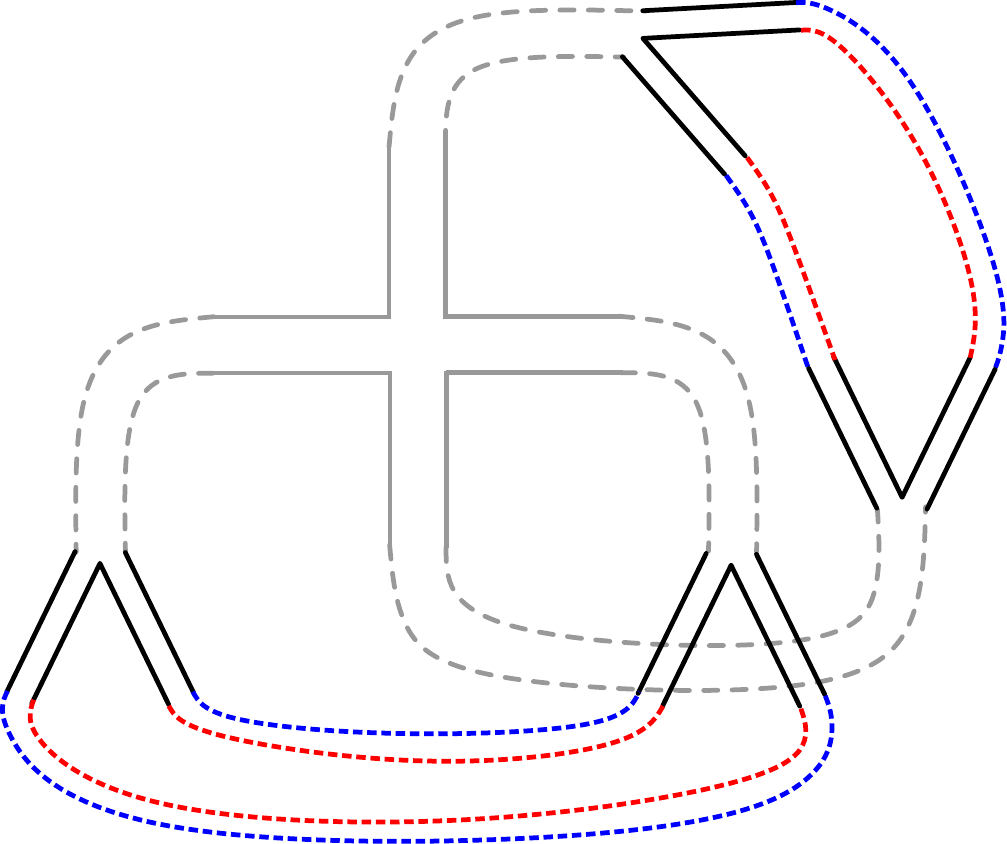}
    \caption{}
    \label{fig:equivex2abm}
    \end{subfigure}
    \begin{subfigure}{0.24\textwidth}
    \centering
    \includegraphics[width=
    0.9\linewidth]{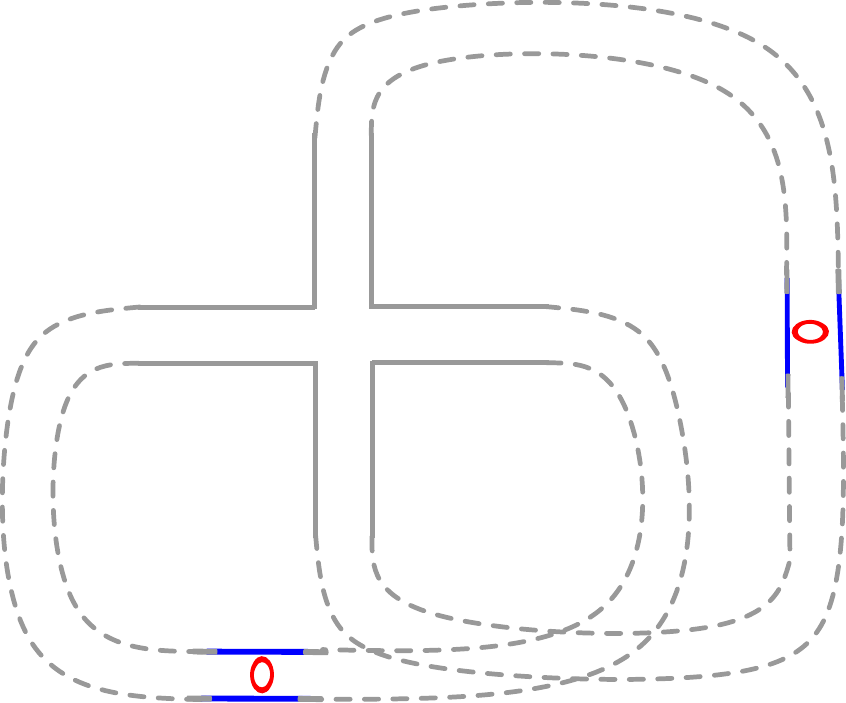}
    \caption{}
    \label{fig:equivex2ab}
    \end{subfigure}
    \begin{subfigure}{0.24\textwidth}
    \centering
    \includegraphics[width=
    0.8\linewidth]{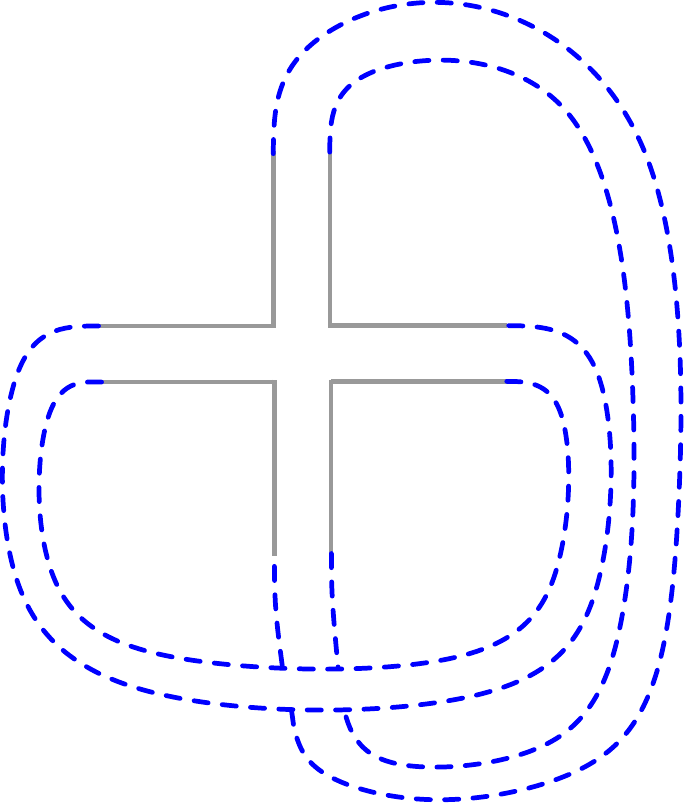}
    \caption{}
    \label{fig:equivex2a}
    \end{subfigure}
    \caption{An example of equivalence between a graph in the: (a) complex matrix theory ${\cal Z}_{\rm CM}[V](P,C_2)$ \eqref{ZCM}; (b) intermediate theory \eqref{eq:Zint}; (c)
     the Hermitian matrix theory $\mathcal{Z}_{\rm HM}[V,
Y_{\ln}
    [P,C_2]]$ \eqref{eq:ZH}; (d) 
     the Hermitian matrix theory $\mathcal{Z}_{\rm HM}[V,
Y_{\ln}
    [P,C_2]]$ \eqref{c2pf} after integrating $B$. The dashed blue lines represent the propagator of $A$ in \eqref{eq:aprop}.
    Intermediate fields $A$ and $B$ are represented in gray lines.
}
    \label{fig:equivex2}
\end{figure}

\noindent Furthermore, when $P={\bb 1}$,
\[\label{eq:Catalan}
\left\langle N^{-1}\Tr[(M^\dagger M)^n]\right\rangle[0]
= {\rm Cat}(m)\left(1+O(N^{-2})\right)
\;,
\]
where ${\rm Cat}(m)$ are the Catalan numbers.  
If $V(M^\dagger M)=-N\frac{g^2}{2}(M^\dagger M)^2$, then
\[\label{eq:quadmatrix}
\mathcal{Z}_{\rm HM}[V,
Y_{\ln}
[{\mathds{1}},C_2]]
=\int dA\;e^{-N\Tr\left(\frac{1}{2}A^2-\frac{g^2}{2}A^2\right)}
\;,
\]
showing that the model \eqref{model1} becomes equivalent to the Gaussian
model above \eqref{eq:quadmatrix}.
Appendix~\ref{app:ampli} explores the case when $Q = C_2$ further via explicit graph
resummation.

\paragraph{$\mathbf{Q=C_k}$: $k$-uniform rigidity and a partition-matroid structure.}

For general $k$, consider
\begin{equation}\label{abzck}
\mathcal{Z}_{\rm HM}[V,Y[P,C_k]]
=\iint dAdB\;
e^{
-iN\Tr(AB)+\frac{i^k}{k}N\Tr((PB)^k)+V(A)
},
\end{equation}
and the corresponding complex matrix model
\[
{\cal Z}_{\rm CM}[V](P,C_k)
=\int dM^\dagger dM\;
e^{-N\Tr(P^{-1}M^\dagger C_k^{-1}M)+V(M^\dagger M)}.
\]
Using \eqref{cnprop}, one has
\[
\Tr_{[\gamma]}(C_k)=0\ \text{unless }n=kd,
\qquad
\Tr_{[\gamma]}(C_k)=N^d\,\delta_{[\gamma],[k^d]}.
\]
Corollary~\ref{cor:exp_val_mm} then gives
\[\label{eq:matroidstructure}
\big\langle\Tr_{[\sigma]}(M^\dagger M)\big\rangle[0](P,C_k)
=
\big\langle\Tr_{[\sigma]}(A)\big\rangle[0, 
Y_{\ln}
[P,C_k]]
= N^{-d}\sum_{\gamma\in[k^d]}\Tr_{[\gamma\sigma]}(P).
\]

\smallskip
\noindent
\emph{Combinatorial meaning (``$k$-uniform blocks'').}
The constraint $[\gamma]=[k^d]$ means that the $Q$--faces (equivalently, the faces traced by the
$C_k$-weighted index) decompose into disjoint cycles of length exactly $k$.
Thus, the relevant strand-corners of the diagram are naturally partitioned into disjoint $k$-tuples,
one $k$-tuple for each such face.

\smallskip
\noindent
\emph{Matroid viewpoint.}
Let $E$ be the set of these corners (or, equivalently, the set of half-edges incident to the
$C_k$-faces), and let $\{E_1,\dots,E_d\}$ be the induced partition of $E$ into the $d$ disjoint
blocks of size $k$.
Then, there is a canonical partition matroid on $E$,
\[
\mathcal{I}=\{I\subseteq E:\ |I\cap E_r|\le 1\ \text{for all }r\},
\]
encoding the fact that each $C_k$-face behaves as a single $k$-ary ``rigid unit'':
one may choose at most one corner from each $k$-block.
In this sense, the choice $Q=C_k$ implements a 
$k$-uniform 
rigidity whose underlying incidence
structure is naturally organized by a partition matroid (a direct sum of rank--$1$ uniform matroids
on $k$ elements).

\subsection{Miscellaneous results}\label{relres}
\noindent Many analytic techniques for random matrix models rely on rewriting the partition
function or multitrace invariants using the representation theory of the symmetry group of the matrix ensemble (in this case,
${\rm U}(N)$ or ${\rm GL}(N)$).  
In this setting, a standard approach is the \emph{character expansion}, in which one decomposes
functions of a matrix $A$ (such as $e^{\Tr V(A)}$ or products of traces) into
linear combinations of irreducible characters $\chi_\lambda(A)$,
where $\lambda$ is a partition that labels an irreducible representation of ${\rm U}(N)$ or ${\rm GL}(N)$.
This method is particularly useful when the measure is invariant under
conjugation, since the characters form an orthonormal basis for class
functions and the expansion isolates the dependence on the eigenvalues.  

For the models studied in this paper, Proposition~\ref{avAB} provides closed
expressions for the expectation value of all multi-trace invariants of $A$. Therefore, it gives exactly the
data required to evaluate $\langle \chi_\lambda(A)\rangle$ for any partition
$\lambda$.  
Thus the following corollary packages the 
expectation-value
information of
Proposition~\ref{avAB} into a representation–theoretic form suitable for direct
use in character expansions, and connects our results to the standard
Itzykson--Zuber/Weingarten character techniques commonly used in random matrix
theory.

For the following corollary, we recall the partition function $\mathcal{Z}_{\rm HM}$ given in \eqref{eq:ZH0}, and the expectation value $\langle\cdot\rangle$ given in \eqref{eq:avH}.

\begin{corollary}
    Let $A\in \mathrm{Herm}(N)$ be a random matrix distributed
according to ${\cal Z}_{\mathrm{HM}}[0,Y_{\ln}[P,Q]]$.
Then, for any representation $r$ of $\GL(N)$,
\begin{equation}
    \left\langle \chi_{r}(A)\right\rangle[0,Y_{\ln}[P,Q]]=n!N^{-n} \frac{1}{\chi_r(e)}\chi_{r}(P)\chi_r(Q)\;,
\end{equation}
where $n$ is the size of $r$ and $e\in\Sy_n$ is the identity permutation.
\end{corollary}

\begin{proof}
The matrix characters have the multi-trace expansion
\begin{equation}
    \chi_r(A)=\frac{1}{n!}\sum_{\sigma\in {\rm S}_n}\chi_{r}(\sigma)\Tr_{[\sigma]}(A)\;.
\end{equation}
We can use Proposition~\ref{avAB} to find that
\begin{equation}
    \left\langle \chi_r(A)\right\rangle[0,Y_{\ln}[P,Q]]=\frac{1}{n!}N^{-n} \sum_{\substack{\sigma,\gamma\in {\rm S}_n}}\chi_{r}(\sigma)\Tr_{[\gamma \sigma]}(P)\Tr_{[\gamma]}(Q)\;.
\end{equation}
Using the Schur-Weyl duality decomposition
\begin{equation}
    \left\langle \chi_r(A)\right\rangle[0,Y_{\ln}[P,Q]]=\frac{1}{n!}N^{-n}\sum_{\substack{\sigma,\gamma\in {\rm S}_n}}\chi_{r}(\sigma)\sum_{s}\chi_{s}(\gamma\sigma)\chi_s(P)\Tr_{[\gamma]}(Q)\;.
\end{equation}
Using the sum over a class,
\begin{equation}
    \left\langle \chi_r(A)\right\rangle[0,Y_{\ln}[P,Q]]=\frac{1}{n!}N^{-n}\sum_{\substack{\sigma,\gamma\in {\rm S}_n}}\chi_{r}(\sigma)\sum_{s}\frac{\chi_{s}(\gamma)\chi_s(\sigma)}{\chi_s(e)}\chi_s(P)\Tr_{[\gamma]}(Q)
    \;,
\end{equation}
using character orthogonality
\begin{equation}
    \left\langle \chi_r(A)\right\rangle[0,Y_{\ln}[P,Q]]=N^{-n} \sum_{\substack{\gamma\in {\rm S}_n}}\sum_{s}\frac{\chi_{s}(\gamma)\delta_{rs}}{\chi_s(e)}\chi_s(P)\Tr_{[\gamma]}(Q)
    \;.
\end{equation}
Summing over $s$
\begin{equation}
    \left\langle \chi_r(A)\right\rangle[0,Y_{\ln}[P,Q]]=N^{-n} \sum_{\substack{\gamma\in {\rm S}_n}}\frac{\chi_{r}(\gamma)}{\chi_r(e)}\chi_r(P)\Tr_{[\gamma]}(Q)
    \;.
\end{equation}
Summing over $\gamma$
\begin{equation}
    \left\langle \chi_r(A)\right\rangle[0,Y_{\ln}[P,Q]]=n!N^{-n} \frac{1}{\chi_r(e)}\chi_r(P)\chi_r(Q)
    \;.
\end{equation}

\end{proof}
This expression is valid for any $N>0$ and $n\geq 0$. If $n\leq N$, we can use $\chi_r(e)=N^{-n}n!\chi_r(C_1)$. For the following corollary, we recall  the matrix $C_1$ given in \eqref{cnprop}, the partition function ${\cal Z}_{\mathrm{HM}}$ given in \eqref{eq:ZH0}, and the expectation value $\langle\cdot\rangle$ given in \eqref{eq:avH}.

\begin{corollary}\label{Cor:character_general_A}
Let $A\in \mathrm{Herm}(N)$ be a random matrix distributed
according to ${\cal Z}_{\mathrm{HM}}[0,Y_{\ln}[P,Q]]$.
Then, for any representation $r$ of $\GL(N)$ with size $n\leq N$,
\begin{equation}
    \left\langle \chi_r(A)\right\rangle[0,Y_{\ln}[P,Q]]=\frac{\chi_r(P)\chi_r(Q)}{\chi_r(C_1)}\;.
\end{equation}
\end{corollary}
\noindent Since the character $\chi_r(e)$ does not depend on $N$, one has that, at fixed $n$,
\begin{equation}
    \chi_r(e)=n!\lim_{N\rightarrow \infty}N^{-n}\chi_r(C_1)\;.
\end{equation}
This means that, by defining 
\begin{equation}
    \tilde\chi_r(C_1):=N^{n}\lim_{N\rightarrow \infty}N^{-n}\chi_r(C_1)\;,
\end{equation}
one has the following corollary.
\begin{corollary}\label{Cor:character_more_general_A}
    For any representation $r$ of $\GL(N)$ with size $n$, the expectation value of $\chi_r$ is 
\begin{equation}
    \left\langle \chi_r(A)\right\rangle[0,Y_{\ln}[P,Q]]=\frac{\chi_r(P)\chi_r(Q)}{\tilde\chi_r(C_1)}\;.
\end{equation}
\end{corollary}

\noindent This result allows us to evaluate several integrals through character expansion methods. This is in a way a character version of Proposition~\ref{avAB}, telling us that
\begin{corollary}
    For any class function $f$ of $A$, and given its expansion 
\begin{equation}
    f=\sum_{n=0}^{\infty}\sum_{r\vdash n}f
    ^{r}\chi_{r}\;,
\end{equation}
the expectation value of $f$ is 
\begin{equation}
    \left\langle f(A)\right\rangle[0,Y_{\ln}[P,Q]]=\sum_{n=0}^\infty \sum_{r\vdash n}f^{r}\frac{\chi_r(P)\chi_r(Q)}{\tilde\chi_r(C_1)}\;.
\end{equation}
\end{corollary}

\begin{corollary}
    For any representation $r$ of $\GL(N)$ with size $n\leq N$, and given the expansion of its character
\begin{equation}
    \chi_r(M^\dagger M)=\frac{1}{n!}\sum_{\sigma\in {\rm S}_n}\chi_{r}(\sigma)\Tr_{[\sigma]}(M^\dagger M)\;,
\end{equation}
where $n$ is the size of $r$, the expectation value of $\chi_r$ is 
\begin{equation}
    \left\langle \chi_r(M^\dagger M)\right\rangle[0](P,Q)=\frac{\chi_r(P)\chi_r(Q)}{\chi_r(C_1)}\;.
\end{equation}
\end{corollary}

\section{Equivalence in random tensor models}
\label{sec:equivtensor}

\noindent
In this section we generalize the matrix equivalence to tensor models. We establish an explicit correspondence between random complex tensor models with covariance $R$ \footnote{Similarly to the matrix model case, this term is used to indicate that the tensor $R$ appears in the quadratic part of the action. It is not a formal statement about the probabilistic nature of these models.} and interaction $W(\phi,\phi^\dagger)$, and self-adjoint tensor models whose intermediate field carries a logarithmic potential. As in the matrix case, the equivalence holds at the level of partition functions and all 
trace-invariant observables
built from $\phi\phi^\dagger$ (or equivalently, from the self-adjoint field), and, in suitable situations, it admits a further reduction to lower–order self-adjoint tensors.

\subsection{The complex tensor model}
\label{subsec:complex-tensor}
\noindent We introduce the tensor analogue of the complex matrix model presented in Section~\ref{subsec:GLmodel}. Let $\phi\in\Cp_N^{\otimes D}$ be a complex order–$D$ tensor with Gaussian covariance $R\in\Cp_N^{\otimes D}\otimes\Cp_N^{*\otimes D}$ and interaction potential $W$ in the form of any function of $\phi$ and $\phi^\dagger$.  
The partition function we consider is
\begin{equation}
\label{ZCTs}
    {\cal Z}_{\rm CT}[W]
    (R)
    =\iint {\cal D}\phi^\dagger {\cal D}\phi\;e^{-\phi^\dagger 
    R^{-1}
    \phi+W(\phi, {\phi^\dagger})}\;.
\end{equation}
For any function $f$ of $\phi$ and $\phi^\dagger$, its expectation value in this model is given by
\begin{equation}
    \langle f(\phi,\phi^\dagger)\rangle[W]
    (R)
    =\frac{\iint {\cal D}\phi^\dagger {\cal D}\phi\;f(\phi,\phi^\dagger)\;e^{-\phi^\dagger 
    R^{-1}
    \phi+W(\phi,{\phi^\dagger})}}{\iint {\cal D}\phi^\dagger {\cal D}\phi\;e^{-\phi^\dagger 
    R^{-1}
    \phi+W(\phi,{\phi^\dagger})}}\;.
\end{equation}

\paragraph{Propagator.}
In the Gaussian model,
the covariance is simply
\begin{equation}
    \langle \phi\phi^\dagger\rangle[0]
    (R) = R
    \;.
\end{equation}
At this stage, the model behaves as a vector theory in the large tensor space, that is, the detailed structure of the $D$ tensor indices has not yet been used.

\paragraph{Trace invariants of tensors.}

To probe the genuine tensorial structure, we focus on invariants under the natural $U(N)^D$ action on $\Cp_N^{\otimes D}$. A distinguished family of such invariants is given by the trace invariants of tensors (see Section~\ref{sec:notations}), indexed by $\bm\sigma=(\sigma_1,\dots,\sigma_D)\in\Sy_n^D$. The following proposition gives their Gaussian ($W=0$) expectation values.

\begin{proposition}\label{avTM}
If $\phi$ is a random tensor distributed according to the partition function ${\cal Z}_{\rm CT}[0](R)$ \eqref{ZCTs}, then, for any $\bm\sigma\in \Sy_n^D$, the 
${\bm\sigma}$–bubble/trace invariant 
satisfies
\begin{equation}\label{eq:trphieexpec}
    \langle\Tr_{[{\bm\sigma}]}(\phi\phi^\dagger)\rangle[0]
    (R)
    =\sum_{\mu\in \Sy_n}\Tr_{[\mu \bm\sigma]}
    (R)
    \;.
\end{equation}
\end{proposition}
\begin{proof}
    The proof is given in Appendix~\ref{app:proof}.
\end{proof}

\noindent Thus, as in the matrix case, amplitudes are expressed entirely in terms
of the covariance 
$R$,
but they are now organized by the 
tensor trace-invariant
structure encoded in $\bm\sigma$.

\paragraph{Factorized covariance.}

A particularly simple structure appears when the covariance factorizes as\footnote{In particular, when $D=2$ and $R_1=Q$ and $R_2=P$, this reduces to the matrix model case.}
\[
R=\bigotimes_{c=1}^D R_c\;.
\]
In that case, the Gaussian propagator reads
\[\label{propfact}
\langle \phi\phi^\dagger\rangle[0]
(R)=\bigotimes_{c=1}^D R_c
\]
and is illustrated in Figure~\ref{fig:tensprop}.
\begin{figure}[H]
    \centering
    \includegraphics[width=0.3
    \linewidth]{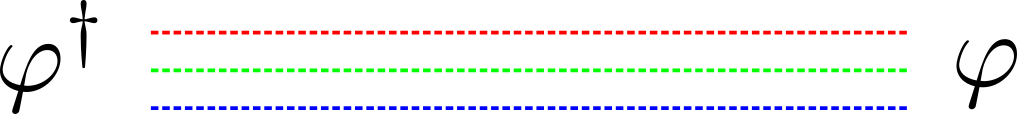}
    \caption{Ribbon graph representation of the propagator of the theory defined by the partition function \eqref{ZCTs} as given in \eqref{propfact}. The blue, green, and red strands respectively represent the presence of $R_1, R_2$, and $R_3$ in $R=R_1\otimes R_2\otimes R_3$.
    }
    \label{fig:tensprop}
\end{figure}
Furthermore, the expectation value \eqref{eq:trphieexpec} factorizes accordingly, as stated in the following Corollary.

\begin{corollary}
If $R=\bigotimes_{c}R_c$, then
\begin{equation}
    \langle\Tr_{[\bm\sigma]}(\phi\phi^\dagger)\rangle[0]
    (R)
    = N^{-n}\sum_{\mu\in S_n}
      \prod_{c=1}^D \Tr_{[\mu\sigma_c]}
      (R_c)
      \;.
\end{equation}
\end{corollary}

\noindent
\noindent
This factorized expression has a
clear
geometric meaning in the stranded Feynman-graph representation of tensor models: it decomposes into $D$ independent families of faces, one for each tensor index. This is because when the covariance factorizes as $R=\bigotimes_{c=1}^D R_c$,  the tensor propagator splits into $D$ independent strands, where the $c$-th strand
represents the propagator weight $R_c$.
This implies that the Wick contractions separate by color: the same permutation $\mu\in S_n$ governs the pairings of all tensor slots, but in each color $c$, it combines with the bubble permutation $\sigma_c$ to form $\mu\sigma_c$. The cycles of $\mu\sigma_c$ are precisely the faces traced by the $c$-th index of the tensor, and a cycle of length $l$ corresponds to a color-$c$ face of length $l$ in the stranded graph. Since a color-$c$ strand
represents a propagator with weight $R_c$,
 a color-$c$ face of length $l$ has a weight  $\Tr(R_c^l)$. This is the tensor analogue of the dually weighted structure of the complex matrix model presented in Section \ref{subsec:GLmodel}, now extended to $D$ colors.

\paragraph{Symmetry classes of potentials.}

So far, the complex model~\eqref{ZCTs} is defined for a completely general potential
$W(\phi,\phi^\dagger)$.  
In practice, we are interested in three nested symmetry classes given by:
\begin{itemize}
    \item scalar-$U(1)$ invariance:  
    for any $z\in U(1)$,
    \(
    W(z\phi,(z\phi)^\dagger )=W(\phi,\phi^\dagger)
    \).

    \item 
repeated-$U(N)$ invariance:
    for any $U\in U(N)$,
    \(
    W(U^{\otimes D}\phi,(U^{\otimes D}\phi)^\dagger )=W(\phi,\phi^\dagger)
    \).

    \item 
partial-$U(N)$ invariance:
    for any $U\in U(N)$,
    \(
    W(\check U\phi,(\check U\phi)^\dagger )=W(\phi,\phi^\dagger)
    \),
    where $\check U:=U\otimes {\bb 1}^{\otimes (D-1)}$.

    \item full-$U(N)^D$ invariance:  
    for any $U\in U(N)^D$,
    \(
    W(U\phi,(U\phi)^\dagger )=W(\phi,\phi^\dagger)
    \).
\end{itemize}

\noindent In particular, the scalar-$U(1)$ and partial-$U(N)$ cases will be relevant for the equivalence theorems we prove, and we elaborate on these below. Notice that scalar-$U(1)$ $\subset$ partial-$U(N)$ $\subset$ full-$U(N)^D$, and that scalar-$U(1)$ $\subset$ repeated-$U(N)$ $\subset$ full-$U(N)^D$, but there is no inclusion 
relation
between partial-$U(N)$ and repeated-$U(N)$.

\paragraph{Scalar-$U(1)$ invariance.}
Assume that
\(
W(z\phi,(z\phi)^\dagger)=W(\phi,\phi^\dagger)
\)
for all $z\in U(1)$.  
Then $W$ can be rewritten as a function $V$ of the positive operator
$\phi\phi^\dagger$, so that $W(\phi,\phi^\dagger)=V(\phi\phi^\dagger)$. A standard argument proceeds in two steps:
\begin{enumerate}
    \item Orbital invariance: if $z\in U(1)$, then
    \(
    (z\phi)(z\phi)^\dagger= \phi\phi^\dagger
    \),
    since $z^*=z^{-1}$.  

    \item Orbital uniqueness: Consider a set of complex tensors $\ophi$ and $\ophi^\dagger$ of order-$D$.
    If $\ophi\ophi^\dagger=\phi\phi^\dagger$, then there exists $z\in U(1)$ such that $\ophi=z\phi$. Indeed, for any tensor $\psi$ the identity
    \(
    \phi(\phi^\dagger\psi)=\ophi(\ophi^\dagger\psi)
    \)
    implies $\ophi=z\phi$ for some $z\in\Cp$, and inserting this into
    $\ophi\ophi^\dagger=\phi\phi^\dagger$ gives $|z|=1$.
\end{enumerate}
This shows that $V$ is naturally defined on the image of the map $\phi\mapsto\phi\phi^\dagger$, i.e.\ on rank–$1$ positive self-adjoint tensors. Tietze's extension theorem guarantees that $V$ can be extended to all of ${\rm H}(\Cp_N^{\otimes D})$, although such an extension need not be unique.

The difference between a general $W(\phi,\phi^\dagger)$ and a potential of the form $V(\phi\phi^\dagger)$ is illustrated in Figure~\ref{fig:TWV}: in the former, $\phi$ and $\phi^\dagger$ may appear in arbitrary orderings or in different ammounts, while in the latter they are naturally paired into copies of $\phi\phi^\dagger$.

\begin{figure}[h]
\centering
    \begin{subfigure}[H]{0.19\textwidth}
    \centering
    \includegraphics[width=0.8\linewidth]{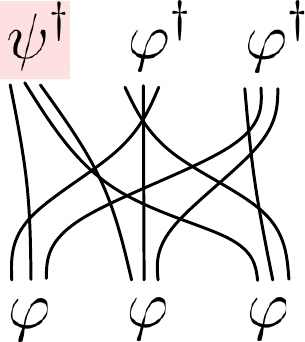}
    \caption{}
    \label{fig:symm1}
    \end{subfigure}
    \begin{subfigure}[H]{0.19\textwidth}
    \centering
    \includegraphics[width=0.8\linewidth]{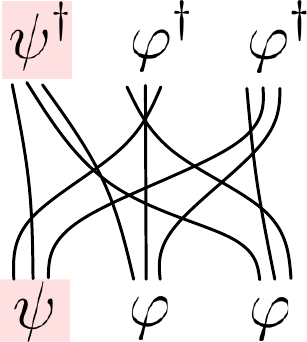}
    \caption{}
    \label{fig:symm2}
    \end{subfigure}
    \begin{subfigure}[H]{0.19\textwidth}
    \centering
    \includegraphics[width=0.8\linewidth]{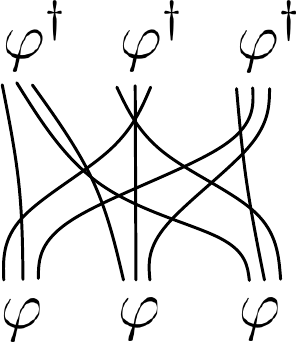}
    \caption{}
    \label{fig:symm3}
    \end{subfigure}
    \begin{subfigure}[H]{0.19\textwidth}
    \centering
    \includegraphics[width=0.8\linewidth]{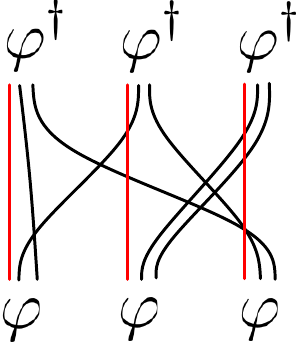}
    \caption{}
    \label{fig:symm4}
    \end{subfigure}
    \begin{subfigure}[H]{0.19\textwidth}
    \centering
    \includegraphics[width=0.8\linewidth]{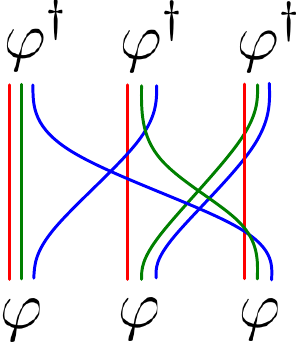}
    \caption{}
    \label{fig:symm5}
    \end{subfigure}
    \caption{We represent a few examples of potentials via stranded notation in the case $D=3$. Some examples include a constant tensor $\psi$ independent of $\phi$. Some strands may be colored if they exhibit a partial-$U(N)$ symmetry; the first, second and third indexes of $\phi$ are respectively represented in red, green and blue. The examples are: (a) A potential $W$ without scalar-$U(1)$ symmetry. 
    (b) A potential $V$ with scalar-$U(1)$ symmetry but without repeated-$U(N)$ symmetry. (c) A potential $V$ with repeated-$U(N)$ symmetry but without partial-$U(N)$ symmetry.
    (d) A potential $V$ with partial-$U(N)$ symmetry but without full-$U(N)^3$ symmetry. (e) A potential $V$ with full-$U(N)^3$ symmetry. 
    }
    \label{fig:TWV}
\end{figure}
\noindent In the scalar-$U(1)$ invariance setting,
the partition function we are interested in is
\begin{equation}\label{CTVscalar}
    {\cal Z}_{\rm CT}[V]
    (R)
    =\iint {\cal D}\phi^\dagger {\cal D}\phi\;e^{-\phi^\dagger 
    R^{-1}\phi+V(\phi{\phi^\dagger})}\;.
\end{equation}

\paragraph{Partial-$U(N)$ invariance.}
Assume that $W(\check U\phi,(\check U\phi)^\dagger)=W(\phi,\phi^\dagger)$ for any $
\check U:=U\otimes{\bb 1}^{\otimes(D-1)}$ with $U\in U(N)$.
Then, if $U(1)\subset U(N)$ invariance is included, $W$ can be rewritten as a function $V$ such that 
$W(\phi,\phi^\dagger)=V(\phi\phi^\dagger)$.
The additional invariance under $\check U$ implies that $V$ itself
factors  
in
the partial trace
${\rm Tr}_{(1)}(\phi\phi^\dagger)\in{\rm H}(\Cp_N^{\otimes(D-1)})$, and there exists a $\hat V$ such that 
\[
V(\phi\phi^\dagger)=\hat V({\rm Tr}_{(1)}(\phi\phi^\dagger))\;.
\]
The argument parallels the scalar-$U(1)$ case:
\begin{itemize}
    \item Orbital invariance:
    if $U\in U(N)$, then
    \(
    {\rm Tr}_{(1)} ((\check U\phi)(\check U\phi)^\dagger)
    ={\rm Tr}_{(1)} (\phi\phi^\dagger)
    \),
    which follows from\linebreak$U^\dagger=U^{-1}$ and the definition of the
    partial trace.

    \item Orbital uniqueness:
    if ${\rm Tr}_{(1)} (\psi\psi^\dagger)={\rm Tr}_{(1)} (\phi\phi^\dagger)$, then there
    exists $U\in U(N)$ such that $\psi=\check U\phi$.  
    One may see this by viewing $\phi$ and $\psi$ as $N\times N^{D-1}$
    matrices and using singular value decompositions
    \[
        \phi=(U_\phi\otimes {\hat U_\phi})S_\phi\;, \qquad
        \psi= (U_\psi\otimes {\hat U_\psi})S_\psi\;,
    \]
    with $U_\phi,U_\psi\in U(N)$, $\hat U_\phi,\hat U_\psi\in
    U(N^{D-1})$ and diagonal $S_\phi,S_\psi$.  
    The condition
    \(
    {\rm Tr}_{(1)}(\psi\psi^\dagger)={\rm Tr}_{(1)}(\phi\phi^\dagger)
    \)
    then implies that the singular values and the $\hat U$'s may be
    taken equal, and one obtains
    \(
    \psi=((U_\psi^{-1}U_\phi)\otimes{\bb 1}^{\otimes(D-1)})\phi\;.
    \)
\end{itemize}
This allows one to define $\hat V$ on the image of
$\phi\mapsto{\rm Tr}_{(1)}(\phi\phi^\dagger)$ via
\(
\hat V({\rm Tr}_{(1)}(\phi\phi^\dagger)):=V(\phi\phi^\dagger)
\),
and extend it to all of ${\rm H}(\Cp_N^{\otimes(D-1)})$, again by Tietze's theorem if needed. The situation is depicted schematically in Figure~\ref{fig:TWV}.

\paragraph{Potentials that have a trace-invariant expansion.}
Given these symmetries and motivated by the matrix case, we focus on potentials  $V(\phi\phi^\dagger)$ that decompose in the bubble/trace-invariant basis,
\begin{equation}
\label{eq:bubblebasispotential}
    V(\phi\phi^\dagger)
    =\sum_{n=1}^\infty \sum_{\bm\sigma\in\Sy_n^D}
      v_{[\bm\sigma]}\Tr_{[\bm\sigma]}(\phi\phi^\dagger)\;.
\end{equation}
This implies that its exponential admits a corresponding expansion in terms of bubbles,
\begin{equation}
    e^{V(\phi\phi^\dagger)}
    =\sum_{n=1}^\infty \sum_{\bm\sigma\in\Sy_n^D}
      w_{[\bm\sigma]}\Tr_{[\bm\sigma]}(\phi\phi^\dagger)\;.
\end{equation}
Substituting this into the partition function \eqref{ZCTs} yields
\begin{equation}
    \frac{{\cal Z}_{\rm CT}[V]
    (R)
    }{{\cal Z}_{\rm CT}[0]
    (R)
    }
    =\sum_{n=1}^\infty \sum_{\mu\in \Sy_n}\sum_{\bm\sigma\in\Sy_n^D}
      N^{-n}w_{[\bm\sigma]}\Tr_{[\mu \bm\sigma]}
      (R)
      \;.
\end{equation}
Moreover, in the factorized case $R=\bigotimes_{c}R_{c}$,
\begin{equation}
    \frac{{\cal Z}_{\rm CT}[V]
    (R)
    }{{\cal Z}_{\rm CT}[0]
    (R)
    }
    =\sum_{n=1}^\infty \sum_{\mu\in \Sy_n}\sum_{\bm\sigma\in\Sy_n^D}
      N^{-n}w_{[\bm\sigma]}\prod_{c}\Tr_{[\mu\sigma_c]}
      (R_{c})
      \;.
\end{equation}
These expressions are the tensor analogue of the matrix expansion in Section~\ref{subsec:GLmodel}. Each strand of color $c$ carries its own matrix $R_c$, and each color–$c$ face of length $l$ has a weight $\Tr(R_c^l)$. This is the natural starting point for introducing self-adjoint tensor models and for formulating a tensorial analogue of the dually weighted structures encountered in the matrix case.

\subsection{The self-adjoint two-tensor model}
\label{subsec:herm-tensor}
\noindent
Let $\Phi,\Psi\in\Hm(\mathbb C_N^{\otimes D})$ be random self-adjoint
tensors distributed according to the partition function
\begin{equation}
\label{HCT}
    {\cal Z}_{\rm HT}[V,Y]=\iint {\cal D}\Phi {\cal D}\Psi\;e^{-i \Tr\lp\Phi\Psi\rp+Y(\Psi)+V(\Phi)}\;.
\end{equation}
The potentials $V$ and $Y$ place us in the genuine tensor-model setting: unlike ordinary self-adjoint matrix models, here the natural invariants are those under the product group $U(N)^D\subset U(N^D)$, acting independently on each tensor factor
$\mathbb C_N$.
This symmetry imposes interactions to be of bubble type \eqref{eq:bubblebasispotential} and ensures that the Feynman diagrams retain their $D$-colored, multi-stranded structure\footnote{It is important to emphasize that enlarging the symmetry group reduces the effective order of the model. If the symmetry becomes $U(N^{D/2})^2\subset U(N^D)$, pairs of tensor indices are `locked' together and transform as composite indices; the theory then behaves as a matrix model. If one imposes the full $U(N^D)$ symmetry, all $D$ indices fuse into a single index, and the model degenerates into a vector model. Thus the degree of symmetry dictates the effective order of the theory: too much symmetry collapses a tensor model into a lower-order theory.}.

For any function $f$ of $\Phi,\Psi$, its expectation in this models is given by
\begin{equation}
    \langle f(\Phi,\Psi)\rangle[V,Y]=\frac{\iint {\cal D}\Phi {\cal D}\Psi\; f(\Phi,\Psi)\;e^{-i \Tr\lp\Phi\Psi\rp+Y(\Psi)+V(\Phi)}}{\iint {\cal D}\Phi {\cal D}\Psi\;e^{-i \Tr\lp\Phi\Psi\rp+Y(\Psi)+V(\Phi)}}\;.
\end{equation}

\paragraph{Propagator.}
If $Y=V=0$, one finds
\begin{equation}\label{HT00}
\langle \Phi\otimes\Psi\rangle[0,0]
    = - i \,\Sigma,\qquad
\langle \Phi\otimes\Phi\rangle[0,0]=0,\qquad
\langle \Psi\otimes\Psi\rangle[0,0]=0\,,
\end{equation}
where $\Sigma^{\bm a \bm c}_{\bm b \bm d}=\delta^{\bm a}_{\bm d}\delta^{\bm c}_{\bm b}$ as defined in Section~\ref{sec:notations}.
This means that the only non–vanishing propagators are those connecting $\Phi$ and $\Psi$. This propagator is illustrated in Figure~\ref{fig:PhiPsiprop}.
\begin{figure}[H]
    \centering
    \includegraphics[width=0.3
    \linewidth]{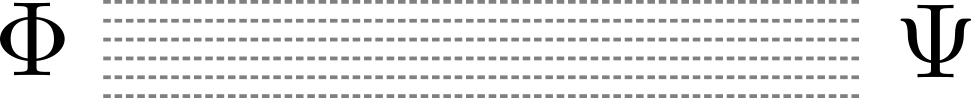}
    \caption{Stranded graph representation of the propagator of the theory defined by the partition function \eqref{HCT} as given in \eqref{HT00}. 
    }
    \label{fig:PhiPsiprop}
\end{figure}

Furthemore, if $V=0$ and $Y\neq 0$, 
\begin{equation}\label{HT00-again}
\langle \Phi\otimes\Psi\rangle[0,Y]
    = - i \,\Sigma,\qquad
\langle \Psi\otimes\Psi\rangle[0,Y]=0.
\end{equation}
However, the $\Phi\otimes\Phi$ propagator is corrected by second derivatives of $Y$,
\begin{equation}\label{eq:propderiv}
    \langle \Phi\otimes\Phi\rangle[0,Y]=-\left(\frac{\partial}{\partial \Psi}\otimes\frac{\partial}{\partial \Psi}\right)e^{Y(\Psi)}|_{\Psi=0}\;.
\end{equation}

\paragraph{Gaussian example.}
An illustrative case is the Gaussian choice
\(
Y_2(\Psi)=-\frac{1}{2}
\Tr((R\Psi)^2)
\;.
\)
In this case, it is convenient to rewrite
\begin{equation}
    Y_2(\Psi)=-\frac{1}{2}\Tr[\Sigma\cdot
    (R\otimes R)(\Psi\otimes\Psi)]\;,
\end{equation}
where the trace is on
$\Cp_N^{\otimes(2D)}\otimes \Cp_N^{*\otimes(2D)}$.  
Using this notation, \eqref{eq:propderiv} 
reads
\begin{equation}
    \langle \Phi\otimes\Phi\rangle[0,Y_2]=-\Sigma
    (R\otimes R)
    \;.
\end{equation}
The Gaussian potential $Y_2$ 
allows one to integrate $\Psi$ out
explicitly in~\eqref{HCT}, yielding
\begin{equation}\label{HCT2}
    {\cal Z}_{\rm HT}[V,Y_2]={\cal N}\int {\cal D}\Phi \;e^{-\frac{1}{2} \Tr\lp(
    R^{-1}\Phi)^2\rp+V(\Phi)}\;,
\end{equation}
with ${\cal N}=\int {\cal D}\Psi \;e^{-\frac{1}{2} \Tr\lp(
R
\Psi)^2\rp}$.  
Thus, the two-tensor model \eqref{HCT} with a Gaussian potential in $\Psi$
and an arbitrary $V(\Phi)$
effectively reduces to a self-adjoint one–tensor model \eqref{HCT2} in $\Phi$ 
with 
$V(\Phi)$.
This is directly analogous to the matrix-model case \eqref{eq:w2pf}.

\paragraph{Trace-invariant observables.}
For $V=0$ and general $Y$, the expectation value of  
tensor trace invariants may be expressed in terms of derivatives of $Y$.  
That is, for a multi–permutation $\bm \sigma\in\Sy_n^D$, 
\begin{equation}\label{eq:dpsi}
    \langle\Tr_{[{\bm\sigma}]}(\Phi)\rangle[0,Y]=i^{-n}\Tr_{[{\bm\sigma}]}\left ( \frac{\partial}{\partial \Psi} \right)e^{Y(\Psi)}|_{\Psi=0}\;.
\end{equation}
This formula provides an analytic handle on expectation values of trace invariants
purely in
terms of $Y$.

\paragraph{A logarithmic potential.} Of particular interest is the dually weighted logarithmic potential
\begin{equation}\label{eq:logpot}
Y_{\ln}
    [R](\Psi)=-\Tr\ln\lp {\bb 1}^{\otimes D}-i R\Psi\rp=\sum_{k=1}^{\infty}\frac{i^k}{k}\Tr((R\Psi)^k)
\end{equation}
for a fixed tensor $R\in \Cp_N^{\otimes D}\otimes\Cp_{N}^{*\otimes D}$. In this case the bubble expectations reproduce exactly the structure seen in the complex tensor model.

\begin{proposition}\label{avPhiPsi}
Let $\Phi\in{\rm H}(\Cp_N^{\otimes D})$ be distributed according to
$\mathcal{Z}_{\rm HT}[0,
Y_{\ln}
[R]]$.
For any $\bm \sigma\in S_n^D$,
\begin{equation}
\left\langle \Tr_{[\bm\sigma]}(\Phi)\right\rangle
[0,
Y_{\ln}
[R]]
    = \sum_{\mu\in S_n}
      \Tr_{[\mu\bm\sigma]}
      (R)
      \;.
\end{equation}
\end{proposition}
\begin{proof}
    The proof is given in Appendix~\ref{app:proof}.
\end{proof}

\noindent
The similarities between Proposition~\ref{avTM} and Proposition~\ref{avPhiPsi} indicate that the Gaussian ($W=0$ for 
Proposition~\ref{avTM}
and
$V=0$ for 
Proposition~\ref{avPhiPsi}) sectors of the complex and self-adjoint tensor models agree at the level of trace-invariant observables. In light of the matrix-model case, this suggests an analogous equivalence between the complex tensor model $\mathcal{Z}_{\rm CT}$ with covariance $R$ and the self-adjoint tensor model $\mathcal{Z}_{\rm HT}$ with logarithmic potential \eqref{eq:logpot}. The next subsection shows that this equivalence extends to the partition functions $\mathcal{Z}_{\rm CT}$ and $\mathcal{Z}_{\rm HT}$ for any potential $V$.

\subsection{Exact tensor equivalences via intermediate fields}\label{sec:exacttensorequiv}

\noindent
In this section, we show that this agreement extends to the full ($V \ne 0$)
partition functions and to all observables (i.e.\ expectation values)
depending only on the self-adjoint combination $\phi\phi^\dagger$. Similarly to the matrix model equivalence, this is done through an intermediate field representation.
In this formulation, the self-adjoint tensor $\Phi$ appears as an
{\it intermediate effective field}, while a second self-adjoint tensor $\Psi$ appears as a {\it dually-weighted intermediate field} in which the logarithmic interaction, deformed by the covariance inherited from the propagator of the complex model, gives rise to the characteristic dually weighted structure.

\subsubsection{Equivalence for a scalar-\texorpdfstring{$U(1)$}{U(1)} invariant potential \texorpdfstring{$V(\phi {\phi^\dagger})$}{V(phi dagger)}}

\noindent Assuming a scalar-$U(1)$ invariance, we consider a complex tensor
$\phi\in\Cp_N^{\otimes D}$ distributed according to the partition function
\begin{equation}
\label{ZCT}
    {\cal Z}_{\rm CT}[V](R)
    =\iint {\cal D}\phi^\dagger {\cal D}\phi\;
    e^{-\phi^\dagger R^{-1} \phi+ V(\phi\phi^\dagger)}\;,
\end{equation}
with the complex field $\phi$ carrying the potential $V(\phi\phi^\dagger)$. We also consider self-adjoint tensors
$\Phi,\Psi\in\Hm(\mathbb C_N^{\otimes D})$ distributed according to
\begin{equation}
\label{eq:scalarU1_self-adjoint}
     {\cal Z}_{\rm HT}[V,Y_{\ln}[R]]=\iint {\cal D}\Phi {\cal D}\Psi\;e^{-i \Tr\lp\Phi\Psi\rp+Y_{\ln}[R](\Psi)+V(\Phi)}\;,
    \quad
{Y}_{\ln}[R](\Psi)
    =-\Tr\ln({\bb 1}^{\otimes D}-iR\Psi)\;.
\end{equation}
In this formulation, $\Phi$ plays the role of an {\it intermediate effective field}, carrying an effective potential $V(\Phi)$,
while $\Psi$ appears as a {\it dually-weighted intermediate field}, carrying a dually-weighted logarithmic potential $Y_{\ln}[R](\Psi)$.
The following theorem establishes the equivalence between these partition functions.

\begin{theorem}\label{mainthm}
   Consider the scalar-$U(1)$–invariant  $\Cp_N^{\otimes D}$ random tensor model with partition function \eqref{ZCT}
   and the ${\rm H}(\Cp_N^{\otimes D})$ random tensor partition function $\mathcal{Z}_{\rm HT}[V,Y_{\ln}[R]]$ \eqref{eq:scalarU1_self-adjoint}. For any potential function $V$, 
\begin{equation}
    \frac{\mathcal{Z}_{\rm CT}[V]
    (R)
    }{\mathcal{Z}_{\rm CT}[0]
    (R)}=\frac{\mathcal{Z}_{\rm HT}[V,
Y_{\ln}
    [R]
    ]}{\mathcal{Z}_{\rm HT}[0,
Y_{\ln}
    [R]
    ]}\;
\end{equation}
as formal power series.
\end{theorem}

\begin{proof}
For any tensor $\phi\in \Cp_N^{\otimes D}$, the tensor $\phi\phi^\dagger$
is self-adjoint, $\phi\phi^\dagger \in \Hm(\Cp_N^{\otimes D})$.  
We insert into the partition function a delta function over self-adjoint
tensors, rewriting the potential as a function of $\Phi$:
\begin{equation}
    {\cal Z}_{\rm CT}[V]
    (R)
    =\iint {\cal D}\phi^\dagger {\cal D}\phi\int {\cal D}\Phi\;\delta(\phi\phi^\dagger-\Phi)\;e^{-\phi^\dagger
    {R}^{-1}
    \phi+V(\Phi)}\;.
\end{equation}
The delta function admits a Fourier representation over
$\Hm(\Cp_N^{\otimes D})$,
\begin{equation}
    \delta(\phi\phi^\dagger-\Phi)={\cal N}_{\rm HT}^{-1}\int {\cal D}\Psi\;e^{i\Tr \lp\Psi\lp\phi\phi^\dagger-\Phi\rp\rp}\;,
\end{equation}
where
\[
{\cal N}_{\rm HT}=\iint {\cal D}\Phi {\cal D}\Psi\;e^{-i\Tr \lp\Psi\Phi\rp}={\cal Z}_{\rm HT}[0,0]\;.
\]
Substituting this representation yields
\begin{equation}
    {\cal Z}_{\rm CT}[V]
    (R)
    ={\cal N}_{\rm HT}^{-1}\iint {\cal D}\phi^\dagger {\cal D}\phi\iint {\cal D}\Phi {\cal D}\Psi\;e^{-\phi^\dagger
    {R}^{-1}
    \phi+i\Tr \lp\Psi\lp\phi\phi^\dagger-\Phi\rp\rp+V(\Phi)}\;.
\end{equation}
Using $\Tr\lp\Psi\phi\phi^\dagger\rp=\phi^\dagger\Psi\phi$,
we rewrite
the exponent as
\begin{equation}\label{intpftens}
    {\cal Z}_{\rm CT}[V]
    (R)
    ={\cal N}_{\rm HT}^{-1}\iint {\cal D}\phi^\dagger {\cal D}\phi\iint {\cal D}\Phi {\cal D}\Psi\;e^{-\phi^\dagger \lp 
    R^{-1}
    -i\Psi\rp\phi-i\Tr \lp\Psi\Phi\rp+V(\Phi)}\;.
\end{equation}
For sufficiently well-behaved interactions $V$, one may exchange the order of
integration and perform the integral over $\phi$ first. More generally, for any $V$, this
operation can always be performed at the level of formal power series. The integral over
$\phi$ is Gaussian, and can therefore be computed explicitly,
\begin{equation}
    \iint {\cal D}\phi^\dagger {\cal D}\phi\;e^{-\phi^\dagger\lp
    {R}^{-1}
    -i\Psi\rp\phi}={\cal N}_{\rm CT}\det({\bb 1}^{\otimes D}-i
    R
    \Psi)^{-1}\;,
\end{equation}
with
\(
{\cal N}_{\rm CT}=\iint {\cal D}\phi^\dagger {\cal D}\phi\;e^{-\phi^\dagger 
R^{-1}
\phi}={\cal Z}_{\rm CT}[0]
(R)
\;.
\)
Using $\det=e^{\Tr\ln}$ we obtain
\begin{equation}
    {\cal Z}_{\rm CT}[V]
    (R)
    ={\cal N}_{\rm HT}^{-1}{\cal N}_{\rm CT} \iint {\cal D}\Phi {\cal D}\Psi\;e^{-\Tr\ln\lp {\bb 1}^{\otimes D}-i
    R
    \Psi\rp-i\Tr \lp\Psi\Phi\rp+V(\Phi)}\;.
\end{equation}
The remaining integral is precisely ${\cal Z}_{\rm HT}[V,
Y_{\ln}
[R]]$, hence
\begin{equation}
    {\cal Z}_{\rm CT}[V]
    (R)
    ={\cal N}_{\rm HT}^{-1}{\cal N}_{\rm CT}{\cal Z}_{\rm HT}
    [V,
Y_{\ln}
    [R]]
    \;.
\end{equation}
Since $
Y_{\ln}
[R](0)=0$, one has ${\cal N}_{\rm HT}={\cal Z}_{\rm HT}[0,
Y_{\ln}
[R]]$, and the claim follows.
\end{proof}
\medskip
\noindent
\emph{Remark.}  As in the matrix case, the equivalence is carried at the
level of formal power series, since for general choices of interaction the
tensor integrals may fail to be absolutely convergent or may only be defined
in a distributional sense. For suitable potentials, however, the
equivalence extends to well-defined tensor integrals.

One explicit example is given by mapping the sextic complete interaction model to a cubic interaction model as shown in Figure~\ref{fig:excomp}, which includes the vertices shown in Figure~\ref{fig:exvert}. The sextic model is explored further in \eqref{eq:ZCTV6}.
\begin{figure}[H]
    \begin{subfigure}[c]{0.24\textwidth}
    \centering
    \includegraphics[width=0.5
    \linewidth,angle=0]{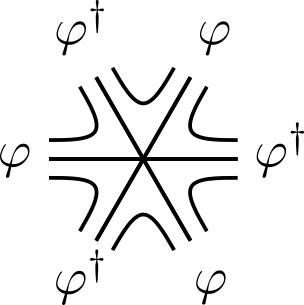}
    \caption{}
    \label{fig:exvert1}
    \end{subfigure}
    \begin{subfigure}[c]{0.24\textwidth}
    \centering
    \includegraphics[width=0.5
    \linewidth,angle=0]{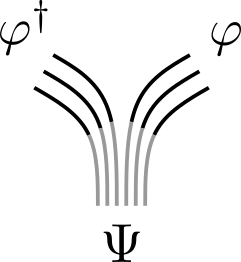}
    \caption{}
    \label{fig:exvert3}
    \end{subfigure}
    \begin{subfigure}[c]{0.24\textwidth}
    \centering
    \includegraphics[width=0.5
    \linewidth,angle=0]{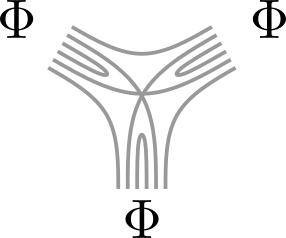}
    \caption{}
    \label{fig:exvert2}
    \end{subfigure}
    \begin{subfigure}[c]{0.24\textwidth}
    \centering
    \includegraphics[width=0.5
   \linewidth,angle=0]{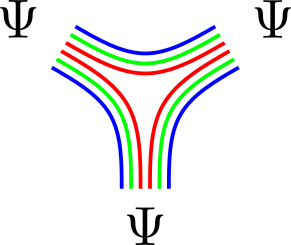}
    \caption{}
    \label{fig:exvert4}
    \end{subfigure}
    \caption{(a)
    A sextic interaction vertex that can be present in the potential $W$ of the partition function \eqref{ZCTs}. 
    (b)
    The $i \phi^\dagger \Psi \phi$ vertex in \eqref{intpftens}.
    (c)
    An example of a cubic potential of 
the intermediate effective field 
    $\Phi$ which may be present in $V$ in \eqref{eq:HTpf}.
    (d)
The cubic potential of the 
dually-weighted intermediate field $\Psi$, which appears as the third order term in the expansion of the potential $\Tr\ln\lp {\bb 1}^{\otimes D}-i
    R \Psi\rp$ in \eqref{eq:HTpf}.
}
    \label{fig:exvert}
\end{figure}

\begin{figure}[H]
    \begin{subfigure}[c]{0.499\textwidth}
    \begin{minipage}[c][6cm][c]{\linewidth}
    \centering
    \includegraphics[width=0.5
    \linewidth,angle=0]{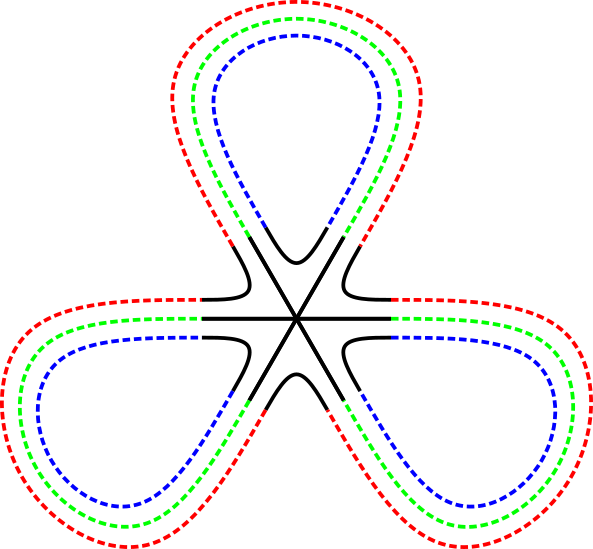}
    \end{minipage}
    \caption{}
    \label{fig:excompc}
    \end{subfigure}
    \begin{subfigure}[c]{0.499\textwidth}
    \begin{minipage}[c][6cm][c]{\linewidth}
    \centering
    \includegraphics[width=0.65
    \linewidth,angle=0]{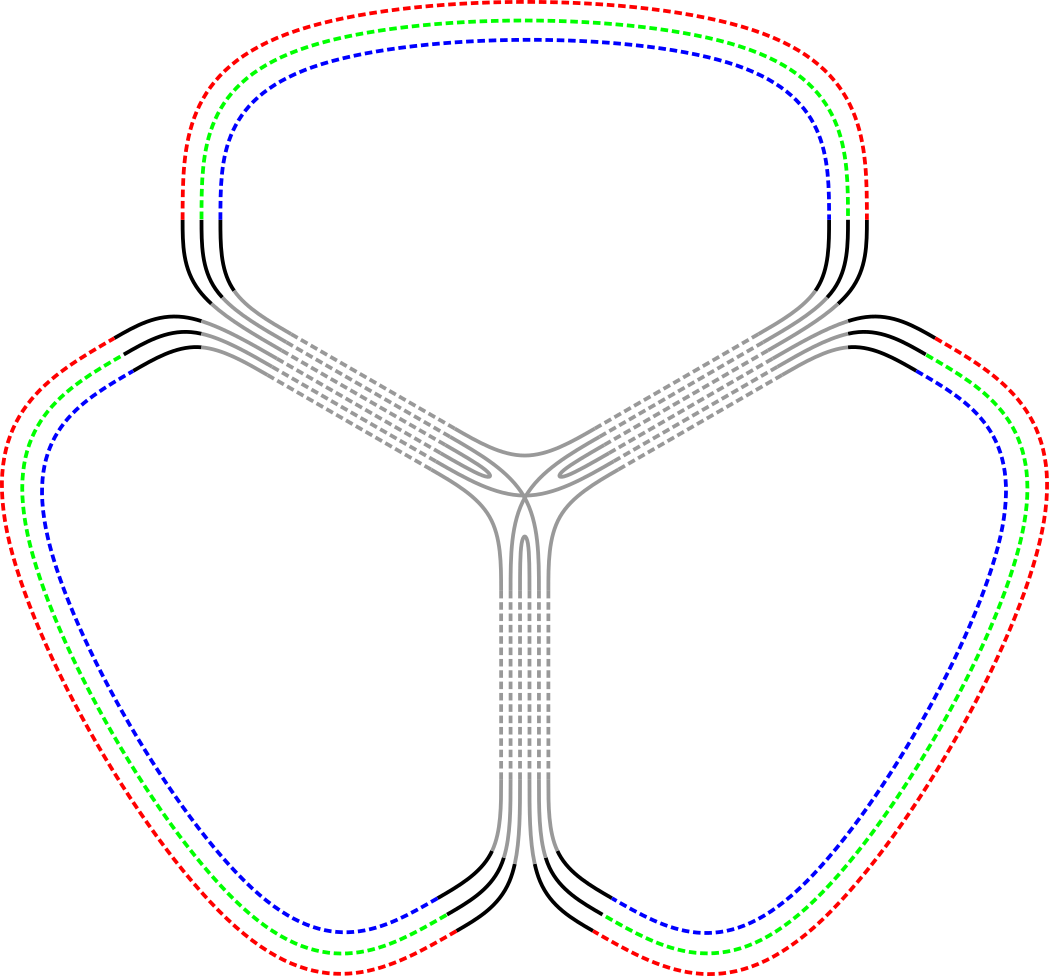}
    \end{minipage}
    \caption{}
    \label{fig:excompcs1}
    \end{subfigure}
    \begin{subfigure}[c]{0.499\textwidth}
    \begin{minipage}[c][6cm][c]{\linewidth}
    \centering
    \includegraphics[width=0.5
    \linewidth,angle=0]{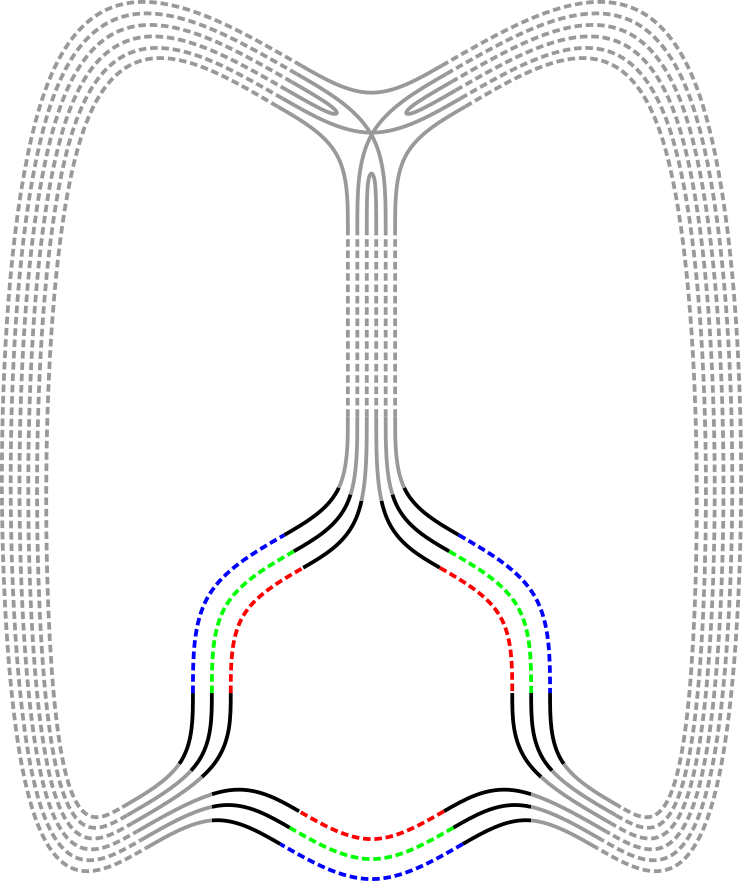}
    \end{minipage}
    \caption{}
    \label{fig:excompcs2}
    \end{subfigure}
    \begin{subfigure}[c]{0.499\textwidth}
    \begin{minipage}[c][6cm][c]{\linewidth}
    \centering
    \includegraphics[width=0.4
    \linewidth,angle=0]{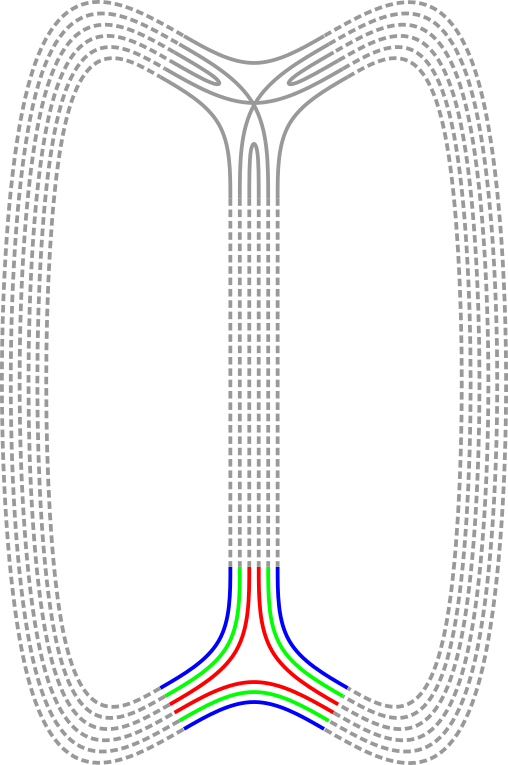}
    \end{minipage}
    \caption{}
    \label{fig:excomps}
    \end{subfigure}
    \caption{An example of graphs whose amplitudes are equal in the different tensor models. Those graphs are composed by the vertices shown in Figure~\ref{fig:exvert} and the propagators shown in Figure~\ref{fig:tensprop} and Figure~\ref{fig:PhiPsiprop}. (a) A graph that can be generated by the partition function \eqref{ZCTs}; (b-c)
    Two representations of the same graph that can be generated by the partition function \eqref{intpftens}; (d) A graph that can be generated by the partition function \eqref{eq:HTpf}.   
}
    \label{fig:excomp}
\end{figure}

\paragraph{Expectation values.}

As in the matrix case, the equivalence of the partition functions
directly implies an equivalence between observables  (i.e. expectation values).  
For any function $f$,
\[
\left\langle f(\phi\phi^\dagger)\right\rangle[V]
(R)
    = \frac{\mathcal{Z}_{\rm CT}
    (R)
    [V+\ln f]}
           {\mathcal{Z}_{\rm CT}[V]
           (R)
           },\qquad
\left\langle f(\Psi)\right\rangle[V,Y
[R]
]
    = \frac{\mathcal{Z}_{\rm HT}[V+\ln f,\,
Y_{\ln}
    [R]
    ]}
           {\mathcal{Z}_{\rm HT}[V,
Y_{\ln}
           [R]
           ]}\;,
\]
which leads to the following corollary.

\begin{corollary}\label{cor:exp_val_tm}
For any potential $V$ and any trace invariant $f:{\rm H}(\Cp_N^{\otimes D})\to\mathbb C$,
\[
\left\langle f(\phi\phi^\dagger)\right\rangle[V]
(R)
    = \left\langle f(\Phi)\right\rangle[V,
Y_{\ln}
    [R]
    ].
\]
\end{corollary}

\noindent In particular, if $V=0$ and $f$ is a  trace invariant of tensors, one recovers Propositions~\ref{avTM} and~\ref{avPhiPsi} as special cases.

\medskip

\noindent
The arguments above establish an equivalence between random $\Cp_N^{\otimes D}$ 
tensor models \eqref{CTVscalar} and random $\Hm(\Cp_N^{\otimes D})$ 
tensor models  \eqref{HCT} with logaritmic potential \eqref{eq:logpot}. Moreover, under the additional partial-$U(N)$ invariance, one can further reduce the self-adjoint side to tensors of order $2(D-1)$ as we now describe.

\subsubsection{Equivalence for partial-\texorpdfstring{$U(N)$}{U(N)} invariant potential \texorpdfstring{$\hat V({\rm Tr}_{(1)}(\phi{\phi^\dagger}))$}{hat V(Tr1 phi phidagger)}
}
\noindent Assuming partial-$U(N)$ invariance, we consider a complex tensor
$\phi\in\Cp_N^{\otimes D}$ distributed according to the partition function
\begin{equation}
\label{ZCT1}
    {\cal Z}_{\rm CT}[\hat V\circ {\rm Tr}_{(1)}](R)
    =\iint {\cal D}\phi^\dagger {\cal D}\phi\;
    e^{-\phi^\dagger R^{-1} \phi+\hat V({\rm Tr}_{(1)}(\phi\phi^\dagger))}\;.
\end{equation}
We also consider self-adjoint tensors
$\hat\Phi,\hat\Psi\in\Hm(\mathbb C_N^{\otimes (D-1)})$ distributed according to
\begin{equation}
\label{eq:partialUN_self-adjoint}
    {\cal Z}_{\rm \hat HT}[\hat V,\hat{Y}_{\ln}[R]]
    =\iint d\hat\Phi d\hat\Psi\;
    e^{-i \Tr(\hat\Psi\hat\Phi)
    +\hat{Y}_{\ln}[R](\hat\Psi)
    +\hat V(\hat\Phi)}\,,
    \quad
\hat{Y}_{\ln}[R](\hat\Psi)
    =-\Tr\ln({\bb 1}^{\otimes D}-iR({\bb 1}\otimes \hat\Psi))\;.
\end{equation}
In this formulation, $\hat\Phi$ plays the role of an {\it intermediate effective field},
while $\hat\Psi$ appears as a {\it dually-weighted intermediate field}.
The following theorem establishes the equivalence between these partition functions.

\begin{theorem}
\label{theorem:partialtensor}
   Consider the $\Cp_N^{\otimes D}$ random tensor model defined by
   \eqref{ZCT1},
   and the ${\rm H}(\Cp_N^{\otimes (D-1)})$ random tensor model defined by \eqref{eq:partialUN_self-adjoint}.
   For any potential function $\hat{V}$, 
\begin{equation}
    \frac{\mathcal{Z}_{\rm CT}[\hat V\circ{\rm Tr}_{(1)}]
    (R)
    }{\mathcal{Z}_{\rm CT}[0]
    (R)}=\frac{\mathcal{Z}_{\rm \hat HT}[\hat V,
\hat{Y}_{\ln}
    [R]
    ]}{\mathcal{Z}_{\rm \hat HT}[0,
\hat{Y}_{\ln}
    [R]
    ]}\;
\end{equation}
as formal power series.
\end{theorem}

\begin{proof}
For any $\phi\in \Cp_N^{\otimes D}$, the partial trace
${\rm Tr}_{(1)}(\phi\phi^\dagger)$ is self-adjoint,
${\rm Tr}_{(1)}(\phi\phi^\dagger) \in \Hm(\Cp_N^{\otimes (D-1)})$.  
We insert a delta function over $\Hm(\Cp_N^{\otimes (D-1)})$,
\begin{equation}
    {\cal Z}_{\rm CT}[\hat V\circ {\rm Tr}_{(1)}]
    (R)
    =\iint {\cal D}\phi^\dagger {\cal D}\phi\int {\cal D}\hat\Phi\;\delta({\rm Tr}_{(1)}(\phi\phi^\dagger)-\hat\Phi)\;e^{-\phi^\dagger
    {R}^{-1}
    \phi+\hat V(\hat\Phi)}\;.
\end{equation}
The delta function is represented as an integral over
$\hat\Psi\in \Hm(\Cp_N^{\otimes (D-1)})$,
\begin{equation}
    \delta({\rm Tr}_{(1)}(\phi\phi^\dagger)-\hat\Phi)={\cal N}_{\rm \hat HT}^{-1}\int {\cal D}\hat\Psi\;e^{i\Tr \lp\hat\Psi\lp {\rm Tr}_{(1)}(\phi\phi^\dagger)-\hat\Phi\rp\rp}\;,
\end{equation}
with
\[
{\cal N}_{\rm \hat HT}=\iint {\cal D}\hat\Phi {\cal D}\hat\Psi\;e^{-i\Tr \lp\hat\Psi\hat\Phi\rp}={\cal Z}_{\rm HT}[0,0]\;.
\]
This gives
\begin{equation}
    {\cal Z}_{\rm CT}[\hat V\circ {\rm Tr}_{(1)}]
    (R)
    ={\cal N}_{\rm \hat HT}^{-1}\iint {\cal D}\phi^\dagger {\cal D}\phi\iint {\cal D}\hat\Phi {\cal D}\hat\Psi\;e^{-\phi^\dagger
    {R}^{-1}
    \phi+i\Tr \lp\hat\Psi\lp{\rm Tr}_{(1)}(\phi\phi^\dagger)-\hat\Phi\rp\rp+V(\hat\Phi)}\;.
\end{equation}
Using
\(
\Tr\lp\hat\Psi{\rm Tr}_{(1)}(\phi\phi^\dagger)\rp=\phi^\dagger({\bb 1}\otimes\hat\Psi)\phi
\),
we rewrite the exponent and
\begin{equation}
    {\cal Z}_{\rm CT}[\hat V\circ {\rm Tr}_{(1)}]
    (R)
    ={\cal N}_{\rm \hat HT}^{-1}\iint {\cal D}\phi^\dagger {\cal D}\phi\iint {\cal D}\hat\Phi {\cal D}\hat\Psi\;e^{-\phi^\dagger \lp 
    R^{-1}
    -i\,{\bb 1}\otimes\hat\Psi\rp\phi-i\Tr \lp\hat\Psi\hat\Phi\rp+V(\hat\Phi)}\;.
\end{equation}
For sufficiently well-behaved interactions, one may exchange the order of
integration and perform the integral over $\phi$ first. More generally, this
step can be understood at the level of formal power series. The integral over
$\phi$ is Gaussian, and can therefore be computed explicitly,
\begin{equation}
    \iint {\cal D}\phi^\dagger {\cal D}\phi\;e^{-\phi^\dagger\lp
    {R}^{-1}
    -i\,{\bb 1}\otimes\hat\Psi\rp\phi}={\cal N}_{\rm CT}\det({\bb 1}^{\otimes D}-i
    R
    ({\bb 1}\otimes\hat\Psi))^{-1}\;,
\end{equation}
where
\(
{\cal N}_{\rm CT}=\iint {\cal D}\phi^\dagger {\cal D}\phi\;e^{-\phi^\dagger 
R^{-1}
\phi}={\cal Z}_{\rm CT}[0]
(R)
\;.
\)
Using $\det=e^{\Tr\ln}$, we obtain
\begin{equation}
    {\cal Z}_{\rm CT}[\hat V\circ {\rm Tr}_{(1)}]
    (R)
    ={\cal N}_{\rm \hat HT}^{-1}{\cal N}_{\rm CT} \iint {\cal D}\hat\Phi {\cal D}\hat\Psi\;e^{-\Tr\ln\lp {\bb 1}^{\otimes D}-i
    R
    ({\bb 1}\otimes\hat\Psi)\rp-i\Tr \lp\hat\Psi\hat\Phi\rp+V(\hat\Phi)}\;.
\end{equation}
The final integral is ${\cal Z}_{\rm \hat HT}[\hat V,
\hat{Y}_{\ln}
[R]]$, so
\begin{equation}
    {\cal Z}_{\rm CT}[\hat V\circ {\rm Tr}_{(1)}]
    (R)
    ={\cal N}_{\rm \hat HT}^{-1}{\cal N}_{\rm CT}{\cal Z}_{\rm \hat HT}[\hat V,
\hat{Y}_{\ln}
    [R]
    ]\;.
\end{equation}
Since $
\hat{Y}_{\ln}
[R](0)=0$, one has ${\cal N}_{\rm \hat HT}={\cal Z}_{\rm \hat HT}[0,
\hat{Y}_{\ln}
[R]]$, which completes the proof.
\end{proof}
\medskip
\noindent
\emph{Remark.}
The same considerations as in the matrix case and in
Theorem~\ref{mainthm} apply: the equivalence is understood at the level of
formal power series in general, and extends to convergent tensor integrals
whenever the latter are well-defined.

\paragraph{Matrix equivalence as a particular tensor equivalence.}

For clarity of exposition, the equivalence for random matrix models
established in Section~\ref{sec:matrixequiv} was presented separately
from the tensor equivalence proved in Section~\ref{sec:equivtensor}.
Nevertheless, the matrix result is in fact a special case of the tensor
equivalence, corresponding to tensors of order $D=2$. To see this explicitly, we rewrite the partition function
\eqref{eq:ZCM} of the complex matrix model as
\begin{equation} \label{eq:matrixastensor}
\mathcal{Z}_{\rm CM}[V](P,Q) = \iint_{\mathbb C^{N\times N}} {\cal D}M^\dagger\,{\cal D}M\; e^{-N\sum_{ijkl}M^*_{ij} Q^{-1}_{ik} P^{-1}_{lj} M_{kl} + V([\sum_{k}M_{ki}^*M_{kj}]_{ij})}\;. 
\end{equation}
Interpreting $M$ as an order–$2$ complex tensor,
$M\in\mathbb C^{\otimes 2}_N$, the quadratic term is precisely of the
form $\phi^\dagger R^{-1}\phi$ appearing in Theorem~\ref{theorem:partialtensor},
with factorized propagator
$
R = Q \otimes P^{\tp}.
$
Moreover, the interaction depends only on the partial contraction
$M^\dagger M$, which coincides with the partial trace
$\Tr_{(1)}(\phi\phi^\dagger)$ in the tensor notation. Therefore, the matrix equivalence follows directly from
Theorem~\ref{theorem:partialtensor} upon specializing to $D=2$, and the matrix intermediate field corresponds to a
self-adjoint tensor of order $2(D-1)=2$.
The self-adjoint matrix model obtained in Section~\ref{sec:matrixequiv}
is simply the order-$2$ instance of the general self-adjoint tensor
formulation.

\subsection{Examples}
\label{sec:examplestensor}

\noindent In the following, we illustrate the tensor equivalence with several representative examples.
These include interactions of different orders and symmetry types, and demonstrate how
complex tensor models can be mapped to self-adjoint models of equal or reduced effective
tensor order, depending on the structure of the interaction.

\paragraph{Factorized propagators and causal examples.}
In the partial-$U(N)$ setting, if the propagator factorizes as $R=Q\otimes \hat P$ with $Q\in \Cp_N\otimes \Cp_N^*$ and $\hat P\in \Cp_N^{\otimes(D-1)}\otimes \Cp_N^{*\otimes(D-1)}$, the logarithmic potential \eqref{eq:partialUN_self-adjoint}
simplifies to
\begin{equation}\label{eq:factex}
    {\cal Z}_{\rm \hat HT}
    [\hat V,
\hat{Y}_{\ln}
    [R]]
    =\iint {\cal D}\hat\Phi {\cal D}\hat\Psi\;e^{-i \Tr(\hat\Psi\hat\Phi)+\sum_{k=1}^\infty \frac{i^k}{k}\Tr\lp Q^k\rp\Tr\lp(\hat P\hat\Psi)^k\rp+\hat V(\hat\Phi)}\;,
\end{equation}
with $\hat\Phi,\hat\Psi\in\Hm(\mathbb C_N^{\otimes (D-1)})$.
Due to Theorem \ref{theorem:partialtensor}, this partition function is equivalent to
\begin{equation}
    {\cal Z}_{\rm CT}[V](Q\otimes \hat P)
    =\iint {\cal D}\phi^\dagger {\cal D}\phi\;e^{-\phi^\dagger (Q\otimes \hat P)^{-1}\phi+V(\phi {\phi^\dagger})}\;.
\end{equation}
with $\phi\in\Cp_N^{\otimes D}$. Moreover, if $Q=C_2$ (see \eqref{cnprop}), one obtains the Gaussian reduction
\begin{equation}
    {\cal Z}_{\rm \hat HT}
    [\hat V,
\hat{Y}_{\ln}
    [R]]
    =\iint {\cal D}\hat\Phi {\cal D}\hat\Psi\;e^{-i \Tr(\hat\Psi\hat\Phi)-N\frac{1}{2}\Tr\lp(\hat P\hat\Psi)^2\rp+\hat V(\hat\Phi)}\;,
\end{equation}
with $\hat\Phi,\hat\Psi\in\Hm(\mathbb C_N^{\otimes (D-1)})$. Equivalently, after integrating out $\hat\Psi$,
\begin{equation}\label{eq:c2case}
    {\cal Z}_{\rm \hat HT}
    [\hat V,
\hat{Y}_{\ln}
    [R]]
    ={\cal N}\int {\cal D}\hat\Phi \;e^{-N\frac{1}{2}\Tr\lp(\hat P^{-1}\hat\Phi)^2\rp+\hat V(\hat\Phi)}\;,
\end{equation}
with
\(
{\cal N}=\int{\cal D}\hat\Psi \;e^{-N\frac{1}{2}\Tr((\hat P\hat \Psi)^2)}
\).
Due to Theorem \ref{theorem:partialtensor},
this model is equivalent to the complex tensor model
\begin{equation}
    {\cal Z}_{\rm CT}[V](C_2\otimes \hat P)
    =\iint {\cal D}\phi^\dagger {\cal D}\phi\;e^{-\phi^\dagger (C_2\otimes \hat P)^{-1}\phi+V(\phi {\phi^\dagger})}\;,
\end{equation}
with $\phi\in\Cp_N^{\otimes D}$. This provides a tensorial analogue of the causal matrix example presented in Section~\ref{examp}.

\paragraph{Equivalences for trace-invariant potentials.}
Consider a complex tensor model whose potential is given by any trace-invariant 
\[
\label{eq:excomplexscalarU1potential}
V_{[\bm \alpha]}(\phi\phi^\dagger)
=\Tr_{[\bm{\alpha}]}(\phi\phi^\dagger)
=\sum_{\bm k}\prod_{j=1}^n
\phi^\dagger_{\bm k_j}\,
\phi^{(\bm{\alpha}_*\bm k)_j},
\]
where  multi-permutation $\bm \alpha \in \Sy_n^D$ specifies the trace-invariant.
A natural choice for the corresponding self-adjoint model, with
$\Phi \in \Hm(\Cp_N^{\otimes D})$, is the trace-invariant potential dictated by the same multi-permutation $\bm \alpha$
\[
\label{eq:exsascalarU1potential}
V_{[\bm \alpha]}(\Phi)
=\Tr_{[\bm{\alpha}]}(\Phi)
=\sum_{\bm k}\prod_{j=1}^n
\Phi_{\bm k_j}^{(\bm{\alpha}_*\bm k)_j}.
\]
Then, the pair of complex and self-adjoint models with potentials
\eqref{eq:excomplexscalarU1potential} and \eqref{eq:exsascalarU1potential}
is an example 
in which a trace-invariant interaction of an order-$D$
complex tensor model is equivalent to a trace-invariant interaction of a self-adjoint
tensor model of order $2D$, according to 
Theorem~\ref{mainthm}. This choice of effective corresponding potential above \eqref{eq:exsascalarU1potential}, however, is not unique. Indeed, in the complex model the pairing between
$\phi$ and $\phi^\dagger$ is defined only up to left multiplication of
$\bm \alpha$ by an element $\mu\in\Sy_n$, since
\[
\Tr_{[\mu\bm{\alpha}]}(\phi\phi^\dagger)
=\Tr_{[\bm{\alpha}]}(\phi\phi^\dagger),
\qquad \forall\,\mu\in\Sy_n.
\]
In contrast, for the self-adjoint field $\Phi$, different representatives of the same
left-equivalence class may define inequivalent trace invariants, i.e.
$\Tr_{[\mu\bm{\alpha}]}(\Phi)\neq\Tr_{[\bm{\alpha}]}(\Phi)$ in general. This freedom can be exploited to reach a particularly simple representative, which we demonstrate below.

Writing $\bm \alpha=(\alpha_1,\bm \alpha_{\hat 1})$ with
$\bm \alpha_{\hat 1}\in\Sy_n^{D-1}$, and choosing $\mu=\alpha_1^{-1}$, one finds
\begin{equation}
\Tr_{[\bm{\alpha}]}(\phi\phi^\dagger)
=\Tr_{[\alpha_1^{-1}\bm{\alpha}_{\hat 1}]}
\!\left(\Tr_{(1)}(\phi\phi^\dagger)\right),
\end{equation}
where $\Tr_{(1)}$ denotes the partial trace over the first tensor index. As a consequence, the interaction can be equivalently expressed in terms of a
self-adjoint tensor
$\hat\Phi\in\Hm(\Cp_N^{\otimes(D-1)})$, leading to the reduced trace-invariant potential
\[
\label{eq:exsaspartialUNpotential}
\hat V_{[\bm \alpha]}(\hat\Phi)
=\Tr_{[\alpha_1^{-1}\bm{\alpha}_{\hat 1}]}(\hat\Phi)
=\sum_{\bm k_{\cdot\hat 1}}
\prod_{j=1}^n
\hat\Phi_{\bm k_{j\hat 1}}^{\big((\alpha_1^{-1}\bm{\alpha}_{\hat 1})_*\bm k\big)_{j\hat 1}}.
\]
The self-adjoint tensor model with this potential above \eqref{eq:exsaspartialUNpotential} together with the complex counterpart given in \eqref{eq:excomplexscalarU1potential}, then provides an example in which a trace-invariant interaction of an order-$D$
complex tensor model is equivalent to a trace-invariant interaction of a self-adjoint
tensor model of order $2(D-1)$, according to Theorem~\ref{theorem:partialtensor}.

More specific examples of these trace-invariant potentials 
can be found in
\begin{itemize}
\item[(i)]
\eqref{eq:pillowpotential}  and \eqref{eq:effpillowpotential} for the complex model potential \eqref{eq:excomplexscalarU1potential} and the corresponding self-adjoint potential of scalar-$U(1)$ symmetry \eqref{eq:exsascalarU1potential} respectively, 
and
\item[(ii)]
\eqref{eq:lowpillowpotential} and \eqref{eq:efflowpillowpotential} for the complex model potential \eqref{eq:excomplexscalarU1potential} and the corresponding self-adjoint model of partial-$U(N)$ symmetry \eqref{eq:exsaspartialUNpotential} respectively. 
\end{itemize}
See also the corresponding Figures  for (i) Figures \ref{fig:pillows} and \ref{fig:pillows_sa}, and (ii) Figures  \ref{l17} and \ref{l17e} for the graphic illustrations of these trace-invariant potentials.

\paragraph{ (i) From an order–$D$ quartic pillow model to an order–$2D$ quadratic self-adjoint model.}

Consider an
order-$D$
complex tensor model with quartic pillow interaction
\begin{equation}\label{eq:pillowpotential}
    V_{\rm 4}(\phi{\phi^\dagger})= N\frac{\lambda}{2}\sum_{c=1}^D\sum_{\bm a,\bm b}{\phi^\dagger}_{\bm a}\phi^{{\bm a}_{\hat{c}}\bm b }{\phi^\dagger}_{\bm b}\phi^{{\bm b}_{\hat{c}}{\bm a} }\;,
\end{equation}
where $\lambda$ is a coupling  
constant. Figure~\ref{fig:pillows} illustrates the interaction vertices of this model.

\begin{figure}[h]
    \centering
    \includegraphics[width=
    0.6\linewidth]{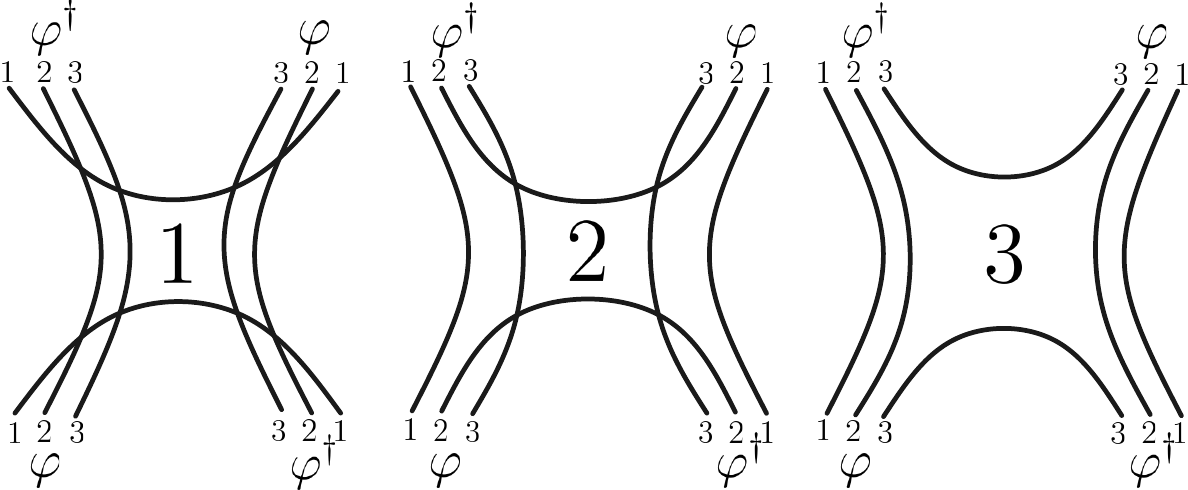}
    \caption{The three types of pillow vertices present in the $D=3$ tensor model with potential \eqref{eq:pillowpotential}.}
    \label{fig:pillows}
\end{figure}

\noindent Changing $\phi^{\bm a}\phi^{\dagger}_{\bm b}\rightarrow\Phi^{\bm a}_{\bm b}$, one finds the
corresponding potential on the self-adjoint side to be
\begin{equation}\label{eq:effpillowpotential}
     V_4(\Phi)= N\frac{\lambda}{2}\sum_{c=1}^D\sum_{\bm a,\bm b}{\Phi}_{\bm a}^{\bm a_{\hat{c}}\bm b }{\Phi}_{\bm b}^{\bm b_{\hat{c}}\bm a }\;,
\end{equation}
where we illustrate in Figure \ref{fig:pillows_sa} the interaction vertices of this model corresponding to Figure \ref{fig:pillows} for the complex tensor model potential \eqref{eq:pillowpotential}.
\begin{figure}[h]
    \centering
    \begin{subfigure}[H]{0.3\textwidth}
    \centering
    \includegraphics[width=0.8\linewidth]{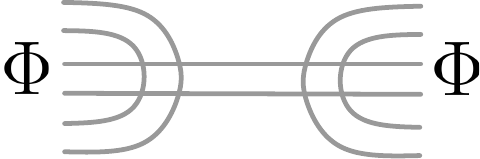}
    \caption{}
    \label{fig:spil1}
    \end{subfigure}
    \;\;
    \begin{subfigure}[H]{0.3\textwidth}
    \centering
    \includegraphics[width=0.8\linewidth]{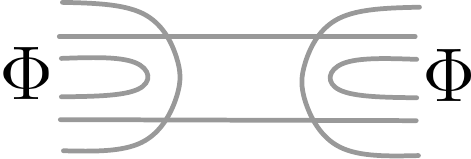}
    \caption{}
    \label{fig:spil2}
    \end{subfigure}
    \;\;
    \begin{subfigure}[H]{0.3\textwidth}
    \centering
    \includegraphics[width=0.8\linewidth]{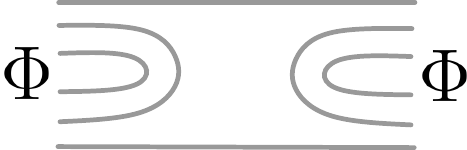}
    \caption{}
    \label{fig:spil3}
    \end{subfigure}
    \caption{The three types of quadratic vertices 
of the intermediate effective field $\Phi$
    present in the self-adjoint tensor model with potential \eqref{eq:effpillowpotential}. These are labelled consistently with the three quartic pillows in Figure \ref{fig:pillows}.
    }
    \label{fig:pillows_sa}
\end{figure}

\noindent
We consider a random complex tensor $\phi\in \mathbb{C}_N^{\otimes D}$ with
partition function
\begin{equation}
    {\cal Z}_{\rm CT}[V_{\rm 4}](R)=\iint {\cal D}{\phi^\dagger} {\cal D}\phi\,e^{-{\phi^\dagger}R^{-1}\phi
    +V_4(\phi\phi^\dagger)
    }
    \;,
\end{equation}
and a random self–adjoint tensor
$\Phi\in\Hm(\mathbb{C}_N^{\otimes (D-1)})$ distributed according to
\begin{equation}
   {\cal Z}_{\rm  HT}[ {V}_{\rm 4},
{Y}_{\ln}
   [R]]= \iint {\cal D}\Phi {\cal D}\Psi\;e^{-i\Tr \lp\Psi\Phi\rp-\Tr\ln\lp {\bb 1}^{\otimes D}-i
    R
    \Psi\rp+V_4(\Phi)}\;.
\end{equation}
Theorem~\ref{mainthm} implies that
\begin{equation}
    \frac{{\cal Z}_{\rm CT}[V_4](R)}{{\cal Z}_{\rm CT}[0](R)}=\frac{{\cal Z}_{\rm  HT}[ V_4,
{Y}_{\ln}
    [R]]}{{\cal Z}_{\rm  HT}[0,
{Y}_{\ln}
    [R]]}\;.
\end{equation}
In this example the complex quartic pillow interaction on the order–$D$ tensor is mapped to a purely quadratic interaction on the order–$2D$ self–adjoint tensor.

\paragraph{(ii) From an order–$D$ quartic pillow model to an order–$2(D-1)$ quadratic self-adjoint model.}\label{sec:pillowgaussiantensor}
Consider again an
order-$D$
complex tensor model with quartic pillow interaction as in \eqref{eq:pillowpotential}. That example does not have a unique equivalent representation. For example, by rewriting the potential \eqref{eq:pillowpotential} as 
\begin{equation}\label{eq:lowpillowpotential}
    V_{\rm 4}(\phi{\phi^\dagger})= N\frac{\lambda}{2}\sum_{\bm a,\bm b}\left({\phi^\dagger}_{\bm a}\phi^{{\bm b}_{\hat{1}}{\bm a} }{\phi^\dagger}_{\bm b}\phi^{{\bm a}_{\hat{1}}\bm b }+\sum_{c=2}^D{\phi^\dagger}_{\bm a}\phi^{{\bm a}_{\hat{c}}\bm b }{\phi^\dagger}_{\bm b}\phi^{{\bm b}_{\hat{c}}{\bm a} }\right)\;,
\end{equation}

\noindent Changing $\sum_{b_1}\phi^{\bm a_{\hat 1}\bm b}\phi^{\dagger}_{\bm b}\rightarrow\Phi^{\bm a_{\hat 1}}_{\bm b_{\hat 1}}$, one finds the
corresponding potential on the self-adjoint side to be
\begin{equation}\label{eq:efflowpillowpotential}
    \hat V_4(\Phi)=N\frac{\lambda}{2}\sum_{\bm a_{\hat 1},{\bm b}_{\hat 1}}\left({\Phi^{\bm b_{\hat 1}}_{\bm a_{\hat 1}}}{\Phi^{\bm a_{\hat 1}}_{\bm b_{\hat 1}}}+\sum_{c=2}^D\Phi^{ \bm b_{\hat 1 \hat c}\bm a_{\hat 1}}_{\bm a_{\hat 1}}\Phi^{ \bm a_{\hat 1\hat c}\bm b_{\hat 1}}_{\bm b_{\hat 1}}\right)\;.
\end{equation}
\noindent
We consider a random complex tensor $\phi\in \mathbb{C}_N^{\otimes D}$ with
partition function
\begin{equation}
\label{pfpil}
    {\cal Z}_{\rm CT}[V_{\rm 4}](R)=\iint {\cal D}{\phi^\dagger} {\cal D}\phi\,e^{-{\phi^\dagger}R^{-1}\phi
    +V_4(\phi\phi^\dagger)
    }
    \;,
\end{equation}
and a random self–adjoint tensor
$\hat\Phi\in\Hm(\mathbb{C}_N^{\otimes (D-1)})$ distributed according to
\begin{equation}
\label{pfpilef}
   {\cal Z}_{\rm \hat HT}[ \hat{V}_{\rm 4},
{Y}_{\ln}
   [R]]= \iint_{{\rm H}({\Cp_N^{\otimes (D-1)}})^2} {\cal D}\hat\Phi{\cal D}\hat\Psi\; e^{-i\Tr (\hat\Psi\hat\Phi)-\Tr\ln\lp {\bb 1}^{\otimes D}-i
    R
    ({\bb 1}\otimes\hat\Psi)\rp+\hat V_4(\Phi)}\;.
\end{equation}
Theorem~\ref{theorem:partialtensor} implies that
\begin{equation}\label{eq:V4eq}
    \frac{{\cal Z}_{\rm CT}[V_4](R)}{{\cal Z}_{\rm CT}[0](R)}=\frac{{\cal Z}_{\rm \hat HT}[\hat V_4,
{Y}_{\ln}
    [R]]}{{\cal Z}_{\rm \hat HT}[0,\
{Y}_{\ln}
    [R]]}\;.
\end{equation}
In this example the complex quartic pillow interaction on the order–$D$ tensor is mapped to a purely quadratic interaction on the order–$2(D-1)$ self–adjoint tensor.

In particular, if $R=C_2\otimes \hat P$, the self-adjoint two-tensor model 
\eqref{pfpilef}
becomes a one-tensor model with partition function as seen in \eqref{eq:c2case},
\begin{equation}\label{eq:Z_HT_v4}
    {\cal Z}_{\rm \hat HT}[\hat V_4,
\hat{Y}_{\ln}
    [C_2\otimes \hat P]]=\iint_{{\rm H}({\Cp_N^{\otimes (D-1)}})} {\cal D}\hat\Phi \;e^{-\frac{N}{2}\Tr((\hat P^{-1}\hat\Phi)^2)+\hat V_4(\hat\Phi)}\;,
\end{equation}
Additionally, in this particular case, the equality in \eqref{eq:V4eq} can be verified by explicit computation of the partition function. We elaborate on the case $D=3$ and $\hat P={\bb 1}^{\otimes 2}$ in Appendix \ref{app:pillows}. In particular, we show that
\begin{equation}
    \frac{{\cal Z}_{\rm CT}[V_4](C_2\otimes {\bb 1}^{\otimes 2})}{{\cal Z}_{\rm CT}[0](C_2\otimes {\bb 1}^{\otimes 2})}=(1-\lambda)^{-\frac{1}{2}(N^2-1)^2}
    (1-(1+N)\lambda)^{-(N^2-1)}
    (1-(1+2N)\lambda)^{-\frac{1}{2}}=\frac{{\cal Z}_{\rm \hat HT}[\hat V_4,
\hat{Y}_{\ln}
    [C_2\otimes {\bb 1}^{\otimes 2}]]}{{\cal Z}_{\rm \hat HT}[0,
\hat{Y}_{\ln}
    [C_2\otimes {\bb 1}^{\otimes 2}]]}\;.
\end{equation}

\paragraph{An order-$3$ sextic complex tensor model and its order-$4$ cubic self-adjoint counterpart.}
Consider a tensor model defined by the partition function
\begin{equation}
\label{eq:ZCTV6}
    {\cal Z}_{\rm CT}[V_6](C_2\otimes {\bb 1}^{\otimes 2})=\int d{\phi^\dagger} d\phi\,e^{-\phi^\dagger (C_2^{-1}\otimes {\bb 1}^{\otimes 2})\phi+V_6(\phi\phi^\dagger)}\;,
\end{equation}
with $\phi \in \mathbb C_N^{\otimes 3}$ and potential
\begin{equation}
\label{sextic}
\scalebox{0.8}{$\begin{aligned}
    V_{\rm 6}(\phi{\phi^\dagger})= \frac{\lambda_1}{2}&\sum_{\bm a, \bm b, \bm c}{\phi^\dagger}_{a_1a_2a_3}\phi^{b_1a_2a_3}{\phi^\dagger}_{b_1b_2b_3}\phi^{c_1b_2b_3}{\phi^\dagger}_{c_1c_2c_3}\phi^{a_1c_2c_3}
    +\frac{\lambda_2}{2}\sum_{\bm a, \bm b, \bm c}{\phi^\dagger}_{a_1a_2a_3}\phi^{a_1b_2a_3}{\phi^\dagger}_{b_1b_2b_3}\phi^{b_1c_2b_3}{\phi^\dagger}_{c_1c_2c_3}\phi^{c_1a_2c_3}\\
    +\frac{\lambda_3}{2}&\sum_{\bm a, \bm b, \bm c}{\phi^\dagger}_{a_1a_2a_3}\phi^{a_1a_2b_3}{\phi^\dagger}_{b_1b_2b_3}\phi^{b_1b_2c_3}{\phi^\dagger}_{c_1c_2c_3}\phi^{c_1c_2a_3}
    +\frac{\lambda_4}{2}\sum_{\bm a, \bm b, \bm c}{\phi^\dagger}_{a_1a_2a_3}\phi^{b_1a_2b_3}{\phi^\dagger}_{b_1b_2b_3}\phi^{a_1b_2c_3}{\phi^\dagger}_{c_1c_2c_3}\phi^{c_1c_2a_3}\\
    +\frac{\lambda_5}{2}&\sum_{\bm a, \bm b, \bm c}{\phi^\dagger}_{a_1a_2a_3}\phi^{b_1b_2a_3}{\phi^\dagger}_{b_1b_2b_3}\phi^{c_1a_2b_3}{\phi^\dagger}_{c_1c_2c_3}\phi^{a_1c_2c_3}
    +\frac{\lambda_6}{2}\sum_{\bm a, \bm b, \bm c}{\phi^\dagger}_{a_1a_2a_3}\phi^{a_1b_2b_3}{\phi^\dagger}_{b_1b_2b_3}\phi^{b_1c_2a_3}{\phi^\dagger}_{c_1c_2c_3}\phi^{c_1a_2c_3}\\
    +\frac{\lambda_7}{2}&\sum_{\bm a, \bm b, \bm c}{\phi^\dagger}_{a_1a_2a_3}\phi^{a_1b_2c_3}{\phi^\dagger}_{b_1b_2b_3}\phi^{b_1c_2a_3}{\phi^\dagger}_{c_1c_2c_3}\phi^{c_1a_2b_3}\;.
\end{aligned}$}
\end{equation}
The seven sextic-interaction terms are depicted graphically in Figure~\ref{l17}.
\begin{figure}[H]
    \begin{subfigure}{0.133\textwidth}
    \centering
    \includegraphics[width=
    \linewidth]{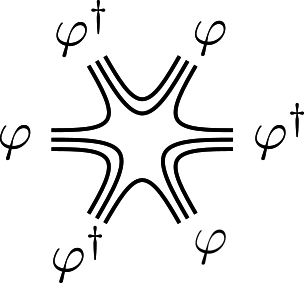}
    \caption{$\lambda_1$ vertex.}
    \label{l1}
    \end{subfigure}
    \begin{subfigure}{0.133\textwidth}
    \centering
    \includegraphics[width=
    \linewidth]{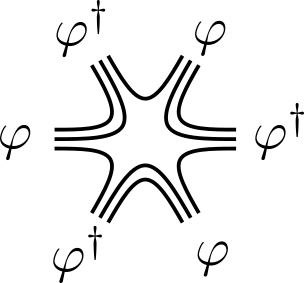}
    \caption{$\lambda_2$ vertex.}
    \label{l2}
    \end{subfigure}
    \begin{subfigure}{0.133\textwidth}
    \centering
    \includegraphics[width=
    \linewidth]{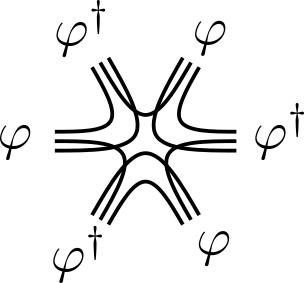}
    \caption{$\lambda_3$ vertex.}
    \label{l3}
    \end{subfigure}
    \begin{subfigure}{0.133\textwidth}
    \centering
    \includegraphics[width=
    \linewidth]{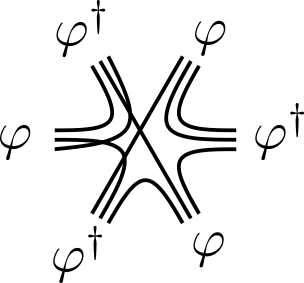}
    \caption{$\lambda_4$ vertex.}
    \label{l4}
    \end{subfigure}
    \begin{subfigure}{0.133\textwidth}
    \centering
    \includegraphics[width=
    \linewidth]{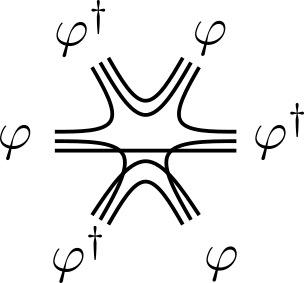}
    \caption{$\lambda_5$ vertex.}
    \label{l5}
    \end{subfigure}
    \begin{subfigure}{0.133\textwidth}
    \centering
    \includegraphics[width=
    \linewidth]{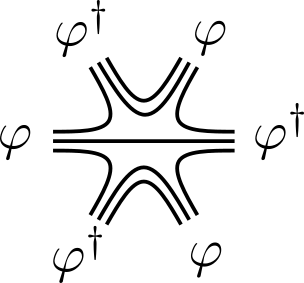}
    \caption{$\lambda_6$ vertex.}
    \label{l6}
    \end{subfigure}
    \begin{subfigure}{0.133\textwidth}
    \centering
    \includegraphics[width=
    \linewidth]{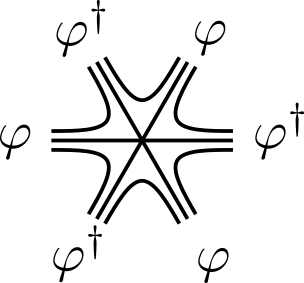}
    \caption{$\lambda_7$ vertex.}
    \label{l7}
    \end{subfigure}
    \caption{Stranded representation of the sextic vertices generated by the partition function \eqref{eq:ZCTV6}. The couplings $\lambda_1$, $\lambda_2$, and $\lambda_3$ correspond to the same bubble with different assignments of colors. The same applies to the couplings $\lambda_4$, $\lambda_5$, and $\lambda_6$. The coupling $\lambda_7$ corresponds to a vertex with the topology of a complete graph.}
    \label{l17}
\end{figure}

\noindent The equivalence stated in Theorem~\ref{theorem:partialtensor} implies that ${\cal Z}_{\rm CT}[V_6](C_2\otimes {\bb 1}^{\otimes 2})$ coincides with the partition function of the effective self-adjoint model of order-$4$ tensor
\begin{equation}
\label{eq:ZHT6}
    {\cal Z}_{\rm \hat HT}[\hat V_6,
\hat{Y}_{\ln}
    [C_2\otimes {\bb 1}^{\otimes 2}]]=\iint_{{\rm H}({\Cp_N^{\otimes 2}})} {\cal D}\hat\Phi \;e^{-\frac{N}{2}\Tr(\hat\Psi^2)+\hat V_6(\hat\Psi)}\;,
\end{equation}
where $\hat\Phi$ is a self–adjoint tensor, $\hat\Phi^{a_2a_3}_{b_2b_3}=(\hat\Phi^{b_2b_3}_{a_2a_3})^*$, and the potential is
\begin{equation}
\label{eq:hatV6}
\begin{split}
    \hat V_{\rm 6}(\hat\Phi)=\frac{\lambda_1}{2}&\sum_{\bm a_{\hat 1}, \bm b_{\hat 1}, \bm c_{\hat 1}}\hat\Phi^{c_2c_3}_{a_2a_3}\hat\Phi^{a_2a_3}_{b_2b_3}\hat\Phi^{b_2b_3}_{c_2c_3}
    +\frac{\lambda_2}{2}\sum_{\bm a_{\hat 1}, \bm b_{\hat 1}, \bm c_{\hat 1}}\hat\Phi^{b_2a_3}_{a_2a_3}\hat\Phi^{c_2b_3}_{b_2b_3}\hat\Phi^{a_2c_3}_{c_2c_3}\\
    +\frac{\lambda_3}{2}&\sum_{\bm a_{\hat 1}, \bm b_{\hat 1}, \bm c_{\hat 1}}\hat\Phi^{a_2b_3}_{a_2a_3}\hat\Phi^{b_2c_3}_{b_2b_3}\hat\Phi^{c_2a_3}_{c_2c_3}
    +\frac{\lambda_4}{2}\sum_{\bm a_{\hat 1}, \bm b_{\hat 1}, \bm c_{\hat 1}}\hat\Phi^{b_2c_3}_{a_2a_3}\hat\Phi^{a_2b_3}_{b_2b_3}\hat\Phi^{c_2a_3}_{c_2c_3}\\
    +\frac{\lambda_5}{2}&\sum_{\bm a_{\hat 1}, \bm b_{\hat 1}, \bm c_{\hat 1}}\hat\Phi^{c_2c_3}_{a_2a_3}\hat\Phi^{b_2a_3}_{b_2b_3}\hat\Phi^{a_2b_3}_{c_2c_3}
    +\frac{\lambda_6}{2}\sum_{\bm a_{\hat 1}, \bm b_{\hat 1}, \bm c_{\hat 1}}\hat\Phi^{b_2b_3}_{a_2a_3}\hat\Phi^{c_2a_3}_{b_2b_3}\hat\Phi^{a_2c_3}_{c_2c_3}\\
    +\frac{\lambda_7}{2}&\sum_{\bm a_{\hat 1}, \bm b_{\hat 1}, \bm c_{\hat 1}}\hat\Phi^{b_2c_3}_{a_2a_3}\hat\Phi^{c_2a_3}_{b_2b_3}\hat\Phi^{a_2b_3}_{c_2c_3}\;.
\end{split}
\end{equation}
In this self-adjoint model ${\cal Z}_{\rm \hat HT}[V_6](C_2\otimes {\bb 1}^{\otimes 2})$ \eqref{eq:ZHT6}, the interactions are cubic, and the corresponding vertices are shown in Figure~\ref{l17e}.
\begin{figure}[H]
    \begin{subfigure}{0.133\textwidth}
    \centering
    \includegraphics[width=
    \linewidth]{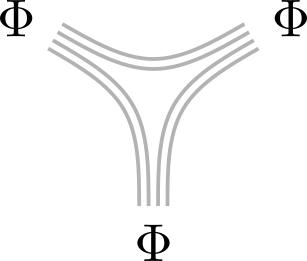}
    \caption{$\lambda_1$ vertex.}
    \label{l1e}
    \end{subfigure}
    \begin{subfigure}{0.133\textwidth}
    \centering
    \includegraphics[width=
   \linewidth]{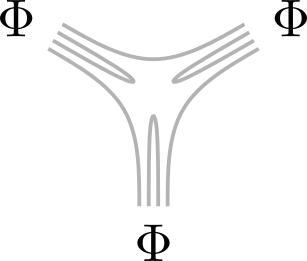}
    \caption{$\lambda_2$ vertex.}
    \label{l2e}
    \end{subfigure}
    \begin{subfigure}{0.133\textwidth}
    \centering
    \includegraphics[width=
    \linewidth]{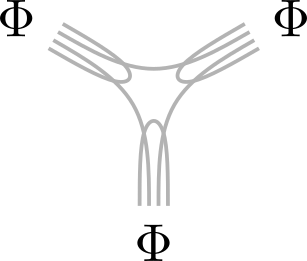}
    \caption{$\lambda_3$ vertex.}
    \label{l3e}
    \end{subfigure}
    \begin{subfigure}{0.133\textwidth}
    \centering
    \includegraphics[width=
    \linewidth]{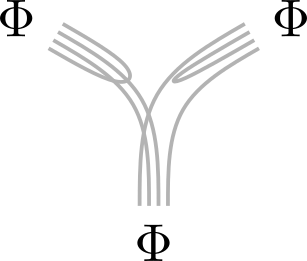}
    \caption{$\lambda_4$ vertex.}
    \label{l4e}
    \end{subfigure}
    \begin{subfigure}{0.133\textwidth}
    \centering
    \includegraphics[width=
    \linewidth]{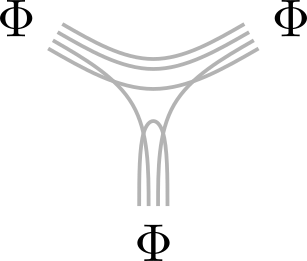}
    \caption{$\lambda_5$ vertex.}
    \label{l5e}
    \end{subfigure}
    \begin{subfigure}{0.133\textwidth}
    \centering
    \includegraphics[width=
    \linewidth]{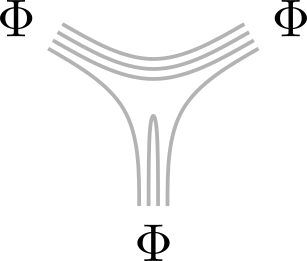}
    \caption{$\lambda_6$ vertex.}
    \label{l6e}
    \end{subfigure}
    \begin{subfigure}{0.133\textwidth}
    \centering
    \includegraphics[width=
    \linewidth]{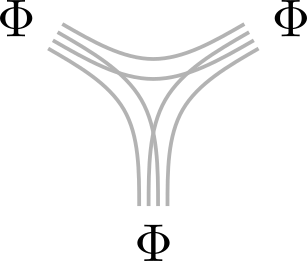}
    \caption{$\lambda_7$ vertex.}
    \label{l7e}
    \end{subfigure}
    \caption{Stranded representation of the cubic vertices 
of the intermediate effective field $\Phi$
    generated by the partition function \eqref{eq:ZHT6}.
    }
    \label{l17e}
\end{figure}

It is worth noticing that there is an intermediate action where all fields are present. In this sextic model, it is equal to
\begin{equation}\label{eq:intpf}
    {\cal Z}_{\rm INT}[\hat V_6,
\hat{Y}_{\ln}
    [C_2\otimes {\bb 1}^{\otimes 2}]]
    =\iint {\cal D}\phi^\dagger {\cal D}\phi\iint {\cal D}\hat\Phi {\cal D}\hat\Psi\;e^{-\phi^\dagger \lp 
    R^{-1}
    -i\,{\bb 1}\otimes\hat\Psi\rp\phi-i\Tr \lp\hat\Psi\hat\Phi\rp+\hat V_6(\hat\Phi)}\;,
\end{equation}
and incorporates the transformations shown in Fig.~\ref{fig:vertexint}.
\begin{figure}[H]
    \centering
    \includegraphics[width=0.15
    \linewidth]{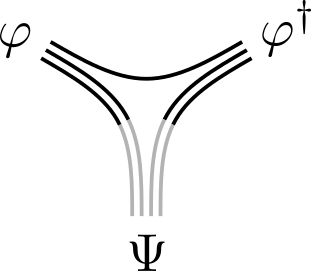}
    \caption{The interaction vertex $\phi^\dagger\Psi\phi$, which appears in \eqref{eq:intpf}.
    }
    \label{fig:Phiphi}
\end{figure}
\begin{figure}[H]
    \centering
    \setcounter{subfigure}{0}
    \begin{subfigure}[H]{0.3\textwidth}
    \centering
    \includegraphics[width=\linewidth]{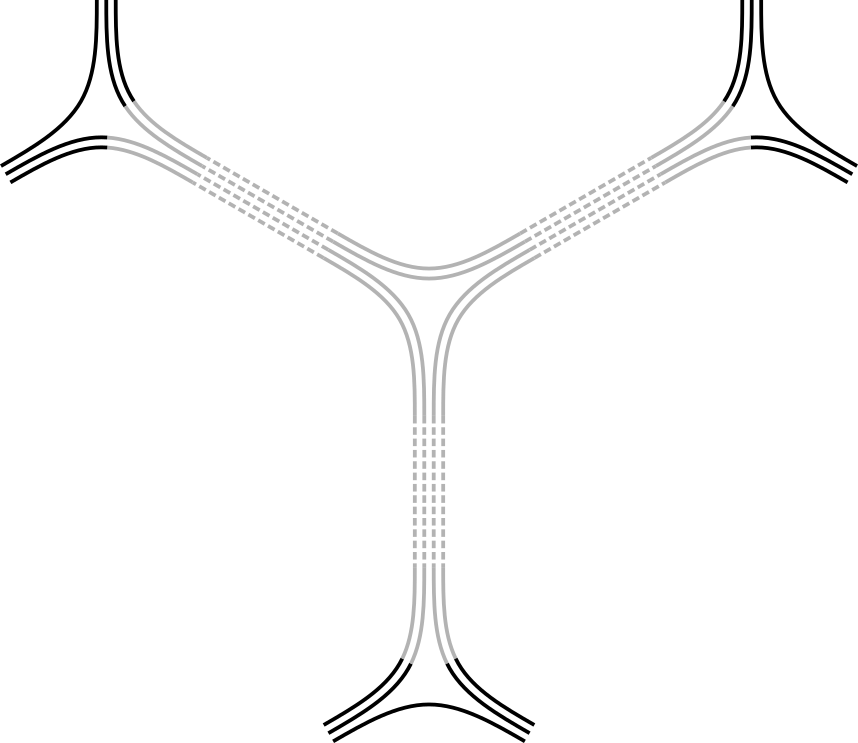}
    \caption{}
    \label{fig:vertexint1}
    \end{subfigure}
    \;\;
    \setcounter{subfigure}{5}
    \begin{subfigure}[H]{0.3\textwidth}
    \centering
    \includegraphics[width=\linewidth]{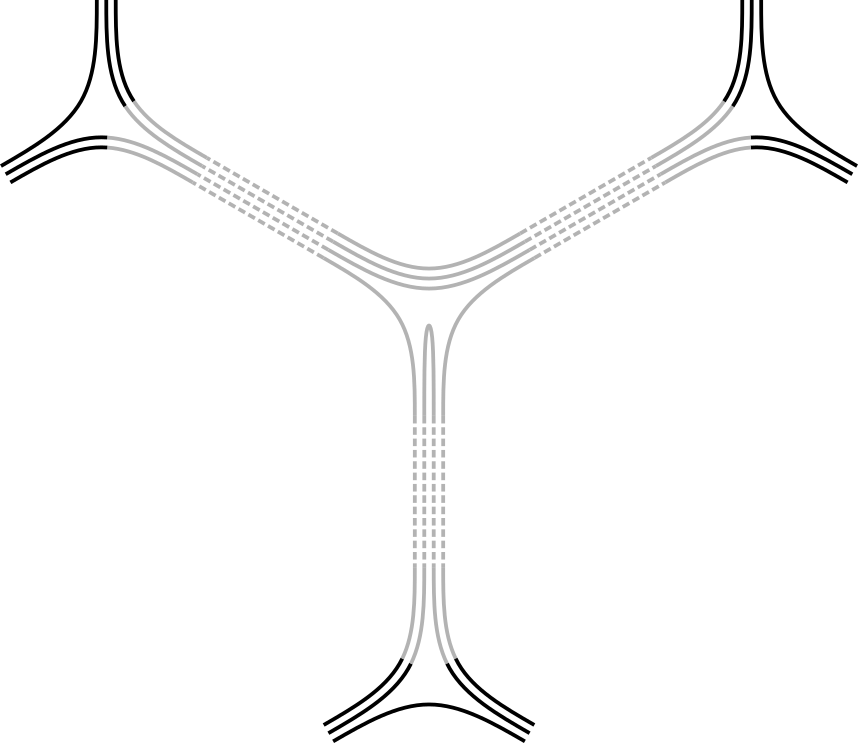}
    \caption{}
    \label{fig:vertexint2}
    \end{subfigure}
    \;\;
    \setcounter{subfigure}{6}
    \begin{subfigure}[H]{0.3\textwidth}
    \centering
    \includegraphics[width=\linewidth]{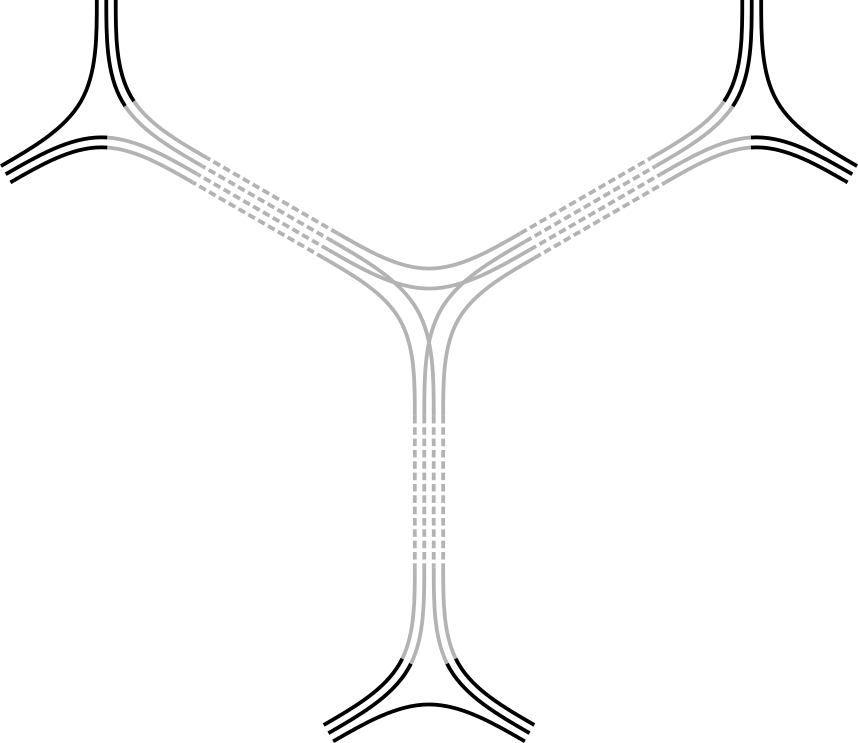}
    \caption{}
    \label{fig:vertexint3}
    \end{subfigure}
    \caption{Examples of the intermediate steps (according to \eqref{eq:intpf}) in the equivalence between the sextic vertex represented in Fig. \ref{l1}, Fig. \ref{l6}, Fig. \ref{l7}, and the cubic vertex represented in Fig. \ref{l1e}, Fig. \ref{l6e}, Fig. \ref{l7e} ((a), (f) and (g), respectively).
Intermediate fields $\Psi$ and $\Phi$ are presented in gray lines as in Figure~\ref{fig:PhiPsiprop}.
}
    \label{fig:vertexint}
\end{figure}

\paragraph{From an order–$D$ quartic enhanced pillow model to an order–$2D$ quadratic self-adjoint model.}

Consider an order-$D$ complex tensor model with enhanced quartic pillow interaction
\begin{equation}\label{eq:enhancedpillowpotential}
    V_{\rm 4}(\phi{\phi^\dagger})= N\frac{\lambda}{2}\sum_{\kappa=1}^D\sum_{\bm a,\bm b, \bm c}{\psi^\dagger}_{\bm c}\psi^{{\bm c}_{\hat{\kappa}}{\bm a} }{\phi^\dagger}_{\bm a}\phi^{{\bm a}_{\hat{\kappa}}\bm b }{\phi^\dagger}_{\bm b}\phi^{{\bm b}_{\hat{\kappa}}{\bm c} }\;,
\end{equation}
where $\lambda$ is a coupling  
constant, and $\psi$ is a constant tensor 
which effectively adds weights to certain effective strands (see Figure \ref{fig:pillows2}): the constant tensor $\psi$, then can be viewed to modify the usual pillow model, which is explored in \eqref{eq:pillowpotential}. 
\begin{figure}[H]
    \centering
    \includegraphics[width=0.7
    \linewidth]{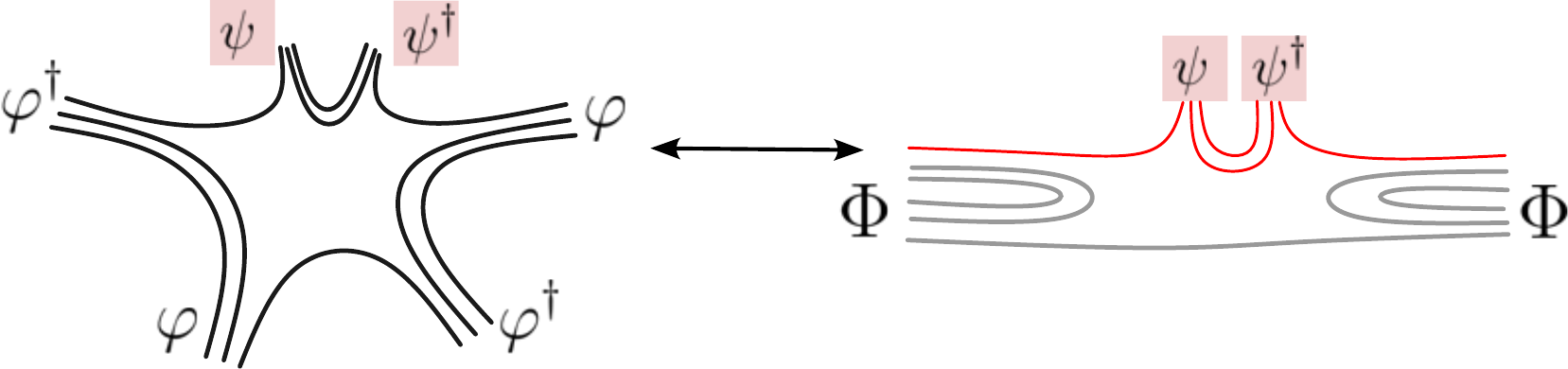}
    \caption{Examples of the interaction bubbles present in \textbf{left}: the  complex tensor model with quartic potential \eqref{eq:enhancedpillowpotential} and \textbf{right}: the self-transpose tensor model with quadratic potential  \eqref{eq:effenhancedpillowpotential} for the case $D=3$. In this case, the thinner red lines correspond to the indices contracted with $\psi$ and $\psi^\dagger$.
    }
    \label{fig:pillows2}
\end{figure}

\noindent Changing $\phi^{\bm a}\phi^{\dagger}_{\bm b}\rightarrow\Phi^{\bm a}_{\bm b}$, one finds the
corresponding potential on the self-adjoint side to be
\begin{equation}\label{eq:effenhancedpillowpotential}
     V_4(\Phi)= N\frac{\lambda}{2}\sum_{\kappa=1}^D\sum_{\bm a,\bm b, \bm c}({\psi\psi^\dagger})_{\bm c}^{{\bm c}_{\hat{\kappa}}{\bm a} }{\Phi}_{\bm a}^{{\bm a}_{\hat{\kappa}}\bm b }{\Phi}_{\bm b}^{{\bm b}_{\hat{\kappa}}{\bm c} }\;.
\end{equation}
An example of the type of vertices in this model is shown in Figure \ref{fig:pillows2}.
We consider a random complex tensor $\phi\in \mathbb{C}_N^{\otimes D}$ with
partition function with potential \eqref{eq:enhancedpillowpotential}
\begin{equation}
    {\cal Z}_{\rm CT}[V_{\rm 4}](R)=\iint {\cal D}{\phi^\dagger} {\cal D}\phi\,e^{-{\phi^\dagger}R^{-1}\phi
    +V_4(\phi\phi^\dagger)
    }
    \;,
\end{equation}
and a random self–adjoint tensor
$\Phi\in\Hm(\mathbb{C}_N^{\otimes D})$ distributed according to
\begin{equation}
   {\cal Z}_{\rm  HT}[ {V}_{\rm 4},
{Y}_{\ln}
   [R]]= \iint {\cal D}\Phi {\cal D}\Psi\;e^{-i\Tr \lp\Psi\Phi\rp-\Tr\ln\lp {\bb 1}^{\otimes D}-i
    R
    \Psi\rp+V_4(\Phi)}\;,
\end{equation}
with potential \eqref{eq:effenhancedpillowpotential}.
Theorem~\ref{mainthm} implies that
\begin{equation}
    \frac{{\cal Z}_{\rm CT}[V_4](R)}{{\cal Z}_{\rm CT}[0](R)}=\frac{{\cal Z}_{\rm HT}[ V_4,
{Y}_{\ln}
    [R]]}{{\cal Z}_{\rm  HT}[0,
{Y}_{\ln}
    [R]]}\;.
\end{equation}
In this example the enhanced quartic pillow interaction on the order–$D$ complex tensor is mapped to a quadratic interaction on the order–$2D$ self–adjoint tensor.

\subsection{A prospect on real tensor models}\label{sec:tetragaussian}
\noindent Up to this point, we have shown equivalences of complex tensor models only. However, real tensor models are of considerable interest for quantum gravity, and several families have been extensively studied in the physics literature~\cite{Eichhorn:2018phj,Eichhorn:2018ylk,Eichhorn:2019hsa}. One motivation is that the numerical values of the critical exponents characterizing the continuum limit appear to be insensitive to whether the tensor ensemble is complex or real. Since the real case is technically simpler, such models are often preferred in FRG.
Nevertheless, although some results have suggested the existence of a continuum limit compatible with four-dimensional quantum gravity~\cite{Eichhorn:2018phj,Eichhorn:2018ylk,Eichhorn:2019hsa},
no study has yet demonstrated the emergence of geometries beyond the branched–polymer phase. This limitation has renewed interest in introducing a causal structure, in close analogy with what has been achieved in the causal matrix model case~\cite{Benedetti:2008hc,Castro:2020dzt}. It is precisely in this context that the equivalences established in this work may provide a useful framework for implementing such constraints.

In this subsection, we establish an equivalence between an order-$3$ real tensor model with the matrix $C_2$ \eqref{cnprop} inserted in the first strand of the propagator, and an order-$4$ real self-transpose tensor model without $C_2$. We define the self-transpose tensor model by the following condition: for $\Phi \in {\bb R}_N^{\otimes(D-1)}\otimes{\bb R}_N^{\otimes(D-1)}$, their tensor elements $\Phi^{\bm{a}}_{\bm{b}}$, with $\bm{a},\bm{b}\in \{1,...,N\}^{D-1}$, satisfy $\Phi^{\bm{a}}_{\bm{b}}=\Phi^{\bm{b}}_{\bm{a}}$.

This result shows that a tensor model implementing a causal-like\footnote{
However, this model generates graphs that are strictly a subset of CDT as their dual triangulation just like the model presented in Subection~\ref{subsec:GLmodel}.}
constraint can be reformulated as a model in which this constraint is traded for a symmetry of the tensor indices (namely, symmetry under transposition). As a consequence, this equivalence offers a setting in which the continuum limit may be investigated using FRG techniques without having to deal explicitly with the presence of the constant matrix $C_2$. Whether this reformulation is technically simpler in practice remains an open question and is left for future investigation. Moreover, in the tensor models considered here, the interaction are of tetrahedral type. This structure is more closely aligned with the geometric picture of gluing $D$-simplices to construct a discrete spacetime, thereby their connection with simplicial approaches to quantum gravity is more straightforward.

Given this motivation, we move on to probing the equivance by computing the partition functions of both models. We do this for the case $D=3$, however, we expect that this equivalence should be generalized easily to any $D$.

First, for the order-$3$ real tensor model we obtain the following result.
\begin{proposition} \label{prop:realtensor}
Let $\phi \in\mathbb R_N^{\otimes 3}$ be a random tensor distributed according to the partition function
\begin{equation}
\label{eq:ZRT4}
{\cal Z}_{\rm RT}(\lambda)=\int d\phi\,e^{-S_{\rm RT}[\phi](\lambda)}\;
\end{equation}
where,
\begin{equation}
    S_{\rm RT}[\phi](\lambda)= N\Big(\frac{1}{2}\sum_{ijk}\phi^{ijk}[C^{-1}_2]_{ii'}\phi^{i'jk}+\frac{\lambda}{4}\sum_{ijk,i'j'k'}\phi^{ijk}\phi^{ij'k'}\phi^{i'jk'}\phi^{i'j'k}\Big)\;
\end{equation}
and $C_2$ is defined as in \eqref{cnprop}.
Then, \eqref{eq:ZRT4} can be written as follows 
\begin{equation}
    \frac{{\cal Z}_{\rm RT}(\lambda)}{{\cal Z}_{\rm RT}(0)}=(1+\lambda)^{-\frac{1}{2}\left(\frac{N(N+1)}{2}\right)^2}(1-\lambda)^{-\frac{1}{2}\left(\frac{N(N-1)}{2}\right)^2}\;.
\end{equation}
\end{proposition}

\begin{proof}
We start with the free energy ${\cal F}_{\rm RT}(\lambda)=\ln {\cal Z}_{\rm RT}(\lambda)$ and write
\begin{equation}
    {\cal F}_{\rm RT}(\lambda)-{\cal F}_{\rm RT}(0)=\sum_{p=1}^\infty \frac{1}{p!}\left(-\frac{\lambda N}{4}\right)^p\langle (\phi^{ijk}\phi^{ij'k'}\phi^{i'jk'}\phi^{i'j'k})^p\rangle_c(0)\;,
\end{equation}
where $\langle\cdot\rangle_c(0)$ is the connected part of the 
\begin{equation}
    \langle (\phi^{ijk}\phi^{ij'k'}\phi^{i'jk'}\phi^{i'j'k})^p\rangle(0)=\frac{\int d\phi\,e^{-S_{\rm RT}[\phi](\lambda=0)}\; (\phi^{ijk}\phi^{ij'k'}\phi^{i'jk'}\phi^{i'j'k})^p}{\int d\phi\,e^{-S_{\rm RT}[\phi](\lambda=0)}},
\end{equation}
at vanishing coupling $\lambda=0$.
If $p$ is even, the graph resummation gives
\begin{equation}
\langle (\phi^{ijk}\phi^{ij'k'}\phi^{i'jk'}\phi^{i'j'k})^p\rangle_c(0)
    =
\begin{cases}
    (p-1)!4^{p-1}N^{-2p}N^p(N^4+N^2)\;, 
    \quad {\text{for $p$ even}}
    \,,
    \\
    (p-1)!4^{p-1}N^{-2p}N^p(N^3+N^3)
    \quad {\text{for $p$ odd}}
    \,.
\end{cases}
\end{equation}
Together, these imply that
\begin{equation}
\begin{split}
        {\cal F}_{\rm RT}(\lambda)-&{\cal F}_{\rm RT}(0)=\\
        &2N^3\sum_{p=1,3,...}^\infty \frac{1}{p!}\left(-\frac{\lambda N}{4}\right)^p (p-1)!4^{p-1}N^{-p}+(N^4+N^2)\sum_{p=2,4,...}^\infty \frac{1}{p!}\left(-\frac{\lambda N}{4}\right)^p (p-1)!4^{p-1}N^{-p}.
\end{split}
\end{equation}
Simplifying the factors of $N$ and $p$, we obtain
\begin{equation}
    {\cal F}_{\rm RT}(\lambda)-{\cal F}_{\rm RT}(0)=\frac{N^3}{2}
    \sum_{p \;{\text{odd}}}^\infty
    \frac{(-\lambda)^p}{p}+\frac{N^4+N^2}{4}
    \sum_{p \; {\text{even}}}^\infty 
    \frac{(-\lambda)^p}{p}\;.
\end{equation}
By splitting the sum into even and odd indices, this summation can be rearranged to obtain
\begin{equation}
    {\cal F}_{\rm RT}(\lambda)-{\cal F}_{\rm RT}(0)=\frac{2N^3+N^4+N^2}{8}\sum_{p=1}^\infty \frac{(-\lambda)^p}{p}+\frac{-2N^3+N^4+N^2}{8}\sum_{p=1}^\infty \frac{\lambda^p}{p}\;.
\end{equation}
Identifying the expansion of the logarithm in the previous expression, we write
\begin{equation}
    {\cal F}_{\rm RT}(\lambda)-{\cal F}_{\rm RT}(0)=\frac{N^2(N+1)^2}{8}(-\ln(1+\lambda))+\frac{N^2(N-1)^2}{8}(-\ln (1-\lambda))\;.
\end{equation}
Finally, by using ${\cal F}_{\rm RT}=\ln {\cal Z}_{\rm RT}$, we get
\begin{equation}
    {\cal F}_{\rm RT}(\lambda)-{\cal F}_{\rm RT}(0)=(1+\lambda)^{-\frac{N^2(N+1)^2}{8}}(1-\lambda)^{-\frac{N^2(N-1)^2}{8}}\;.
\end{equation}
\end{proof}
\noindent For the real self-transpose tensor model, we obtain the following result.
\begin{proposition}
\label{prop:partialsymtensor}
    Consider the random tensor $\Phi \in {\bb R}_N^{\otimes 4}$ distributed according to the partition function
\begin{equation}\label{eq:thm_part_sym_tensor}
     {\cal Z}_{\rm ST}(\lambda)= \iint {\cal D}\Phi e^{-S_{\rm ST}[\Phi](\lambda)} \;,\quad  {\rm where} \quad S_{\rm ST}[\Phi](\lambda)=N\Big(\frac{1}{2}\sum_{ijkl}\Phi_{ijkl}\Phi_{ijkl}+\frac{\lambda}{2}\sum_{ijkl}\Phi_{ijkl}\Phi_{ikjl} \Big),
\end{equation}
and $\Phi$ is a  self-transpose real tensor ($\Phi_{ijkl}=\Phi_{lkji}$). Then, this partition function can be written explicitly in terms of $\lambda$ as follows
\begin{equation}
    \frac{{\cal Z}_{\rm ST}(\lambda)}{{\cal Z}_{\rm ST}(0)}=(1+\lambda)^{-\frac{1}{2}\left(\frac{N(N+1)}{2}\right)^2}(1-\lambda)^{-\frac{1}{2}\left(\frac{N(N-1)}{2}\right)^2}\;.
\end{equation}
\end{proposition}

\begin{proof}
We start by writing the action \eqref{eq:thm_part_sym_tensor} as
\begin{equation}
     {S}_{\rm ST}[\Phi](\lambda)
    =\frac{N}{2}\sum_{ijkl}
    \Phi_{ijkl}
    C^{ijkl;i'j'k'l'}
    \Phi_{i'j'k'l'}\;,
\end{equation}
where
\begin{equation}
    C^{ijkl;i'j'k'l'}=\delta_{ii'}\delta_{jj'}\delta_{kk'}\delta_{ll'}+\lambda \delta_{ii'}\delta_{jk'}\delta_{kj'}\delta_{ll'}.
\end{equation}
Then, the computation of the partition function \eqref{eq:thm_part_sym_tensor} reduces to computing the determinant of the operator $C$. By noticing that
\begin{equation}
    \begin{split}
        S_{\rm ST}[\Phi](\lambda)&=2\frac{N}{2}\sum_{ijkl;i>l}\sum_{i'j'k'l'}\Phi_{ijkl}\left(\delta_{ii'}\delta_{jj'}\delta_{kk'}\delta_{ll'}+\lambda \delta_{ii'}\delta_{jk'}\delta_{kj'}\delta_{ll'} \right)\Phi_{i'j'k'l'}\\
        &+2\frac{N}{2}\sum_{ijk;j>k}\left(1+\lambda\right)\Phi_{ijki}^2+\frac{N}{2}\sum_{ij}\left(1+\lambda\right)\Phi_{ijji}^2
        \,,
    \end{split}
\end{equation}
we proceed by decomposing $C$ into three blocks
\begin{equation}\label{eq:Pdecomp}
    C=C_{(>)}\oplus C_{(=,>)}\oplus C_{(=,=)}\;,
\end{equation}
with (See Section~\ref{sec:notations} for the definition of $\Sigma$)
\begin{equation}
    C_{(>)}:=I_{\frac{N(N-1)}{2}}\otimes (I_N\otimes I_N+\lambda \Sigma)\,,
\quad
    C_{(=,>)}:=(1+\lambda)I_N\otimes I_{\frac{N(N-1)}{2}}\,,
\quad
    C_{(=,=)}:=(1+\lambda )I_N\otimes I_{N}\,.
\end{equation}
The determinant of $I_N\otimes I_N+\lambda \Sigma$ can be computed from its eigenvalues. Considering the canonical basis $e_i$ with coefficients $(e_i)_j=\delta_{ij}$ with $i\leq j$, the vectors 
\begin{equation}
    u_{ij}= e_i\otimes e_j+e_j\otimes e_i
\end{equation}
are eigenvectors of $\Sigma$ with eigenvalues $\lambda$. Additionally, for $i<j$,
\begin{equation}
    v_{ij}= e_i\otimes e_j-e_j\otimes e_i
\end{equation}
are the remaining eigenvectors of $\Sigma$ with eigenvalues $-\lambda$.  These are also eigenvectors of $I_N\otimes I_N+\lambda \Sigma$, $u_{ij}$, with eigenvalues $1+\lambda$, and $v_{ij}$ with eigenvalues $1-\lambda$. Thefore, we have that
\begin{equation}
    \det(I_N\otimes I_N+\lambda \Sigma)=(1+\lambda)^{\frac{N(N+1)}{2}}(1+\lambda)^{\frac{N(N+1)}{2}}\;,
\end{equation}
which implies that
\begin{equation}
    \det(C_{(>)})=\lb(1+\lambda)^{\frac{N(N+1)}{2}}(1-\lambda)^{\frac{N(N-1)}{2}}\rb^{\frac{N(N-1)}{2}}\;.
\end{equation}
The computation of the remaining parts of \eqref{eq:Pdecomp} is trivial and the result is
\begin{equation}
    \det(C_{(=,>)})=(1+\lambda)^{N\frac{N(N-1)}{2}}\,,\qquad
    \det(C_{(=,=)})=(1+\lambda)^{N^2}\,.
\end{equation}
Therefore, given the decomposition \eqref{eq:Pdecomp}, we have that
\begin{equation}
\det(C)=\det(C|_{\lambda=0})(1+\lambda)^{\lp \frac{N(N+1)}{2}\rp^2}(1-\lambda)^{\lp \frac{N(N-1)}{2}\rp^2}\;.
\end{equation}
Finally, using this expression in 
\begin{equation}
    \frac{{\cal Z}_{\rm ST}(\lambda)}{{\cal Z}_{\rm ST}(0)}=\frac{\det(C_\lambda)^{-\frac{1}{2}}}{\det(C|_{\lambda=0})^{-\frac{1}{2}}}\;,
\end{equation}
we conclude that
\begin{equation}
    \frac{{\cal Z}_{\rm ST}(\lambda)}{{\cal Z}_{\rm ST}(0)}=(1+\lambda)^{-\frac{1}{2}\left(\frac{N(N+1)}{2}\right)^2}(1-\lambda)^{-\frac{1}{2}\left(\frac{N(N-1)}{2}\right)^2}\;.
\end{equation}
\end{proof}
\noindent Due to Proposition~\ref{prop:realtensor} and Proposition~\ref{prop:partialsymtensor}, we have that the two models are equivalent.
\begin{figure}[h]
    \centering
    \includegraphics[width=0.8
    \linewidth]{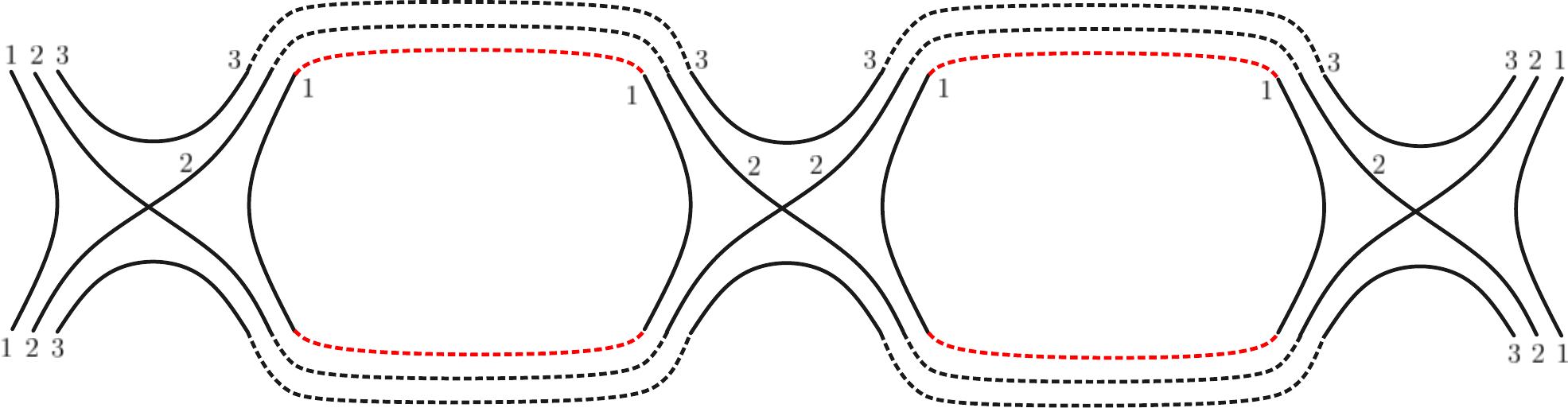}
    \caption{An example of a segment of a graph in the order-$3$ real tensor model. The red lines indicate the presence of the matrix $C_2$ in the propagator. In this case, the interactions are tetrahedral. In general, any graph generated by the partition function \eqref{eq:ZRT4} will be a closure of a sequence of vertices as presented here.}
    \label{tetragraph}
\end{figure}
\begin{figure}[h]
    \centering
    \includegraphics[width=0.7
    \linewidth]{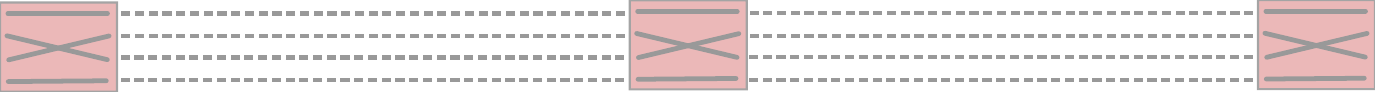}
    \caption{An example of a segment of a graph in the effective model: a real self-transpose tensor model without $C_2$. The red boxes correpond to the quadratic interactions in \eqref{eq:thm_part_sym_tensor}. The parallel dashed gray lines represent the propagators. Any graph in this model will be a closure of a sequence of vertices as presented here.}
    \label{tetragrapheff}
\end{figure}

In summary, the proved equivalence shows that a real tensor model of order 3 with a rigidity constraint is equivalent to a real self-transpose tensor model of order 4 without rigidity constraints.  
In terms of the Feynman graph expansion, the rigidity constraint in the order-$3$ model leads to necklace/chain-type graphs as illustated in  Figure~\ref{tetragraph}. This is remiscent of the complex matrix model \eqref{bc}.
In other words, the tetrahedral interactions deform into necklace/chain-type graphs.

This result may appear unappealing in the context of QG since the tetrahedral interactions do not produce extended geometries in all directions, they produce necklace/chain-type geometries which extend in one direction only.  However, a deeper study of the continuum limit of this model, particularly the values of its critical exponents, is necessary to corroborate if this is in the branched polymer universality class. We elaborate on this in the future direction section.

We also want to highlight that since $C_2$ was suitable to generate CDT 
using \emph{Hermitian matrix models}, an alternative direction is to find an
equivalent constant matrix/tensor 
to $C_2$ which is suitable for imposing a causal structure in \emph{complex matrix and tensor models}.

\section{Conclusions and outlook}\label{sec:concl}
\noindent In this paper, we have explored 
and presented newly found
equivalences between complex and self-adjoint random matrix and tensor models. More specifically, we established exact equivalences via intermediate–field representations between
\begin{itemize}
    \item complex matrix models with 
    a general covariance $(P,Q)$
    and potentials depending on $M^\dagger M$, and self–adjoint two–matrix models whose intermediate field carries a logarithmic dually-weighted potential 
(Theorem \ref{matrixequiv_thm});
    \item complex tensor models with covariance $R$ and potentials depending on $\phi\otimes\phi^\dagger$, and self–adjoint tensor models
    with dually-weighted logarithmic potentials
(Theorems
     \ref{mainthm}, and \ref{theorem:partialtensor}).
\end{itemize}

In all cases, the equivalence holds at the level of partition functions and all trace invariants built from the self–adjoint combinations: $M^\dagger M$ for matrices, and $\phi\otimes\phi^\dagger$ for tensors. In general, intermediate field representation reduces the order of the polynomial interaction potential and can provide a simpler analysis. As an application of the matrix equivalence, we showed in
Section~\ref{examp} that a class of non-Gaussian matrix models can be
reformulated as Gaussian ones. By focusing on the equivalence of observables, including disc functions, we obtained a new representation of the Catalan numbers.

In the random matrix model literature, a widely used technique is the
character expansion, which expresses partition functions and observables in
terms of irreducible characters of the symmetry group, often reducing the
effective degrees of freedom and simplifying the analysis. In the present
setting, Section~\ref{relres} shows that the expectation values obtained in
this work provide the necessary input to evaluate
$\langle \chi_\lambda(A)\rangle$ for arbitrary partitions $\lambda$,
thereby connecting our results to standard character-based methods.

\subsection*{Future directions}
\noindent From a mathematical point of view, it would be interesting to understand whether the rigid
block structures arising in the $Q=C_k$ models admit a more intrinsic combinatorial
characterization, and whether the associated random matrix and tensor ensembles can be
systematically related to well–defined classes of matroidal or hypergraph structures \cite{oxley2006matroid}.

A point that remains to be explored concerns the convergence conditions for the tensor $R$, which plays the role of the covariance in the tensor model case. This analysis should be analogous to the one presented in Appendix~\ref{change} for the covariance $(P,Q)$ in the matrix model case. We plan to investigate this question in future work.
Furthermore, intermediate field representation played a crucial role \cite{Lionni:2016ush} to prove Borel summability (which then uniquely defines the theory non-perturbatively)  of the vector, matrix, and tensor models of quartic interaction. 
Therefore, our novel intermediate field representation techniques which now cover more general (i.e., dually-weighted) models may allow to proe Borel summability of a larger class of models.

 On the physics side, there remains to be investigated if the matrix and tensor models including the matrix $C_2$ we introduced, belong to universality classes interesting for quantum gravity. As we discussed, the fact that the graph struture of the models we introduced reduces to chain-like graphs does not discard them. As it has been shown~\cite{Eichhorn:2018phj,Eichhorn:2018ylk,Eichhorn:2019hsa}, even when the continuum limit of some tensor models is dominated by melonic graphs, which in principle do not resemble extended geometries, their critical exponents are compatible with those of some QG models (e.g. unimodular gravity). This may be analogous to what happens in two-dimensional CDT~\cite{Benedetti:2008hc,Castro:2020dzt} where one of the critical exponents of the matrix model (namely, the string susceptibility) corresponds to that of the tree universality class. However, this tree structure does not represent the full two-dimensional geometry per se; rather, it encodes the global foliation structure. We plan to explore this research direction by carrying out an FRG analysis of the real tensor model \eqref{eq:ZRT4}.

In this work, we decided to focus on the case $Q=C_k$ and, in particular, $Q=C_2$ as defined in \eqref{cnprop}. However, in \cite{Benedetti:2008hc} $C_2$ was found to be suitable to impose causality in the graph expansion of Hermitian matrix models in particular. Therefore, it remains to be investigated if suitable tensor structures analogous to $C_2$ that serve to impose a causal structure in complex matrix and tensor models exist. If they do, it will be of interest to explore whether the equivalences proved in this work can be applied to them. Indeed, the equivalences we construct are sufficiently general to accommodate a broad class of covariance structures $(P,Q)$ in the matrix case and $R$ in the tensor model case. This is an avenue we hope to explore in future work.

On the more speculative side, we noted that in the matrix model case, the dual graphs Figure~\ref{fig:korder} have a strong resemblence with layered posets (two-orders) appearing in the Causal Set configuration space \cite{Carlip:2022nsv}. These correspond to spacetimes with only three time steps, which dominate the Causal Set path integral. They are therefore regarded as problematic, since such spacetimes do not resemble our extended universe. It remains an open question whether the continuum limit of the random matrix model \eqref{model1} presented here retains these pathological configurations. We expect this analysis can be done using FRG techniques similarly to \cite{Eichhorn:2013isa,Castro:2020dzt}.

\section*{Acknowledgements}
\noindent We thank Dario Benedetti and Luca Lionni for discussions, and Beno\^it Collins, Colin McSwiggen, and  Jean-Bernard Zuber for correspondence. A.C. work was partially funded by the ANR-20-CE48-0018 “3DMaps” grant and  the Deutsche Forschungsgemeinschaft (DFG, German Research Foundation) under Germany's Excellence Strategy EXC 2181/1 - 390900948 (the Heidelberg STRUCTURES Excellence Cluster).

\paragraph{Conflicts of Interest} The authors have no conflict of interest to declare.

\appendix
\section{The matrix \texorpdfstring{$C_2$}{C2}}\label{c2app}
\noindent As seen in \eqref{cnprop}, the $N \times N$ matrix $C_m$, with $m=1,...,N$, is a matrix that satisfies the condition
\begin{equation}
    \mathrm{Tr}(C_m^{p})=N\delta_{p,m}\;,
\end{equation}
for $p=1,...,N$. An explicit expression for $C_m$ in the large N limit was obtained in \cite{abranches2025}. The matrix $C_m$, and particularly $C_2$, appears in different and seemingly unrelated contexts in matrix models and representation theory. Some examples are:
\begin{itemize}
    \item Character expansion.  
It is known \cite{fulton1991representation} that for two given matrices $A$ and $B$ in $\GL(N)$, their traces and characters satisfy the following property:
\begin{equation}
    e^{\sum_{k=1}^\infty\frac{1}{p}\mathrm{Tr}(A^p)\mathrm{Tr}(B^p)}=\prod_{i,j}\frac{1}{1-a_ib_j}=\sum_{r}\chi_r(A)\chi_r(B)\;,
\end{equation}
where the first equality is obtained by expanding the l.h.s. in terms of the eigenvalues $a_i$ and $b_i$ of $A$ and $B$, respectively. The second equality comes from Cauchy's identity (see equation A.13 in \cite{fulton1991representation}). By choosing $B=C_m$, we get the following character expansion formula in the large-$N$ limit:
\begin{equation}
    e^{\frac{N}{m}\mathrm{Tr}(A^m)}=\sum_{r}\chi_r(C_m)\chi_r(A)\;,
\end{equation}
In particular, for $m=2$, we get
\begin{equation}
    e^{\frac{N}{2}\mathrm{Tr}(A^2)}=\sum_{r}c_r\,\chi_r(A)\;,
\end{equation}
with coefficients $c_r=\chi_r(C_2)$.

\item Di Francesco-Itzykson integral. The Di Francesco-Itzykson integral appears when evaluating the Gaussian average of the character function. That is,
\begin{equation}
    \left \langle\chi_r(A)\right\rangle=\frac{\int dA\,\chi_r(A)\,e^{-\frac{N}{2}\mathrm{Tr}A^2}}{\int dA\,e^{-\frac{N}{2}\mathrm{Tr}A^2}}\;.
\end{equation}

As shown in \cite{difrancesco1992generatingfunctionfatgraphs}, as well as an alternative proof in \cite{abranches2025}, this integral satisfies
\begin{equation}
    \left \langle\chi_r(A)\right\rangle=\frac{\chi_r(\mathbb 1)\chi_r(C_2)}{\chi_r(C_1)},
\end{equation}
In Section~\ref{relres}, we show a generalization of this. Namely, in Corolaries \ref{Cor:character_general_A} and \ref{Cor:character_more_general_A}, we show that $\left \langle\chi_r(A)\right\rangle$ can be expressed in terms of $C_1$ in some non-Gaussian ensembles for any representation $r$ of $\GL(N)$.

\item Causal matrix model. The causal matrix model defined by the partition function \eqref{cdtmm} (Section \ref{subsec:CDT}) and its ribbon graphs have three basic properties: the propagators of A and B correspond respectively to spacelike ribbon-graph edges and timelike ribbon-graph edges; all vertices have degree three and have two ribbon-graph spacelike edges and one ribbon-graph timelike edge; all faces have either two or zero timelike ribbon-graph edges. The matrix $C_2$ is responsible for the last condition. The property $\mathrm{Tr}(C_2^p)=N\delta_{p,2}$ imposes that all the faces with timelike ribbon-graph edges must have exactly two.
\end{itemize}

\section{Change of variables for matrix integrals}\label{change}
\noindent In this Appendix, we prove the validity of the change of variables from \eqref{bc} to \eqref{model1}. That is, we show that
\begin{equation}\label{theo}
    \frac{1}{\mathcal N'}\int dM'\,e^{-N\mathrm{Tr}\left[\frac{1}{2}M'^\dagger M' + \frac{g}{2}(CM'^\dagger M')^2\right]}=\frac{1}{\mathcal N}\int dM\,e^{-N\mathrm{Tr}\left[\frac{1}{2}C^{-1}M^\dagger M + \frac{g}{2}(M^\dagger M)^2\right]}
    \,.
\end{equation}
We do so by performing a change of variables $M'\rightarrow M$ and $M'^\dagger\rightarrow C^{-1}M^\dagger$ with complex matrices $M$ and $M'$, where $C$ is a constant matrix\footnote{Not to be confused with $C_k$ defined in \eqref{cnprop}.}. \\
As a first step, let us analyze a simpler example equivalent to the case $N=1$, that is the integration of a single complex variable $z$ over the whole complex plane. This is
\begin{equation}
    \int_{\mathbb C} dz'\, e^{-\left[\frac{1}{2}z'^*z'+\frac{g}{2}(c \,z'^*z')^2\right]}=c^{-1}\int dz\, e^{-\left[\frac{1}{2}c^{-1}\,z^*z+\frac{g}{2}(z^*z)^2\right]}\;.
\end{equation}
for a nonzero complex number $c$. Starting with the integral
\begin{equation}
    I=\int_{\mathbb C} dz'\, e^{-\left[\frac{1}{2}z'^*z'+\frac{g}{2}(c \,z'^*z')^2\right]},
\end{equation}
we separate the r.h.s. into integrals of the real part and imaginary parts of $z'$,
\begin{equation}\label{eq:I_xy}
    I=\int_{\mathbb R \times \mathbb R} dx'dy'\, e^{-\left[\frac{1}{2}(x'-iy')(x'+iy')+\frac{g}{2}(c \,(x'-iy')(x'+iy'))^2\right]}\;.
\end{equation}
Then, we define two new complex numbers $x$ and $y$ such that
\begin{equation}
    x+iy = x'+iy'\qquad\mathrm{and}\qquad x-iy=c\,(x'-iy')\;.
\end{equation}
These are given in terms of $x'$ and $y'$ by
\begin{equation}
    x= \frac{1+c}{2}x'+\frac{1-c}{2}iy'\qquad\mathrm{and}\qquad y=\frac{1-c}{2i}x'+\frac{1+c}{2}y'\;,
\end{equation}
or inversely,
\begin{equation}
    x'= \frac{1+c^{-1}}{2}x+\frac{1-c^{-1}}{2}iy\qquad\mathrm{and}\qquad y'=\frac{1-c^{-1}}{2i}x+\frac{1+c^{-1}}{2}y\;.
\end{equation}
With this definition, we make a change variables in \eqref{eq:I_xy} from $y'$ to $y$. This is
\begin{equation}
    I=\frac{2}{1+c}\int_{\mathbb R} dx'\int_{\mathbb L_{x'}  } dy\, e^{-\left[\frac{1}{2}(x'-iy'(y))(x'+iy'(y))+\frac{g}{2}(c \,(x'-iy'(y))(x'+iy'(y)))^2\right]}\;,
\end{equation}
where the contour $\mathbb L_{x'} = \frac{1-c}{2i}x'+\frac{1+c}{2}\mathbb R$. An example of this contour is given in Figure~\ref{Lxpfig}.
\begin{figure}[H]
    \centering
    \begin{subfigure}[t]{0.45\textwidth}
    \centering
        \includegraphics[width=
    0.8\linewidth]{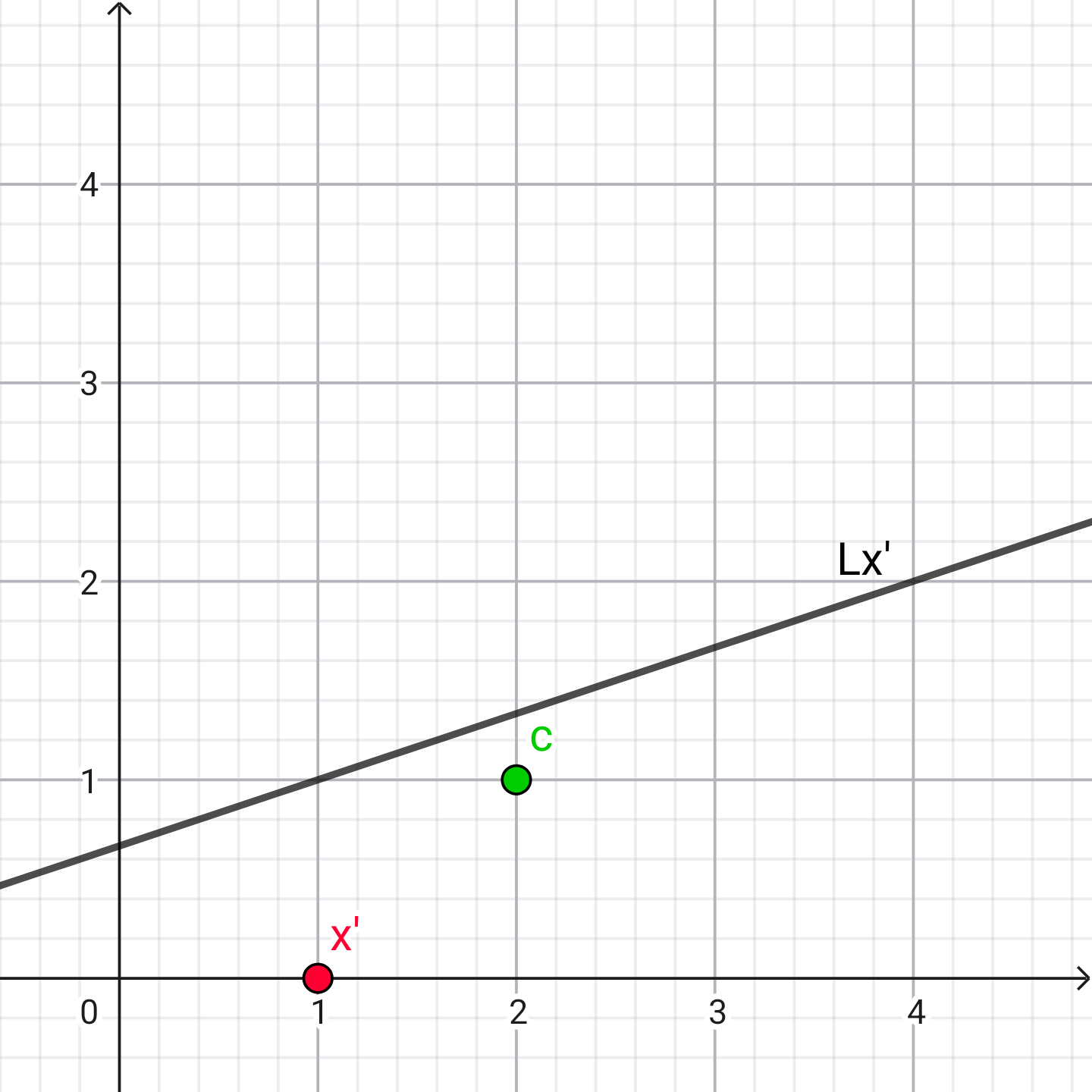}
    \caption{The complex plane featuring an example of the contour $\mathbb L_{x'}$ for given $x'$ and $c$. The slope depends only on $c$.}
    \label{Lxpfig}
    \end{subfigure}
    \quad
    \begin{subfigure}[t]{0.45\textwidth}
    \centering
    \includegraphics[width=0.8\linewidth]{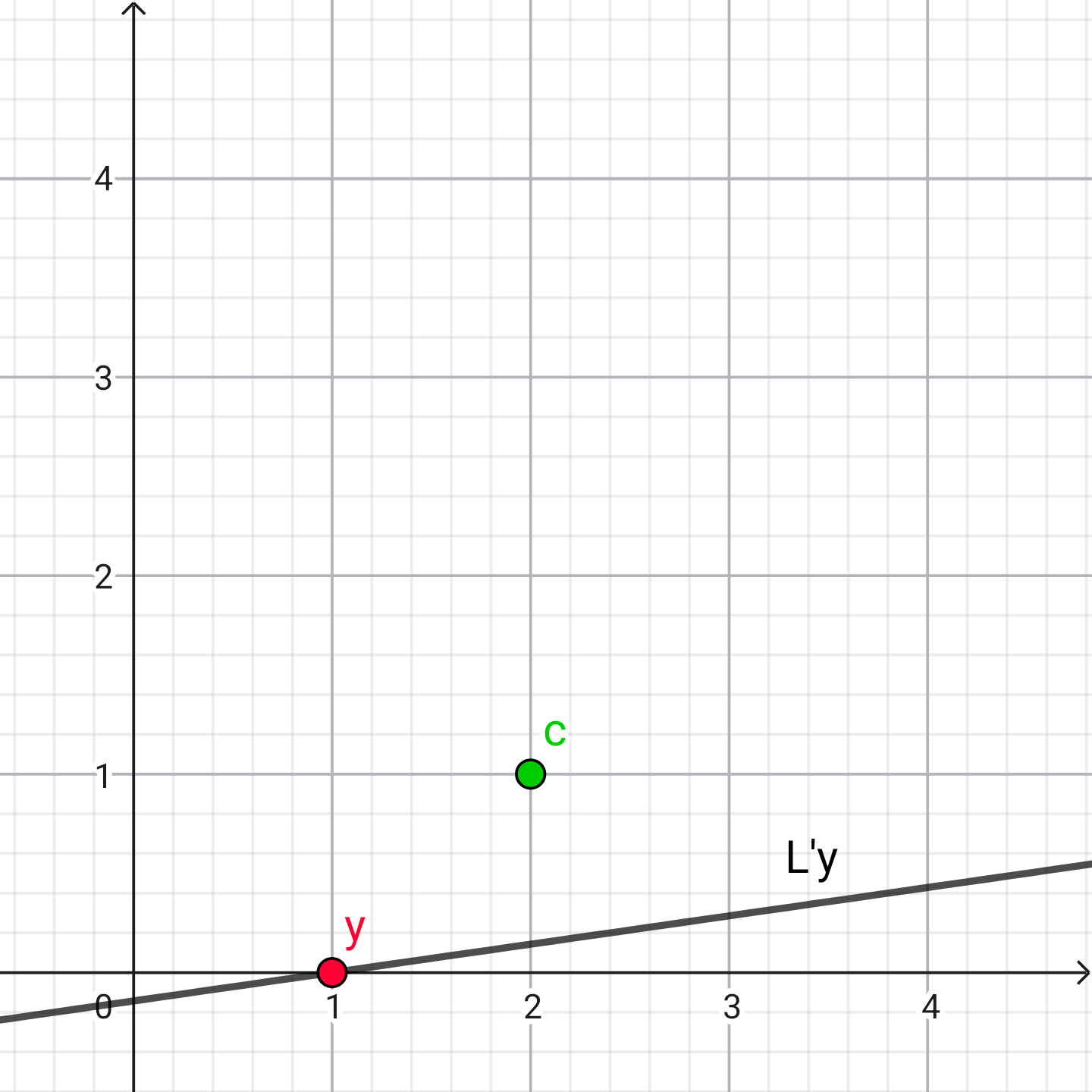}
    \caption{The complex plane featuring an example of the contour $\mathbb L'_{y}$ for given $y$ and $c$. The slope depends only on $c$.}
    \label{Lpyfig}
    \end{subfigure}
\end{figure}
\noindent By changing the contour from $\mathbb L_{x'}$ to $\mathbb R$, we find that
\begin{equation}
    I=\frac{2}{1+c}\int_{\mathbb R} dx'\int_{\mathbb R  } dy\, e^{-\left[\frac{1}{2}(x'-iy'(y))(x'+iy'(y))+\frac{g}{2}(c \,(x'-iy'(y))(x'+iy'(y)))^2\right]}\;.
\end{equation}
Similarly, we make a change variables from $x'$ to $x$ to obtain
\begin{equation}
    I=c^{-1}\int_{\mathbb R  } dy\int_{\mathbb L'_{y}} dx\, e^{-\left[\frac{1}{2}c^{-1}(x-iy)(x+iy)+\frac{g}{2}((x-iy)(x+iy))^2\right]}\;,
\end{equation}
where $\mathbb L'_{y}=\frac{2c}{1+c}\mathbb R+\frac{1-c}{1+c}y$. An example of this contour is given in Figure~\ref{Lpyfig}.\\
As a last step, we change the contour from $\mathbb L'_{y}$ to $\mathbb R$ to find
\begin{equation}
    I=c^{-1}\int_{\mathbb R  } dy\int_{\mathbb R} dx\, e^{-\left[\frac{1}{2}c^{-1}(x-iy)(x+iy)+\frac{g}{2}((x-iy)(x+iy))^2\right]}\;,
\end{equation}
or simply
\begin{equation}
    I=c^{-1}\int dz\, e^{-\left[\frac{1}{2}c^{-1}\,z^*z+\frac{g}{2}(z^*z)^2\right]}\;.
\end{equation}
For the case of complex matrices, we aim to perform a change of variables in the integral
\begin{equation}\label{eq:I_matrix}
    I=\iint_{\Cp^{N\times N}} dM'\,e^{-N\mathrm{Tr}\left[\frac{1}{2}M'^\dagger M' + \frac{g}{2}(CM'^\dagger M')^2\right]}\;,
\end{equation}
where $C$ is a constant matrix\footnote{Not to be confused with $C_k$ with the condition \eqref{cnprop}.}. Following an equivalent set of change of variables in the previous case, we write $M'_{ij}=X'_{ij}+i Y'_{ij}$. The integral \eqref{eq:I_matrix} becomes
\begin{equation}
    I=\int_{\mathbb R^{N\times N}} dX'\int_{\mathbb R^{N\times N}}dY'\,e^{-N\mathrm{Tr}\left[\frac{1}{2}(X'-iY')^T(X'+iY') + \frac{g}{2}(C(X'-iY')^T(X'+iY'))^2\right]}\;.
\end{equation}
Then, we introduce two complex matrices $X$ and $Y$ such that
\begin{equation}\label{mcv}
    X+iY = X'+iY'\qquad\mathrm{and}\qquad X-iY=(X'-iY')C\;.
\end{equation}
For this choice, we have used that $C=C^T$. Since $C$ is a diagonal matrix, $C(X'-iY')^T=[(X'-iY')C]^T$ is diagonal. The matrices $X$ and $Y$ are given explicitly in terms of $X'$, $Y'$ and the eigenvalues $c_j$ of $C$ by
\begin{equation}
    X_{ij}= X'_{ij}\frac{1+c_j}{2}+iY'_{ij}\frac{1-c_j}{2}\qquad\mathrm{and}\qquad Y_{ij}=X'_{ij}\frac{1-c_j}{2i}+Y'_{ij}\frac{1+c_j}{2}\;,
\end{equation}
or inversely,
\begin{equation}
    X'_{ij}= X_{ij}\frac{1+c_j^{-1}}{2}+iY_{ij}\frac{1-c_j^{-1}}{2}\qquad\mathrm{and}\qquad Y'_{ij}=X_{ij}\frac{1-c_j^{-1}}{2i}+Y_{ij}\frac{1+c_j^{-1}}{2}\;,
\end{equation}
With this definition, we make a change variables from $Y'$ to $Y$,
\begin{equation}
    I=\prod_{ij}\frac{2}{1+c_j}\int_{\mathbb R} dX'_{ij}\int_{\mathbb L_{X'_{ij}}} dY_{ij}\,e^{-N\mathrm{Tr}\left[\frac{1}{2}(X'-iY')^T(X'+iY') + \frac{g}{2}(C(X'-iY')^T(X'+iY'))^2\right]}\;,
\end{equation}
where $\mathbb L_{X'_{ij}} = \frac{1-c_j}{2i}X'_{ij}+\frac{1+c_j}{2}\mathbb R$. For simplicity, we avoid explicitly substituting $Y'$ in the integrand.\\
Changing the contour from $\mathbb L_{x'}$ to $\mathbb R$, we find that
\begin{equation}
    I=\det\left(\frac{2}{1+C}\right)^N\int_{\mathbb R^N} dX'\int_{\mathbb R^N} dY\, e^{-N\mathrm{Tr}\left[\frac{1}{2}(X'-iY')^T(X'+iY') + \frac{g}{2}(C(X'-iY')^T(X'+iY'))^2\right]}\;.
\end{equation}
Similarly, we make a change variables from $X'$ to $X$,
\begin{equation}
    I=\det\left(\frac{2}{1+C}\right)^N\prod_{ij}\frac{1+c^{-1}_j}{2}\int_{\mathbb R  } dY_{ij}\int_{\mathbb L'_{Y_{ij}}} dX_{ij}\, e^{-N\mathrm{Tr}\left[\frac{1}{2}C^{-1}(X-iY)^T(X+iY) + \frac{g}{2}((X-iY)^T(X+iY))^2\right]}\;,
\end{equation}
where $\mathbb L'_{Y_{ij}}=\frac{2}{1+c_j^{-1}}\mathbb R+\frac{1-c}{1+c}Y_{ij}$, and we used \eqref{mcv} to write the integrand explicitly in terms of $X$ and $Y$. Finally, we change the contour from $\mathbb L'_{y}$ to $\mathbb R$ to find
\begin{equation}
    I=\det(C)^{-N}\int_{\mathbb R^N  } dY\int_{\mathbb R^N} dX\,  e^{-N\mathrm{Tr}\left[\frac{1}{2}C^{-1}(X-iY)^T(X+iY) + \frac{g}{2}((X-iY)^T(X+iY))^2\right]}\;,
\end{equation}
or simply
\begin{equation}
    I=\det(C)^{-N}\iint_{\GL(N)} dM\, e^{-N\mathrm{Tr}\left[\frac{1}{2}C^{-1}M^\dagger M + \frac{g}{2}(M^\dagger M)^2\right]}\;.
\end{equation}
Therefore, we find \eqref{theo} with $\mathcal N = \det(C)^{-N}\mathcal{N}'$.

\section{On the convergence criteria on \texorpdfstring{$P$}{P} and \texorpdfstring{$Q$}{Q}}\label{convcrit}
\noindent In this Appendix, we analyze the conditions that the invertible  normal matrices $P$ and $Q$ must satisfy for the convergence of the integrals in the partition function
\begin{equation}\label{ptfpq}
    {\cal Z}_{\rm CM}[0](P,Q)=\int dM^\dagger dM\;e^{-N {\rm Tr}(P^{-1}M^\dagger Q^{-1} M )}\;.
\end{equation}
Since $P$ and $Q$ are normal, there exist diagonal matrices $P^{(d)}$ and $Q^{(d)}$, and unitary matrices $P^{(u)}$ and $Q^{(u)}$ such that
\begin{equation}
    P=P^{(u)}P^{(d)}{P^{(u)}}^{\dagger}\qquad{\rm{and}}\qquad Q=Q^{(u)}Q^{(d)}{Q^{(u)}}^{\dagger}\;.
\end{equation}
Since $P$ and $Q$ are invertible, all their eigenvalues are nonzero and $P^{(d)}$ and $Q^{(d)}$ are also invertible. By using this property, the partition function can be written as
\begin{equation}
     {\cal Z}_{\rm CM}[0](P,Q)=\int dM^\dagger dM\;e^{-N {\rm Tr}\lb {P^{(d)}}^{-1}({Q^{(u)}}^{\dagger} M P^{(u)})^{\dagger}{Q^{(d)}}^{-1}({Q^{(u)}}^{\dagger} M P^{(u)})\rb}\;.
\end{equation}
By performing the change of variables $M \rightarrow {Q^{(u)}} M {P^{(u)}}^{\dagger}$ and $M^\dagger \rightarrow {P^{(u)}} M^\dagger {Q^{(u)}}^{\dagger}$, we see that the model is equivalent to the one where $P$ and $Q$ are substituted by the diagonal matrices. This is
\begin{equation}
     {\cal Z}_{\rm CM}[0](P,Q)=\int dM^\dagger dM\;e^{-N {\rm Tr}\lb {P^{(d)}}^{-1}M^{\dagger}{Q^{(d)}}^{-1}M\rb}\;.
\end{equation}
Denoting $p_k$ and $q_k$ ($1\leq k \leq N$) the eigenvalues of $P$ and $Q$ respectively, we note that
\begin{equation}
     {\cal Z}_{\rm CM}[0](P,Q)=\int dM^\dagger dM\;e^{-N \sum_{k,l}\lb q^{-1}_kp^{-1}_l |M_{kl}|^2\rb}\;.
\end{equation}
Therefore, the convergence is achieved iff
\begin{equation}
    {{\rm Re}}(q_k p_l)\geq 0\qquad {\rm for} \qquad 1\leq k, l \leq N\;.
\end{equation}
To make sense of this geometrically, we rewrite this condition as 
\begin{equation}
    -\frac{\pi}{2}\leq{\rm Arg}(q_k p_l)\leq \frac{\pi}{2}\qquad {\rm for} \qquad 1\leq k, l \leq N\;.
\end{equation}
By looking at any given $k$, we see that all $p_l$ must be in the half-plane around the complex conjugate $q_k^*$. The same thing can be said about the distribution of $q_k$ around any $p_l^*$. This means that there exists a $\delta\in[0,2\pi]$ and an $\alpha \in [0, \pi/2]$ such that
\begin{equation}
    -\alpha\leq\Arg(p_k)+\delta\leq +\alpha
\end{equation}
and
\begin{equation}
    -\frac{\pi}{2}
    +\alpha\leq\Arg(q_k)-\delta\leq \frac{\pi}{2}-\alpha\;.
\end{equation}
We can assume that $\delta=0$, since we could either change $P\rightarrow e^{i\delta}P$ and $Q\rightarrow e^{-i\delta}Q$ or $M \rightarrow e^{-i\delta}$ and $M^\dagger \rightarrow e^{i\delta}$ to remove $\delta$. With this assumption, a possible distribution of $p_k$ and $q_l$ is illustrated in Figure~\ref{eigdist}.

\begin{figure}[H]
     \centering
    \includegraphics[width=
    0.3\linewidth]{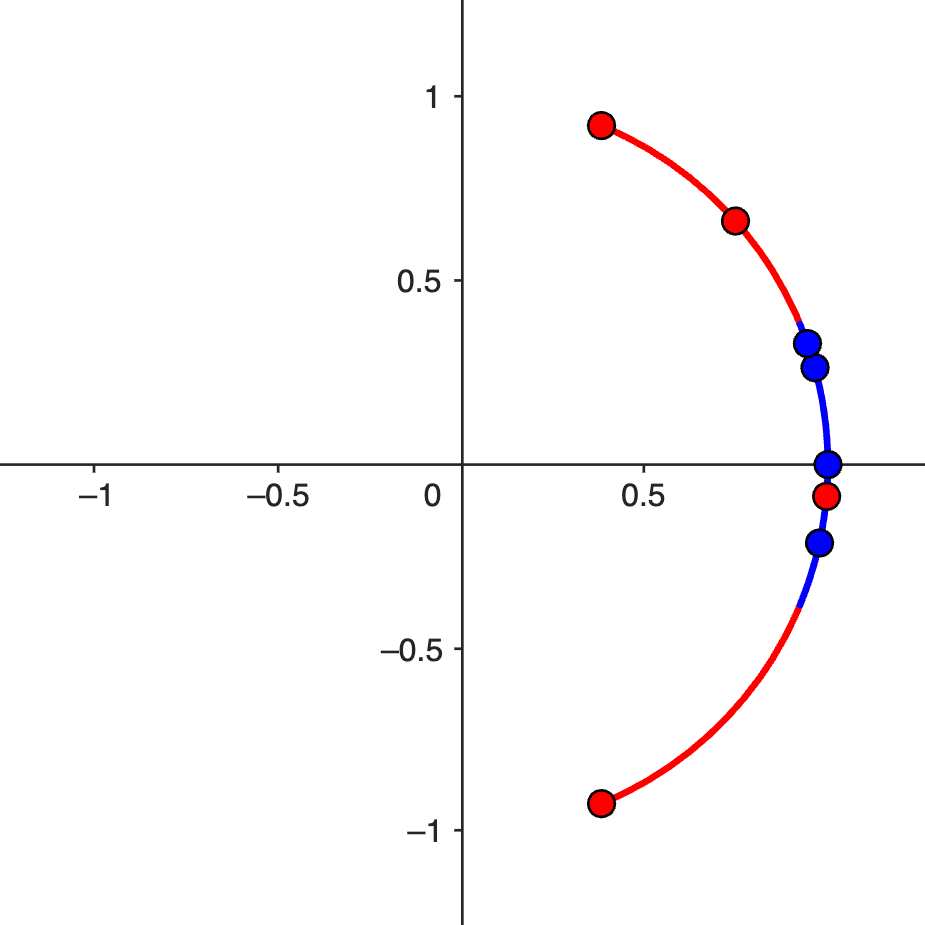}
    \caption{The complex plane of eigenvalues of $P$ and $Q$.
    An example of the distribution of eigenvalues of $P$ and $Q$. If the eigenvalues of $Q$ are distributed along the blue arc (including the end points) then the eigenvalues of $P$ must be distributed along the red arc (not necessarely including the end points) for the integral \eqref{ptfpq} to converge.
    We also show in dots
    an explicit example of eigenvalues of $P$ and $Q$ for which the integral \eqref{ptfpq} converges for $N=4$. The red (blue) points are the eigenvalues of $P$ ($Q$).
    }
    \label{eigdist}
\end{figure}

\newpage
\section{Proofs}\label{app:proof}
\noindent In this appendix we collect the proofs of several propositions  appearing  in the main text.

\paragraph{Proof of Proposition~\ref{avM}.}

\begin{proof}
We introduce a complex source matrix $J$ coupled linearly to $M$ by
\[
W[J,J^\dagger](M,M^\dagger)
  := N\Tr(M^\dagger J + J^\dagger M)\;.
\]
The corresponding generating functional is
\[
\mathcal Z_{\rm G}[J,J^\dagger](P,Q)
  = \iint {\cal D}M^\dagger\,{\cal D}M\;
    e^{-N\Tr(P^{-1}M^\dagger Q^{-1} M) + N\Tr(M^\dagger J + J^\dagger M)}.
\]
A standard Gaussian calculation gives
\begin{equation}\label{eq:ZGJJ}
\mathcal Z_{\rm G}[J,J^\dagger](P,Q)
  = \mathcal Z_{\rm G}[0](P,Q)\,
    e^{N\Tr(P J^\dagger Q J)}\;.
\end{equation}
The multi-trace invariant observable $\Tr_{[\sigma]}(M^\dagger M)$ can be generated by derivatives with respect to $J$ and $J^\dagger$ in the following form
\begin{equation}
    \left\langle \Tr_{[\sigma]} (M^\dagger M)\right\rangle[0](P)
    = N^{-2n}\,\bigl(\mathcal Z_{\rm G}[0](P,Q)\bigr)^{-1}
      \left.
      \Tr_{[\sigma]} \left[\frac{\partial}{\partial J^\dagger}
                          \frac{\partial}{\partial J}\right]^{ \tp}
      \mathcal Z_{\rm G}[J,J^\dagger](P,Q)
      \right|_{J=0}.
\end{equation}
Using \eqref{eq:ZGJJ},
\begin{equation}
    \left\langle \Tr_{[\sigma]} (M^\dagger M)\right\rangle[0](P)
    = N^{-2n}
      \left.
      \Tr_{[\sigma]} \left[\frac{\partial}{\partial J^\dagger}
                            \frac{\partial}{\partial J}\right]^{ \tp}
      e^{N\Tr(P J^\dagger Q J)}
      \right|_{J=0}.
\end{equation}
Since the operator $\Tr_{[\sigma]}(\partial/\partial J^\dagger\,\partial/\partial J)^{\tp}$ is homogeneous of degree $2n$ in the derivatives, only the term of order $n$ in the exponential contributes to the previous expression. This results in
\begin{equation}
\label{eq:ordernmultitrace}
    \left\langle \Tr_{[\sigma]} (M^\dagger M)\right\rangle[0](P)
    = \frac{1}{n!}\,N^{-n}\,
      \Tr_{[\sigma]} \left[\frac{\partial}{\partial J^\dagger}
                           \frac{\partial}{\partial J}\right]^{ \tp}
      \bigl(\Tr(P J^\dagger Q J)\bigr)^n.
\end{equation}
Now, we expand both the multi-trace 
invariant operator and the $n$th power
of $\Tr(P J^\dagger Q J)$ appearing in r.h.s. of \eqref{eq:ordernmultitrace}.
Using the definition of $\Tr_{[\sigma]}$, we write
\begin{equation}
\label{eq:matrix-expanded}
\Tr_{[\sigma]} \left[\frac{\partial}{\partial J^\dagger}
                      \frac{\partial}{\partial J}\right]^{ \tp}
\bigl(\Tr(P J^\dagger Q J)\bigr)^n
=
\sum_{\bm t,\bm t'}
\prod_{i=1}^n
\frac{\partial}{\partial J_{t'_i t_i}}
\frac{\partial}{\partial J^\dagger_{t_{\sigma(i)} t'_i}}
\;
\sum_{\bm k,\bm k',\bm \ell,\bm \ell'}
\prod_{j=1}^n
J_{k'_jk_j}P_{k_j\ell_j}
J^\dagger_{\ell_j\ell'_j}Q_{\ell'_j k'_j},
\end{equation}
where $\bm t=(t_1,\dots,t_n)$, $\bm t'=(t'_1,\dots,t'_n)$, and similarly for
$\bm k,\bm k',\bm \ell,\bm \ell'$.

Applying the multivariate Leibniz rule to distribute the derivatives
among the $n$ factors of $J$ and $J^\dagger$ produces a sum over two
permutations $\mu,\gamma\in S_n$ specifying which factor is hit by
each derivative:
\begin{equation}
\begin{aligned}
 &\prod_{i=1}^n
   \frac{\partial}{\partial J_{t'_i t_i}}
   \frac{\partial}{\partial J^\dagger_{t_{\sigma(i)} t'_i}}
   \prod_{j=1}^n
   J_{k'_jk_j}P_{k_j\ell_j}
   J^\dagger_{\ell_j\ell'_j}Q_{\ell'_j k'_j}
\\[2mm]
 &= \sum_{\mu,\gamma\in S_n}
    \prod_{i=1}^n
    \frac{\partial}{\partial J_{t'_i t_i}}
        J_{k'_{\mu(i)}k_{\mu(i)}}P_{k_{\mu(i)}\ell_{\mu(i)}}
    \frac{\partial}{\partial J^\dagger_{t_{\sigma(i)} t'_i}}
        J^\dagger_{\ell_{\gamma(i)}\ell'_{\gamma(i)}}Q_{\ell'_{\gamma(i)}k'_{\gamma(i)}}.
\end{aligned}
\end{equation}
Each factor can be differentiated explicitly, leading to
\begin{equation}
\frac{\partial}{\partial J_{t'_i t_i}}
 J_{k'_{\mu(i)}k_{\mu(i)}}
= \delta_{k'_{\mu(i)},t'_i}\,\delta_{k_{\mu(i)},t_i},\qquad
\frac{\partial}{\partial J^\dagger_{t_{\sigma(i)} t'_i}}
 J^\dagger_{\ell_{\gamma(i)}\ell'_{\gamma(i)}}
= \delta_{\ell_{\gamma(i)},t_{\sigma(i)}}\,\delta_{\ell'_{\gamma(i)},t'_i}.
\end{equation}
Substituting these expressions back into \eqref{eq:matrix-expanded} and performing the sums over $J$–indices yields a product of Kronecker deltas, we obtain
\begin{equation}
    \Tr_{[\sigma]} \left[\frac{\partial}{\partial J^\dagger}
                          \frac{\partial}{\partial J}\right]^{ \tp}
    \bigl(\Tr(P J^\dagger Q J)\bigr)^n
    = \sum_{\mu,\gamma\in S_n}
      \sum_{\bm t,\bm \ell,\bm k'}
      \prod_{i=1}^n
      P_{t_i,\ell_{\mu(i)}}\,
      \delta_{\ell_{\gamma(i)},t_{\sigma(i)}}\,
      Q_{k'_{\mu(i)},k'_{\gamma(i)}}.
\end{equation}
Evaluating the sums over $t_i$ and relabelling indices in a standard
way reduces the last expression to
\begin{equation}
    \Tr_{[\sigma]} \left[\frac{\partial}{\partial J^\dagger}
                          \frac{\partial}{\partial J}\right]^{ \tp}
    \bigl(\Tr(P J^\dagger Q J)\bigr)^n
    = n!\sum_{\gamma\in S_n}
       \sum_{\bm \ell,\bm k'}
       \prod_{i=1}^n
       P_{\ell_i,\ell_{\gamma\sigma(i)}}\,
       Q_{k'_i,k'_{\gamma(i)}}.
\end{equation}
Using the definition of the multi-traces, the remaining sums factorize,
\[
\sum_{\bm \ell}\prod_{i=1}^n P_{\ell_i,\ell_{\gamma\sigma(i)}}
   = \Tr_{[\gamma\sigma]}(P),\qquad
\sum_{\bm k'}\prod_{i=1}^n Q_{k'_i,k'_{\gamma(i)}}
   = \Tr_{[\gamma]}(Q).
\]
Hence
\begin{equation}
    \Tr_{[\sigma]} \left[\frac{\partial}{\partial J^\dagger}
                          \frac{\partial}{\partial J}\right]^{ \tp}
    \bigl(\Tr(P J^\dagger Q J)\bigr)^n
    = n!\sum_{\gamma\in S_n}
      \Tr_{[\gamma]}(Q)\,\Tr_{[\gamma\sigma]}(P).
\end{equation}
Substituting back into the expression for the expectation value \eqref{eq:ordernmultitrace}, we obtain
\[
    \left\langle \Tr_{[\sigma]} (M^\dagger M)\right\rangle[0](P)
    = \frac{1}{n!}N^{-n}\,
      n!\sum_{\gamma\in S_n}
      \Tr_{[\gamma]}(Q)\,\Tr_{[\gamma\sigma]}(P)
    = N^{-n}\sum_{\gamma\in S_n}
      \Tr_{[\gamma]}(Q)\,\Tr_{[\gamma\sigma]}(P),
\]
which is precisely the statement of Proposition~\ref{avM}.
\end{proof}

\paragraph{Proof of Proposition~\ref{avAB}.}

\begin{proof}
As in Section~\ref{subsec:HMmodel}, repeated integration by parts shows that for any $\sigma\in S_n$ and any potential $Y$ one has
\begin{equation}\label{eq:HM-deriv}
\left\langle \Tr_{[\sigma]}(A)\right\rangle[0,Y]
    = (iN)^{-n}\,
      \Tr_{[\sigma]} \left(\frac{\partial}{\partial B}\right)
      e^{- \Tr\ln\bigl(\mathbb 1^{\otimes 2}- i\,Q\otimes(PB)\bigr)}\Big|_{B=0}.
\end{equation}
Here, we use the expansion
\[
e^{- \Tr\ln\bigl(\mathbb 1^{\otimes 2}- i\,Q\otimes(PB)\bigr)}
 = \sum_{m=0}^\infty \frac{1}{m!}
      \sum_{\mu\in S_m}
      \Tr_{[\mu]}\bigl(i\,Q\otimes(PB)\bigr)= \sum_{m=0}^\infty \frac{i^m}{m!}
      \sum_{\mu\in S_m} \Tr_{[\mu]}(Q)
      \Tr_{[\mu]}\bigl(PB\bigr),
\]
 Substituting this into \eqref{eq:HM-deriv} gives
\begin{equation}
\left\langle \Tr_{[\sigma]}(A)\right\rangle[0,Y_{\ln}[P,Q]]
  = (iN)^{-n}
    \sum_{m=0}^\infty \frac{i^m}{m!}\sum_{\mu\in S_m}
    \Tr_{[\sigma]} \left(\frac{\partial}{\partial B}\right)
    \Tr_{[\mu]}(Q)
      \Tr_{[\mu]}\bigl(PB\bigr)|_{B=0}.
\end{equation}
The operator $\Tr_{[\sigma]}(\partial/\partial B)$ is homogeneous of degree $n$ in the derivatives, so only the term $m=n$ contributes. This leads to
\begin{equation}\label{eq:HM-precomb}
\left\langle \Tr_{[\sigma]}(A)\right\rangle[0,Y_{\ln}[P,Q]]
  = \frac{1}{n!}\,N^{-n}
    \sum_{\mu\in S_n}
    \Tr_{[\sigma]} \left(\frac{\partial}{\partial B}\right)
     \Tr_{[\mu]}(Q)
      \Tr_{[\mu]}(PB)\;.
\end{equation}
Since $\Tr_{[\mu]}(Q)$ does not depend on $B$, the derivative acts only on 
$\Tr_{[\mu]}(PB)$.  Using
\[
\Tr_{[\mu]}(PB)
  = \sum_{\bm a,\bm c}\;
    \prod_{j=1}^n P_{a_j c_j}\,B_{c_j a_{\mu(j)}}\;,
\qquad
\Tr_{[\sigma]} \left(\frac{\partial}{\partial B}\right)
  = \sum_{\bm k}\prod_{j=1}^n
     \left[\frac{\partial}{\partial B}\right]_{k_jk_{\sigma(j)}},
\]
where $\bm a=(a_1,\dots,a_n)$, and similarly for $\bm c$ and $\bm k$, the product rule gives as a result
\begin{align}
\Tr_{[\sigma]} \left(\frac{\partial}{\partial B}\right)\Tr_{[\mu]}(PB)
 &= 
 \sum_{\nu\in S_n}
 \sum_{\bm k}\sum_{\bm a,\bm c}
 \prod_{j=1}^n
  P_{a_{\nu(j)}c_{\nu(j)}}
  \left[\frac{\partial}{\partial B}\right]_{k_jk_{\sigma(j)}}
  B_{c_{\nu(j)}a_{\mu(\nu(j))}} \nonumber\\[4pt]
 &= \sum_{\nu\in S_n}
    \sum_{\bm k}\sum_{\bm a,\bm c}
    \prod_{j=1}^n
      P_{a_{\nu(j)}c_{\nu(j)}}
      \delta_{k_{\sigma(j)},\,a_{\mu(\nu(j))}}
      \delta_{k_j,\,c_{\nu(j)}}\;.
\end{align}
The $\bm a$- and $\bm c$-sums collapse by using the deltas, leaving only a constraint among 
the $\bm k$--indices.  Then, one obtains
\[
\Tr_{[\sigma]} \left(\frac{\partial}{\partial B}\right)\Tr_{[\mu]}(PB)
  = \sum_{\nu\in S_n}
    \sum_{\bm k}\prod_{j=1}^n P_{k_{\sigma\nu^{-1}\mu^{-1}\nu(j)}\,k_{j}}\;,
\]
where the sum over $\bm a$ is precisely the trace invariant
\[
\sum_{\bm k}\prod_{j=1}^n P_{k_{\sigma\nu^{-1}\mu^{-1}\nu(j)}\,k_{j}}
  = \Tr_{[\sigma\nu^{-1}\mu^{-1}\nu]}(P).
\]
Substituting into~\eqref{eq:HM-precomb} yields
\begin{equation}
\left\langle \Tr_{[\sigma]}(A)\right\rangle[0,Y_{\ln}[P,Q]]
  = \frac{1}{n!}\,N^{-n}
    \sum_{\mu,\nu\in S_n}
     \Tr_{[\mu]}(Q)\Tr_{[\sigma\nu^{-1}\mu^{-1}\nu]}(P)
     \;.
\end{equation}
With the substitution $\mu\rightarrow \nu\mu^{-1}\nu^{-1}$ and summing over $\nu$,  we find
\[
\left\langle \Tr_{[\sigma]}(A)\right\rangle[0,Y_{\ln}[P,Q]]
  = N^{-n}
    \sum_{\mu\in S_n}\Tr_{[\mu]}(Q)\,\Tr_{[\mu\sigma]}(P),
\]
which is the desired result.
\end{proof}

\paragraph{Proof of Proposition~\ref{avTM}.}

\begin{proof}
Let $\psi\in\mathbb C_N^{\otimes D}$ be an external source tensor, and consider the linear source term
\[
W[\psi,\psi^\dagger](\phi,\phi^\dagger)
    := \psi^\dagger\phi + \phi^\dagger\psi \;.
\]
The corresponding generating functional is
\begin{equation}
\label{eq:ZCTRW}
     {\cal Z}_{\rm CT}
     [W[\psi,\psi^\dagger]](R)
    =\iint {\cal D}\phi^\dagger {\cal D}\phi\;
      e^{-\phi^\dagger 
      R^{-1}
      \phi+\psi^\dagger\phi+\phi^\dagger\psi}\;.
\end{equation}
By standard Gaussian integration, one has that
\begin{equation}
    {\cal Z}_{\rm CT}
    [W[\psi,\psi^\dagger]](R)
      = {\cal Z}_{\rm CT}[0]
      (R)
      \,
        e^{\psi^\dagger 
        R
        \psi}\;.
\end{equation}
The expectation values of trace invariants of tensors $\Tr_{[{\bm\sigma}]}(\phi\phi^\dagger)$ can be generated by differentiating with respect to the sources. More precisely,
\begin{equation}
    \langle\Tr_{[{\bm\sigma}]}(\phi\phi^\dagger)\rangle[0]
    (R)
    = ({\cal Z}_{\rm CT}[0]
    (R))^{-1}
    \,
      \left.\Tr_{[{\bm\sigma}]} \left(
      \frac{\partial}{\partial \psi^{\dagger}}\,
      \frac{\partial}{\partial \psi}
      \right)
      {\cal Z}_{\rm CT}
      (R)
      [W[\psi,\psi^\dagger]]
      \right|_{\psi=0}\;.
\end{equation}
Using the explicit form of ${\cal Z}_{\rm CT}[W](R)$ \eqref{eq:ZCTRW},
\begin{equation}
    \langle\Tr_{[{\bm\sigma}]}(\phi\phi^\dagger)\rangle[0] (R)= \left.
      \Tr_{[{\bm\sigma}]} \left(
        \frac{\partial}{\partial \psi^{\dagger}}\,
        \frac{\partial}{\partial \psi}
      \right)
      e^{\psi^\dagger R \psi} \right|_{\psi=0}\;.
\end{equation}
Only the term of total degree $n$ in $\psi$ and $\psi^\dagger$ contributes to this expression, hence
\begin{equation}
    \langle\Tr_{[{\bm\sigma}]}(\phi\phi^\dagger)\rangle[0]
    (R)
    = \frac{1}{n!}\,
      \Tr_{[{\bm\sigma}]} \left(
        \frac{\partial}{\partial \psi^{\dagger}}\,
        \frac{\partial}{\partial \psi}
      \right)
      (\psi^\dagger  R
      \psi)^n\;.
\end{equation}
Now, we expand both the trace invariant
and the $n$th power of the quadratic form. 
Writing multi-indices as $\bm a,\bm b,\bm k$ and using the definition of $\Tr_{[{\bm\sigma}]}$, we obtain
\begin{equation}
     \langle\Tr_{[{\bm\sigma}]}(\phi\phi^\dagger)\rangle[0]
     (R)
     = \frac{1}{n!}
       \Biggl(\sum_{\bm k}
              \prod_{j=1}^n
              \left[\frac{\partial}{\partial \psi}\right]_{\bm k_j}
              \left[\frac{\partial}{\partial \psi^{\dagger}}\right]^{(\bm{\sigma}_*\bm k)_j}
       \Biggr)
       \Biggl(\sum_{\bm a,\bm b}
              \prod_{l=1}^n
              \psi^\dagger_{\bm a_l} 
              R^
              {\bm a_l}_{\bm b_l}\psi^{\bm b_l}
       \Biggr).
\end{equation}
Applying the (multi-variable) Leibniz rule and distributing the derivatives among the $n$ factors of $\psi$ and $\psi^\dagger$ generates the following sum over permutations $\mu,\nu\in S_n$
\begin{equation}
     \langle\Tr_{[{\bm\sigma}]}(\phi\phi^\dagger)\rangle[0]
     (R)
     = \frac{1}{n!}
       \sum_{\mu,\nu\in\Sy_n}
       \sum_{\bm k}
       \sum_{\bm a,\bm b}
       \prod_{j=1}^n
       \delta_{\bm a_{\mu(j)}}^{(\bm{\sigma}_*\bm k)_j}\,
       R^{\bm a_j}_{\bm b_j}\,
       \delta^{\bm b_{\nu(j)}}_{\bm k_j}\;.
\end{equation}
Summing over $\bm a$ and $\bm b$ enforces these Kronecker deltas and yields
\begin{equation}
     \langle\Tr_{[{\bm\sigma}]}(\phi\phi^\dagger)\rangle[0]
     (R)
     = \frac{1}{n!}
       \sum_{\mu,\nu\in\Sy_n}
       \sum_{\bm k}
       \prod_{j=1}^n
         R_{\bm k_j}^{(\bm{\sigma}_*\bm k)_{\mu^{-1}\nu(j)}}\;.
\end{equation}
A change of variables $\mu\mapsto \nu\mu$ removes the sum over $\nu$, this results in
\begin{equation}
     \langle\Tr_{[{\bm\sigma}]}(\phi\phi^\dagger)\rangle[0]
     (R)
     = \sum_{\mu\in\Sy_n}
       \sum_{\bm k}
       \prod_{j=1}^n
        R_{\bm k_j}^{(\bm{\sigma}_*\bm k)_{\mu^{-1}(j)}}\;.
\end{equation}
Finally, by the very definition of the action of a multi-permutation on multi-indices,
\[
(\bm{\sigma}_*\bm k)_{\mu^{-1}(j)}
    = ((\mu\bm{\sigma})_*\bm k)_{j}\,,
\]
we have that
\begin{equation}
     \langle\Tr_{[{\bm\sigma}]}(\phi\phi^\dagger)\rangle[0](R)
     = \sum_{\mu\in\Sy_n}
       \sum_{\bm k}
       \prod_{j=1}^n
         R_{\bm k_j}^{((\mu\bm{\sigma})_*\bm k)_{j}}\;.
\end{equation}
The right-hand side is exactly the definition of the trace invariant $\Tr_{[\mu\bm\sigma]}(R)$. Furthermore, summing over $\mu$ results in
\[
    \langle\Tr_{[{\bm\sigma}]}(\phi\phi^\dagger)\rangle[0]
    (R)
    = \sum_{\mu\in\Sy_n}\Tr_{[\mu\bm\sigma]}
    (R)
    \;,
\]
which is the claimed result.
\end{proof}

\paragraph{Proof of Proposition~\ref{avPhiPsi}.}
\begin{proof}
From \eqref{eq:dpsi} we have, for $\bm\sigma\in S_n^D$,
\begin{equation}
    \langle\Tr_{[{\bm\sigma}]}(\Phi)\rangle[0,Y]
    = i^{-n}\,\Tr_{[{\bm\sigma}]} \left ( \frac{\partial}{\partial \Psi} \right)
      e^{-\Tr\ln\bigl( {\bb 1}^{\otimes D}-i 
      R\Psi\bigr)}\Big|_{\Psi=0}\;.
\end{equation}
Expanding the exponential and the logarithm, we only need the term of total degree $n$ in $\Psi$. This is
\begin{equation}
    e^{-\Tr\ln\bigl( {\bb 1}^{\otimes D}-i 
    R
    \Psi\bigr)}
    = \sum_{m=0}^\infty\frac{1}{m!}
      \sum_{\mu\in\Sy_m}\Tr_{[\mu]}(i 
      R
      \Psi)\;,
\end{equation}
so that
\begin{equation}
    \langle\Tr_{[{\bm\sigma}]}(\Phi)\rangle[0,Y]
    = \frac{1}{n!}\sum_{\mu\in \Sy_n}
      \Tr_{[{\bm\sigma}]} \left ( \frac{\partial}{\partial \Psi} \right)
      \Tr_{[\mu]}\bigl(
      R
      \Psi\bigr)\;.
\end{equation}
Using the explicit definitions of the trace invariants, we write
\begin{equation}
      \langle\Tr_{[{\bm\sigma}]}(\Phi)\rangle[0,Y]
      = \frac{1}{n!}\sum_{\mu\in\Sy_n}
        \Biggl(\sum_{\bm k}
               \prod_{j=1}^n
               \left[\frac{\partial}{\partial \Psi}\right]^{(\bm{\sigma}_*\bm k)_j}_{\bm k_j}
        \Biggr)
        \Biggl(\sum_{\bm a,\bm b}
               \prod_{l=1}^n
               R^{\bm a_{\mu(l)}}_{\bm b_l}
               \Psi^{\bm b_l}_{\bm a_{l}}
        \Biggr)\;.
\end{equation}
Applying the distributive (multi-variable) Leibniz rule and distributing the derivatives among the factors of $\Psi$ produces a sum over permutations $\nu\in S_n$. This results in
\begin{equation}
     \langle\Tr_{[{\bm\sigma}]}(\Phi)\rangle[0,Y]
     = \frac{1}{n!}
       \sum_{\mu\in\Sy_n}\sum_{\nu\in\Sy_n}
       \sum_{\bm k}
       \sum_{\bm a,\bm b}
       \prod_{j=1}^n
       R^{\bm a_{\mu(j)}}_{\bm b_j}\,
       \delta_{\bm a_{\nu(j)}}^{(\bm{\sigma}_*\bm k)_j}\,
       \delta^{\bm b_{\nu(j)}}_{\bm k_j}\;.
\end{equation}
Summing over $\bm a$ and $\bm b$ gives
\begin{equation}
    \langle\Tr_{[{\bm\sigma}]}(\Phi)\rangle[0,Y]
    = \frac{1}{n!}
      \sum_{\mu\in\Sy_n}\sum_{\nu\in\Sy_n}
      \sum_{\bm k}
      \prod_{j=1}^n
      R^{(\bm{\sigma}_*\bm k)_{\nu^{-1}\mu\nu(j)}}_{\bm k_j}\;.
\end{equation}
Since a conjugation change of variables $\mu\mapsto \nu\mu\nu^{-1}$ removes the sum over $\nu$, we obtain
\begin{equation}
      \langle\Tr_{[{\bm\sigma}]}(\Phi)\rangle[0,Y]
      = \sum_{\mu\in\Sy_n}
        \sum_{\bm k}
        \prod_{j=1}^n
        R_{\bm k_j}^{(\bm{\sigma}_*\bm k)_{\mu(j)}}\;.
\end{equation}
Replacing $\mu$ by its inverse (which runs again over all of $S_n$) gives
\begin{equation}
      \langle\Tr_{[{\bm\sigma}]}(\Phi)\rangle[0,Y]
      = \sum_{\mu\in\Sy_n}
        \sum_{\bm k}
        \prod_{j=1}^n
        R_{\bm k_j}^{(\bm{\sigma}_*\bm k)_{\mu^{-1}(j)}}\;.
\end{equation}
As before, the action of a multi-permutation on multi-indices satisfies
\[
(\bm{\sigma}_*\bm k)_{\mu^{-1}(j)}
  = \bigl((\mu\bm{\sigma})_*\bm k\bigr)_{j}\,.
\]
Therefore, we obtain
\begin{equation}
      \langle\Tr_{[{\bm\sigma}]}(\Phi)\rangle[0,Y]
      = \sum_{\mu\in\Sy_n}
        \sum_{\bm k}
        \prod_{j=1}^n
        R_{\bm k_j}^{((\mu\bm{\sigma})_*\bm k)_{j}}\;.
\end{equation}
The product in the last expression is precisely the definition of the trace invariant
$\Tr_{[\mu\bm\sigma]}(R)$, hence
\[
\left\langle \Tr_{[\bm\sigma]}(\Phi)\right\rangle[0,Y]
    = \sum_{\mu\in S_n}
      \Tr_{[\mu\bm\sigma]}
      (R)\;,
\]
which is the desired formula.
\end{proof}

\section{Equivalences for the quartic models}\label{quaex}
\noindent In this Appendix, we present explicit demonstrations of the equivalences we proved for quartic 
complex
matrix (Theorem~\ref{matrixequiv_thm}) and tensor models (Theorems~\ref{mainthm} and~\ref{theorem:partialtensor}). These examples are particularly instructive, as they show how the presence of the matrix $C_2$ \eqref{cnprop} in the propagator effectively reduces a quartic interaction to a quadratic/Gaussian model. 

\subsection{Quartic matrix models}\label{app:ampli}
\noindent We start by writing the partition function \eqref{model1} as a power series in~$g$,
\begin{equation}\label{eq:Zmodel1}
    \mathcal{Z}(g)
    =\sum_{k=0}^{\infty}\frac{1}{k!}\left(-\frac{Ng}{2}\right)^k
    \bigl\langle\bigl(\Tr[(M^\dagger M)^2]\bigr)^k\bigr\rangle\;,
\end{equation}
where
\begin{equation}
    \langle f(M)\rangle
    =\frac{1}{\mathcal N}\int dM\,f(M)\,
    e^{-N\Tr\left[M^\dagger C_2^{-1} M\right]}
\end{equation}
denotes the Gaussian average of any function $f$. The two–point functions are
\begin{equation}
    \langle M_{ij}^\dagger M_{kl}\rangle
    = \frac{1}{N}(C_2)_{il}\delta_{kj}\;,
\end{equation}\label{propeq}
while
\begin{equation}
    \langle M_{ij} M_{kl}\rangle = 0\;, \qquad
    \langle M_{ij}^\dagger M_{kl}^\dagger\rangle = 0\;.
\end{equation}\label{nopropeq}
In order to restrict the expansion \eqref{eq:Zmodel1} to connected ribbon graphs, we work with the free energy
\begin{equation}
    \mathcal{F}(g)=\ln \frac{\mathcal Z(g)}{\mathcal Z(0)}
    =\sum_{k=0}^{\infty}\frac{1}{k!}\left(-\frac{Ng}{2}\right)^k
    \bigl\langle\bigl(\Tr[(M^\dagger M)^2]\bigr)^k\bigr\rangle_{c}\;,
\end{equation}\label{fre}
where $\langle \cdot\rangle_{c}$ denotes connected ribbon graphs.

We proceed by analyzing the combinatorics of a single quartic ribbon-graph vertex. Each matrix $M^\dagger$ must contract with a matrix $M$. If it contracts with the matrix $M$ immediately to its right, we obtain the configuration in Figure~\ref{vertprop1}. The corresponding propagator forms a ribbon-graph face that carries $\Tr C_2$, while the enclosed face contributes $\Tr \mathbb{1}=N$. Since we assume $\Tr C_2 = 0$, the amplitude of such graphs vanishes, meaning that they have weight zero.

\begin{figure}[h]
    \centering
    \begin{subfigure}{0.3\textwidth}
    \centering
    \includegraphics[width=0.6\linewidth]{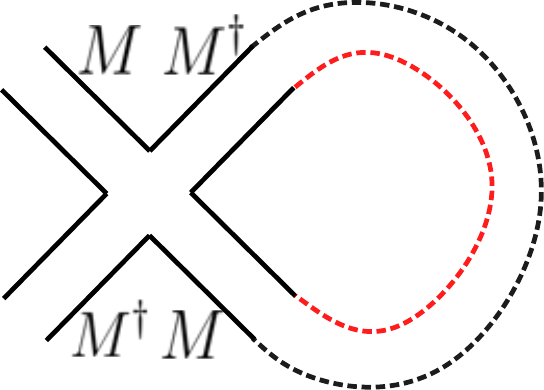}
    \caption{ }
    \label{vertprop1}
\end{subfigure}
\begin{subfigure}{0.3\textwidth}
    \centering
    \includegraphics[width=0.4\linewidth]{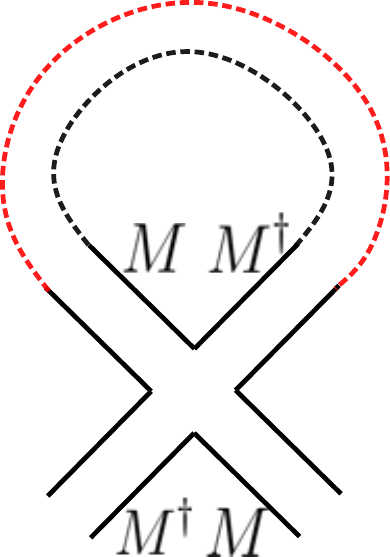}
    \caption{ }
    \label{vertprop2}
\end{subfigure}
\begin{subfigure}{0.3\textwidth}
    \centering
    \includegraphics[width=0.3\linewidth]{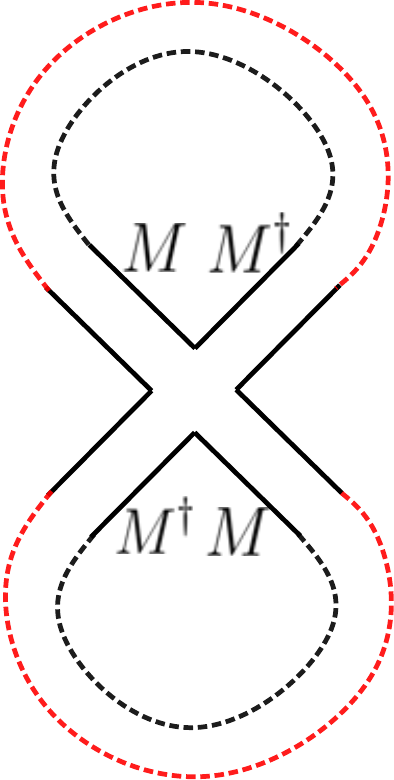}
    \caption{    }
    \label{vertexgraph}
\end{subfigure}
    \caption{Construction of a ribbon graph with one ribbon-graph vertex. The red color indicates the presence of $C_2$ matrix. (a) A vertex with $M^\dagger$ connected to the $M$ immediately to its right. (b) A vertex with $M^\dagger$ connected to the $M$ immediately to its left. (c) The only graph with one vertex with nonzero amplitude.}
    \label{1vert}
\end{figure}
If instead the $M^\dagger$ at a given vertex contracts with the $M$ on its left, we obtain the configuration of Figure~\ref{vertprop2}. Along the face where $C_2$ is present, every propagator contributes a factor of $C_2$. The requirement that this face contains exactly two such factors implies that the remaining two half–edges must connect in such a way that no additional $C_2$ matrices are introduced along that face. This forces the second edge to close on the same vertex, producing the graph in Figure~\ref{vertexgraph}. Thus, whenever an edge connects a vertex to itself, there must be a second such edge.

An analogous mechanism holds for edges connecting two distinct vertices: whenever one propagator connects two given vertices, there must be another propagator connecting the same pair. In other words, edges always come in pairs.

Next, we consider two vertices connected by a single edge as in Figure~\ref{tvop}. By the same argument as above, the presence of $C_2$ on the propagator, together with the restriction that a face can carry only two factors of $C_2$, restricts the second edge to be present as in Figure~\ref{tvtp}.
\begin{figure}[h]
    \centering
    \begin{subfigure}{0.49\textwidth}
    \centering
    \includegraphics[width=0.6
    \linewidth]{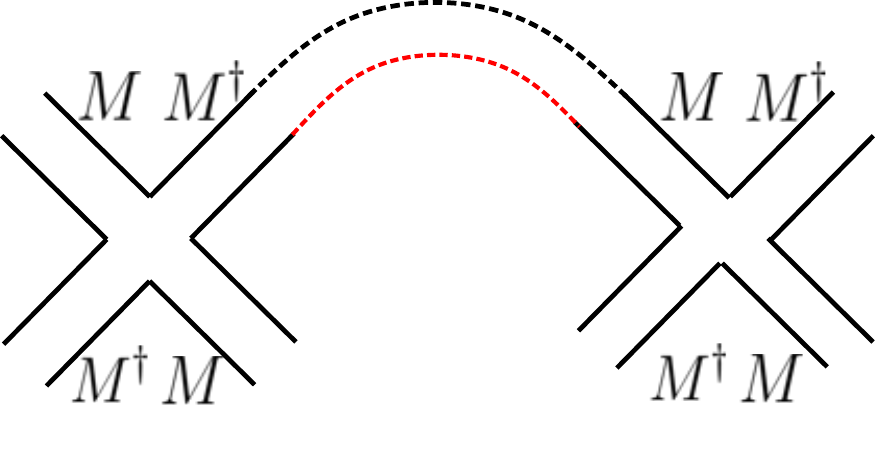}
    \caption{Two vertices connected by one propagator.}
    \label{tvop}
\end{subfigure}
\begin{subfigure}{0.49\textwidth}
    \centering
    \includegraphics[width=0.6
    \linewidth]{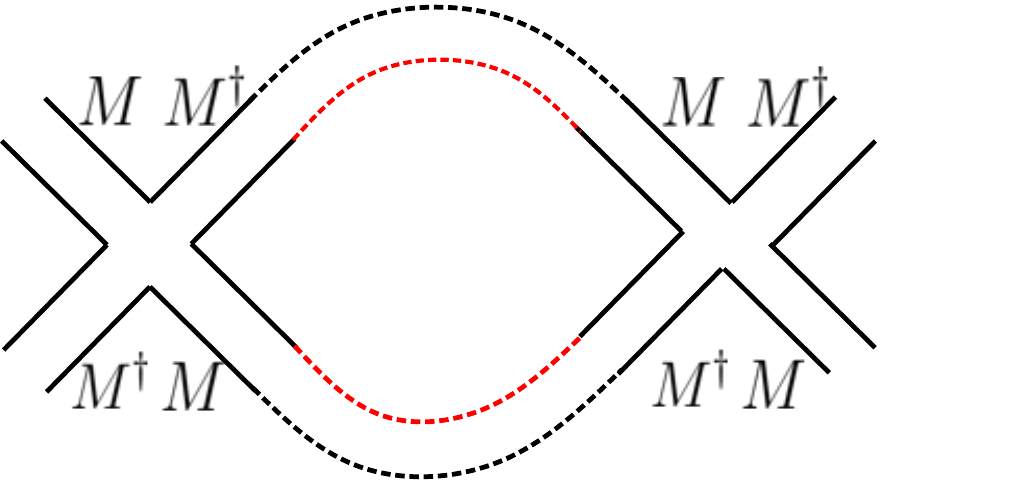}
    \caption{Two vertices connected by two propagators.}
    \label{tvtp}
\end{subfigure}
    \caption{The only possible way 
    to connect two vertices through a propagator with $C_2$ present.}
    \label{2vert}
\end{figure}

From this point, the reasoning is straightforward. After adding a pair of edges between two vertices, the four remaining half–edges are subject to the same constraints as the four half–edges of a single vertex, and the construction repeats. A connected graph with $k$ vertices therefore has the structure shown in Figure~\ref{fig:korder}. By construction, this is the only connected graph with $k$ vertices, up to combinatorial symmetries.

At this point, we can compute the symmetry factor and amplitude of each of these connected graphs. For each connected graph with $k$ vertices there are $(k-1)!$ ways to order these vertices along the loop. Each vertex also has a $\mathbb{Z}_2$ symmetry corresponding to a rotation, so fixing one vertex as a reference yields an additional factor $2^{k-1}$. The $2k$ propagators contribute a factor $N^{-2k}$, while the $k+2$ faces contribute a factor $N^{k+2}$. Thus
\begin{equation}
    \bigl\langle (\Tr[(M^\dagger M)^2])^k\bigr\rangle_{c}
    = (k-1)!\,2^{k-1}N^{-2k} N^{k+2}\;.
\end{equation}
Substituting this into \eqref{fre}, we obtain
\begin{equation}
    \mathcal{F}(g)
    =\frac{N^2}{2}\sum_{k=1}^{\infty}\frac{(-g)^k}{k}
    =-\frac{N^2}{2}\ln (1+g)\;.
\end{equation}
This implies
\begin{equation}
        \mathcal{Z}(g)=(1+g)^{-\frac{N^2}{2}}\;.
\end{equation}
This hints at an equivalence with a simpler self-adjoint model: a Gaussian matrix model. Let us show this by considering the Gaussian self–adjoint matrix model defined by the partition function
\begin{equation}\label{eq:Gaussian_mm_g}
    \mathcal{Z}'(g)=\frac{1}{\mathcal N'}\int dA\,e^{-N\Tr\left(\frac{1+g}{2}A^2\right)}\;,
\end{equation}
where $\mathcal N'$ is chosen such that $\mathcal{Z}'(0)=1$. A simple change of variables shows that
\begin{equation}
    \mathcal{Z}'(g)=(1+g)^{-\frac{N^2}{2}}\;.
\end{equation}
Therefore,
\begin{equation}
    \mathcal{Z}(g)=\mathcal{Z}'(g)\;.
\end{equation}
In other words, the complex matrix model with a propagator containing $C_2$  and a quartic potential \eqref{model1} is \emph{equivalent}, at the level of the partition function, to a self–adjoint matrix model without $C_2$ and with a purely quadratic potential \eqref{eq:Gaussian_mm_g}.

\subsection{Quartic pillow tensor models}\label{app:pillows}

\noindent Having discussed the matrix case, we now consider the analogous tensor model case. In this case, the action of the self–adjoint tensor model is purely quadratic, and, again, we make the equivalence explicit by computing and matching the partition functions. Concretely, we focus on the order $D=3$ complex model with partition function \eqref{pfpil} for the case $R=C_2\otimes\bb 1^{\otimes 2}$. This is
\begin{equation}\label{eq:Z_CT4_app}
    {\cal Z}_{\rm CT}(\lambda)
    =\iint {\cal D}\phi^\dagger {\cal D}\phi\;e^{-\phi^\dagger( C_2^{-1}\otimes {\bb 1}^{\otimes 2})\phi+V_4(\lambda)}\;,
\end{equation}
with
\begin{equation*}
\label{eq:potquarticpillowcomplex}
    V_{\rm 4}(\lambda)= N\frac{\lambda}{2}\sum_{a,b}\sum_{c=1}^3{\phi^\dagger}_a\phi^{a_{\hat{c}}b }{\phi^\dagger}_b\phi^{b_{\hat{c}}a }\;.
\end{equation*}
For convinience, in this Appendix, we refer to the explicit dependence on only the coupling constant $\lambda$. This was made implicit in Subsection \ref{sec:examplestensor}.

The kinetic term of the order-3 complex tensor model \eqref{eq:Z_CT4_app} contains $C_2$ contracted with the first index of the tensor, which breaks the symmetry among the indices/strands and singles out the color–1  index/strand. To reflect this asymmetry, we represent the edges in the Feynman graphs with a highlighted color–1 strand, as shown in Figure~\ref{props}. 
\begin{figure}[h]
    \begin{subfigure}{0.5\textwidth}
    \centering
    \includegraphics[width=
    0.8\linewidth]{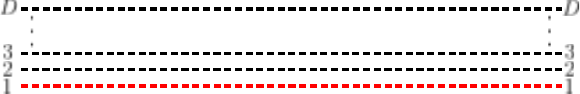}
    \caption{An edge when $D$ is any positive integer.}
    \label{propq}
    \end{subfigure}
    \begin{subfigure}{0.5\textwidth}
    \centering
    \includegraphics[width=
    0.8\linewidth]{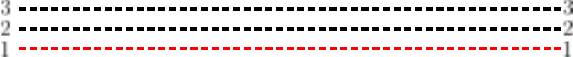}
    \caption{An edge when $D=3$.}
    \label{prop3}
    \end{subfigure}
    \caption{Edges of the complex tensor model \eqref{eq:Z_CT4_app} model are composed of $D$ strands. The strand associated with color 1 is highlighted in red and carries the factor of $C_2$ from the edge.}
    \label{props}
\end{figure}

As for the vertices, there are three types, identical up to a permutation of colors. Each vertex is pillow–like, with one color playing a distinguished role (that is, the contraction of this color includes a factor of $C_2$). This is illustrated in Figure~\ref{fig:pillows}. 

\subsubsection*{Evaluation through resummation.}\label{pillowmod}

\noindent
In this subsection, we evaluate one side of the equivalence by computing the partition function of the complex tensor model using graph resummation.

First, we compute certain averages of the pillow tensor model by resumming the connected graphs of $v$ pillow vertices. There are $v_1$ pillows of color 1, $v_2$ pillows of color 2, and $v_3$ pillows of color 3. That is, $v=v_1+v_2+v_3$. In this case, we compute the connected averages
\begin{equation}\label{v123}
   \left\langle \prod_{k=1}^{3}\left(\sum_{a,b}\varphi^a{\varphi^\dagger}_{a_{\hat k} b_k}{\varphi^\dagger}_{b_{\hat k} a_k}\varphi^b\right)^{v_k}\right\rangle_{c}(\lambda=0)\;,
\end{equation}
for $v_1\geq0$, $v_2\geq0$ and $v_3\geq0$, with $v=v_1+v_2+v_3>0$. In general, the graph of this theory looks like the one shown in Figure~\ref{gp}.
These are necklace graphs composed of a cycle of the three types of pillows, and their amplitudes differ from one another. The average $\eqref{v123}$ is the sum of the amplitudes of all the respective associated connected graphs.
\begin{figure}[h]
    \centering
    \includegraphics[width=
    0.75\linewidth]{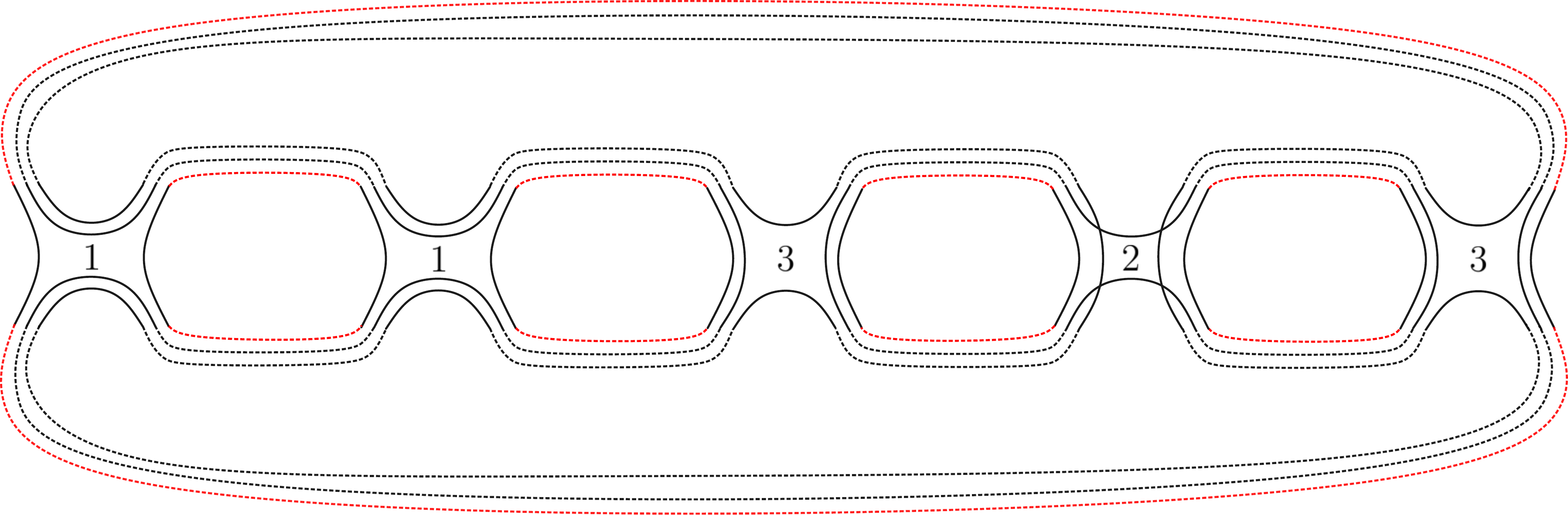}
    \caption{An example of a graph in our complex quartic pillow tensor theory \eqref{eq:Z_CT4_app}.
    This example contains five vertices, including two pillows of color 1, one pillow of color 2, and two pillows of color 3. In general, a graph in this theory is a necklace composed of a sequence of pillows of various colors.}
    \label{gp}
\end{figure}
\begin{figure}[h]
    \begin{subfigure}{0.33\textwidth}
    \centering
    \includegraphics[width=
    0.7\linewidth]{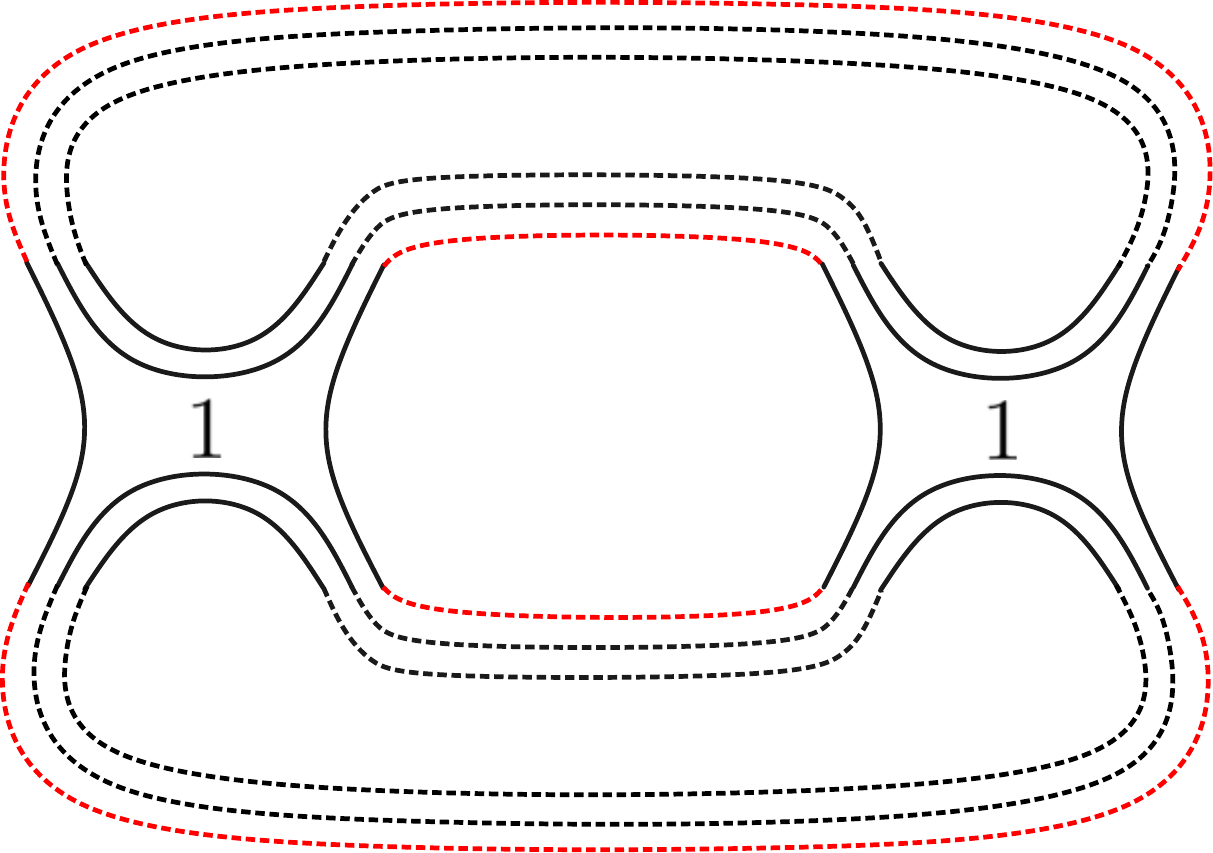}
    \caption{}
    \label{ex1pil}
    \end{subfigure}
    \begin{subfigure}{0.33\textwidth}
    \centering
    \includegraphics[width=
    0.7\linewidth]{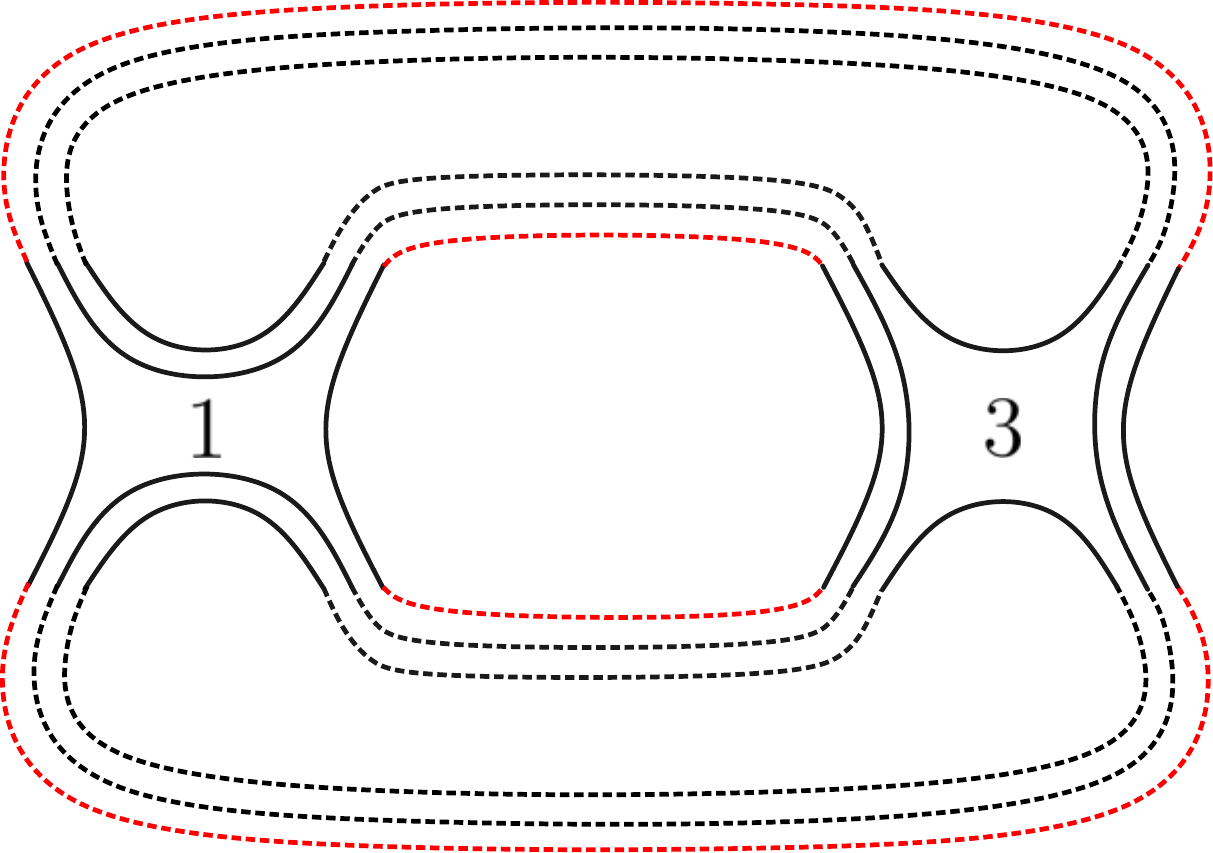}
    \caption{}
    \label{ex2pil}
    \end{subfigure}
    \begin{subfigure}{0.33\textwidth}
    \centering
    \includegraphics[width=
    0.7\linewidth]{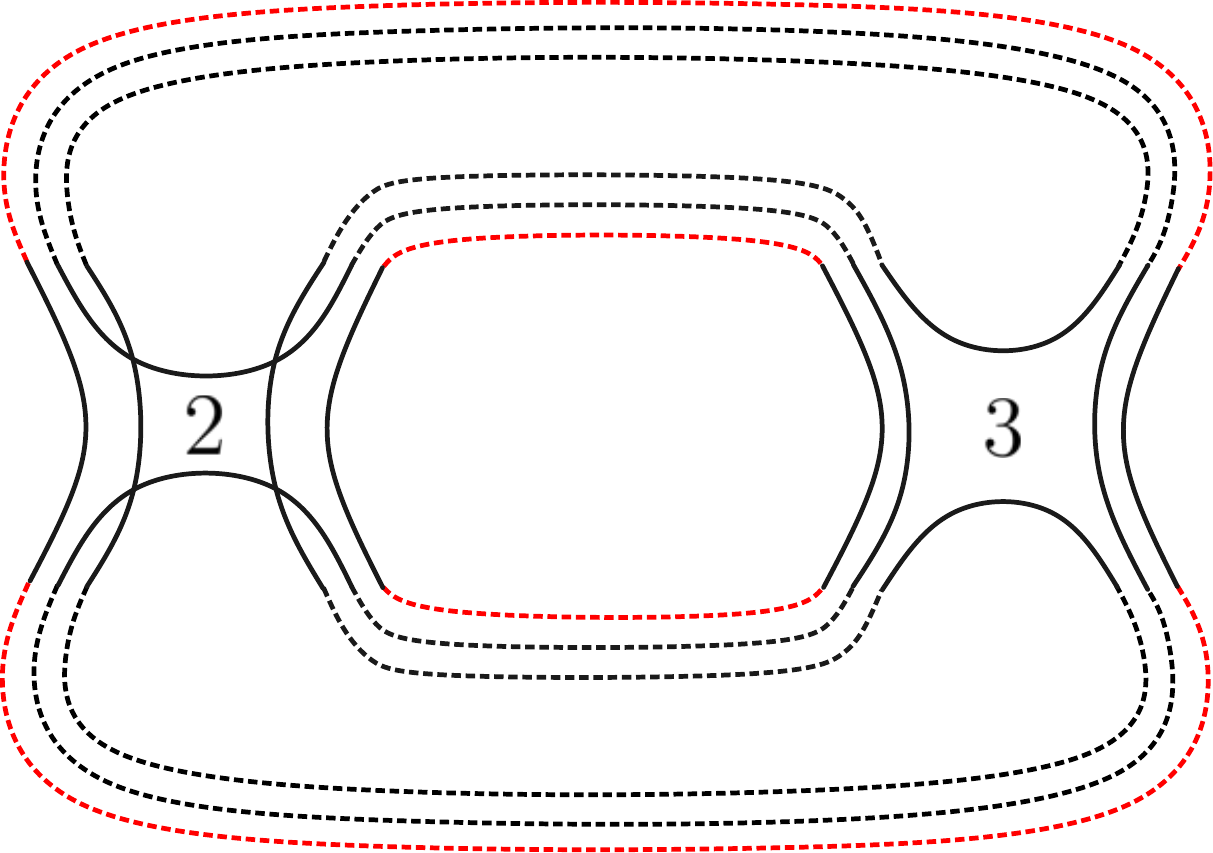}
    \caption{}
    \label{ex3pil}
    \end{subfigure}
    \caption{Examples of graphs in the complex tensor model \eqref{eq:Z_CT4_app}. (a) A graph with only color-1 vertices. (b) A graph with only color-1 and color-3 vertices. (c) A graph with color-2 and color-3 vertices.
    }
    \label{expil}
\end{figure}
Furthermore, for any graph, each edge contributes with a factor of $N^{-1}$ to the amplitude. Since all the $v$ vertices have degree $4$, there are $2v$ edges. Therefore, the edges will contribute with a factor of $N^{2v}$. 

For the faces of color 1, each contributes a factor of $N$. Since each of the $v$ vertices is a part of two color-1 faces, and every color-1 face contains two vertices, every graph with $v$ vertices has $v$ faces of color 1. Therefore, the color-1 faces contribute with a factor of $N^{v}$.

The number of faces of colors 2 and 3 will depend on the numbers $v_3$ and $v_2$, respectively. If $v_2=0$, then there are two faces of color 3. If $v_2>0$, there are $v_2$ faces of color 3. Similarly, if $v_3=0$, there are two faces of color 2, and if $v_3>0$, there are $v_3$ faces of color 2. The multiplicity factor of this graph is $2^{v-1}\frac{v_1!v_2!v_3!}{v}$, obtained by accounting for the permutation of the vertices, the parity of each vertex, and the rotation and parity symmetries of the whole cycle. 

Putting together these factors, the amplitude of a connected graph $G$ is
\begin{equation}
    A(G)=2^{v-1}\frac{v_1!v_2!v_3!}{v}N^{-2v}N^{v}N^{v_2+2 \delta_{v_2,0}}N^{v_3+2 \delta_{v_3,0}}\;. 
\end{equation}
Accounting for permutations of all the colors, there are $\frac{v!}{v_1!v_2!v_3!}$ graphs with the given amounts $v_1, v_2$ and $v_3$. Therefore, we find that
\begin{equation}\label{eq:avpil}
   \left\langle \prod_{k=1}^{3}\left(\sum_{a,b}\varphi^a{\varphi^\dagger}_{a_{\hat k} b_k}{\varphi^\dagger}_{b_{\hat k} a_k}\varphi^b\right)^{v_k}\right\rangle_{c}(0)=2^{v-1}(v-1)!N^{-2v}N^{v}N^{v_2+2 \delta_{v_2,0}}N^{v_3+2 \delta_{v_3,0}}\;.
\end{equation}
This equation gives the desired expression for the average \eqref{v123}.

Given this expression, we prove the following Proposition for the partition function of this tensor model.
\begin{proposition}
Consider a random tensor $\phi \in \mathbb{C}_N^{\otimes 3}$ distributed according to the partition function
\begin{equation}
    {\cal Z}_{\rm CT}(\lambda)
    =\iint {\cal D}\phi^\dagger {\cal D}\phi\;e^{-\phi^\dagger( C_2^{-1}\otimes {\bb 1}^{\otimes 2})\phi+V_4(\lambda)}\;,
\end{equation}
where
\begin{equation}
    V_{\rm 4}(\lambda)= N\frac{\lambda}{2}\sum_{a,b}\sum_{c=1}^3{\phi^\dagger}_a\phi^{a_{\hat{c}}b }{\phi^\dagger}_b\phi^{b_{\hat{c}}a }\;.
\end{equation}
Then,
\begin{equation}
    \frac{ {\cal Z}_{\rm CT}(\lambda)
    }{ {\cal Z}_{\rm CT}(0)
    }=(1-\lambda)^{-\frac{1}{2}(N^2-1)^2}
    (1-(1+N)\lambda)^{-(N^2-1)}
    (1-(1+2N)\lambda)^{-\frac{1}{2}}\;.
\end{equation}
\end{proposition}

\begin{proof}
We start by expressing the free energy ${\cal F}_{\rm CT4}(\lambda)=\ln {\cal Z}_{\rm CT4}(\lambda)$ perturbatively in the coupling $\lambda$. This is
\begin{equation}
    {\cal F}_{\rm CT}(\lambda)-{\cal F}_{\rm CT}(0)
    =\sum_{v=0}^\infty \frac{1}{v!}\left(\frac{N \lambda}{2}\right)^v
    \left\langle \left(
    \sum_{c=1}^3\sum_{a,b}{\phi^\dagger}_{a}\phi^{a_{\hat{c}}b }{\phi^\dagger}_{b}\phi^{b_{\hat{c}}a }
    \right)^v\right\rangle_{c}(0)\;,
\end{equation}
where $\langle \cdot \rangle_{c}(0)$ denotes the connected  Gaussian average at $\lambda=0$. We further separate contributions by color, writing $v=v_1+v_2+v_3$ to obtain
\begin{align}
    {\cal F}_{\rm CT}(\lambda)-{\cal F}_{\rm CT}(0)
    =\sum_{v_1,v_2,v_3=0}^\infty \frac{1}{v_1!v_2!v_3!}
    \left(\frac{N \lambda}{2}\right)^v
    \left\langle \prod_{c=1}^{3}\left(
    \sum_{a,b}{\phi^\dagger}_{a}\phi^{a_{\hat{c}}b }{\phi^\dagger}_{b}\phi^{b_{\hat{c}}a }
    \right)^{v_c}\right\rangle_{c}(0)\;.
\end{align}
Therefore, applying \eqref{eq:avpil}, the free energy evaluates to
\begin{equation}
\begin{split}
    {\cal F}_{\rm CT}(\lambda)-{\cal F}_{\rm CT}(0)
    &=\sum_{v_1=1}^\infty \frac{2^{v-1}}{v_1!}\left(\frac{N \lambda}{2}\right)^v
    (v-1)!N^{-v}N^{4}\\
    &\quad+2\sum_{\substack{v_1=0\\v_2=1}}^\infty \frac{2^{v-1}}{v_1!v_2!}\left(\frac{N \lambda}{2}\right)^v
    (v-1)!N^{-v}N^{2+v_2}\\
    &\quad+\sum_{\substack{v_1=0\\v_2=1\\v_3=1}}^\infty \frac{2^{v-1}}{v_1!v_2!v_3!}\left(\frac{N \lambda}{2}\right)^v
    (v-1)!N^{-v}N^{v_2+v_3}\;.
\end{split}
\end{equation}
This simplifies to
\begin{equation}
\begin{split}
    {\cal F}_{\rm CT}(\lambda)-{\cal F}_{\rm CT}(0)
    &=\frac{1}{2}N^4\sum_{v=1}^\infty \frac{1}{v}\lambda^v
    +N^2\sum_{\substack{v_1=0\\v_2=1}}^\infty\frac{1}{v} \frac{v!}{v_1!v_2!}\lambda^v N^{v_2}+\frac{1}{2}\sum_{\substack{v_1=0\\v_2=1\\v_3=1}}^\infty \frac{1}{v}\frac{v!}{v_1!v_2!v_3!}\lambda^v N^{v_2+v_3}\;.
\end{split}
\end{equation}
By rewriting the sums, we obtain
\begin{equation}
\begin{split}
    {\cal F}_{\rm CT}(\lambda)-{\cal F}_{\rm CT}(0)
    &=\frac{1}{2}N^4\sum_{v=1}^\infty \frac{1}{v}\lambda^{v}
    +N^2\sum_{v=1}^\infty\sum_{v_1+v_2=v} \frac{1}{v}\frac{v!}{v_1!v_2!}\lambda^v N^{v_2}
      -N^2\sum_{v=1}^\infty \frac{1}{v}\lambda^{v}\\
    &\quad+\frac{1}{2}\sum_{v=1}^\infty\sum_{v_1+v_2+v_3=v} \frac{1}{v}\frac{v!}{v_1!v_2!v_3!}\lambda^v N^{v_2+v_3}
      -\sum_{v=1}^\infty\sum_{v_1+v_2=v} \frac{1}{v}\frac{v!}{v_1!v_2!}\lambda^v N^{v_2}\\
    &\quad+\frac{1}{2}\sum_{v=1}^\infty \frac{1}{v}\lambda^{v}\;.
\end{split}
\end{equation}
Several sums repeat and can be combined. Doing so and applying the binomial theorem, we find
\begin{equation}
\begin{split}
    {\cal F}_{\rm CT}(\lambda)-{\cal F}_{\rm CT}(0)
    &=\frac{1}{2}(N^4-2N^2+1)\sum_{v=1}^\infty \frac{1}{v}\lambda^{v}
    +(N^2-1)\sum_{v=1}^\infty\frac{1}{v}[(1+N)\lambda]^v+\frac{1}{2}\sum_{v=1}^\infty \frac{1}{v}[(1+2N)\lambda]^v\\
    &=-\frac{1}{2}(N^2-1)^2\ln(1-\lambda)
      -(N^2-1)\ln(1-(1+N)\lambda)-\frac{1}{2}\ln(1-(1+2N)\lambda)\;.
\end{split}
\end{equation}
Therefore,
\begin{equation}
    \frac{ {\cal Z}_{\rm CT}(\lambda)
    }{ {\cal Z}_{\rm CT}(0)
    }=\frac{{\cal Z}_{\rm CT}[V_4](C_2\otimes {\bb 1}^{\otimes 2})}{{\cal Z}_{\rm CT}[0](C_2\otimes {\bb 1}^{\otimes 2})}=(1-\lambda)^{-\frac{1}{2}(N^2-1)^2}
    (1-(1+N)\lambda)^{-(N^2-1)}
    (1-(1+2N)\lambda)^{-\frac{1}{2}}\;.
\end{equation}
\end{proof}

\subsubsection*{Evaluation through determinant.}
\noindent Next, we show that the partition function of the self-adjoint effective model \eqref{pfpilef} for $R=C_2\otimes\bb 1^{\otimes 2}$ gives the same result. We do so by computing the Gaussian integral explicitly.

Consider order–4 tensor $\Phi$ and the partition function 
\begin{equation}
\label{eq:Z_quart_tm}
   {\cal Z}_{\rm \hat HT}(\lambda)= 
   \iint_{{\rm H}({\Cp_N^{\otimes 2}})}
   {\cal D}\hat\Phi\; e^{-\frac{N}{2}\Tr(\hat\Phi^2)+\hat V_4(\hat \Phi)}\;.
\end{equation}
The potential $\hat V_4$ is given by
\begin{align}
\label{eq:V_self_ad}
    \hat V_{\rm 4}(\lambda)=N\frac{\lambda}{2}\sum_{a_{\hat 1},b_{\hat 1}}\Bigl(
    \hat\Phi^{b_2b_3}_{a_2a_3}\hat\Phi^{a_2a_3}_{b_2b_3}
    +\hat\Phi^{b_2a_3}_{a_2a_3}\hat\Phi^{a_2b_3}_{b_2b_3}
    +\hat\Phi^{a_2b_3}_{a_2a_3}\hat\Phi^{b_2a_3}_{b_2b_3}\Bigr)
    \;.
\end{align}
Figure~\ref{expilef} illustrates a few examples of the connected graphs in this self-adjoint theory.

\vskip 10pt
\begin{figure}[H]
    \begin{subfigure}{0.333\textwidth}
    \centering
    \includegraphics[width=
    0.9\linewidth]{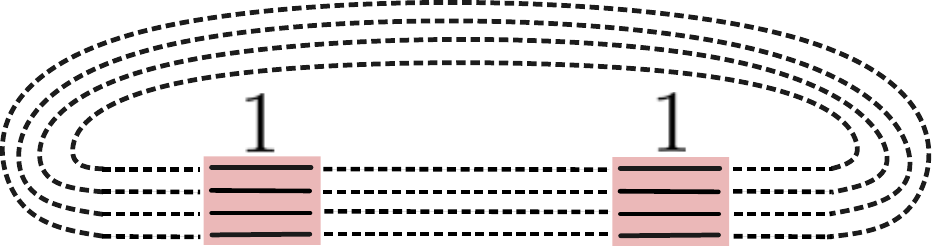}
    \caption{}
    \label{ex1pilef}
    \end{subfigure}
    \begin{subfigure}{0.333\textwidth}
    \centering
    \includegraphics[width=
    0.9\linewidth]{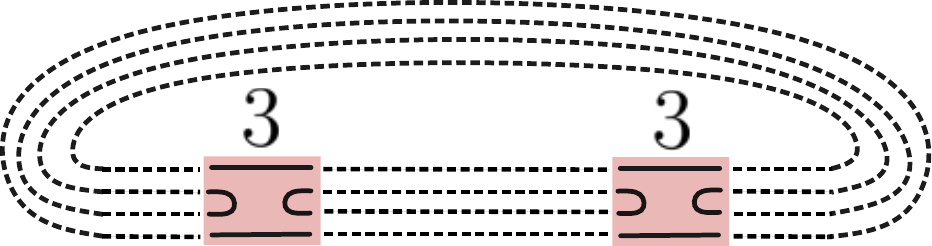}
    \caption{}
    \label{ex2pilef}
    \end{subfigure}
    \begin{subfigure}{0.333\textwidth}
    \centering
    \includegraphics[width=
    0.9\linewidth]{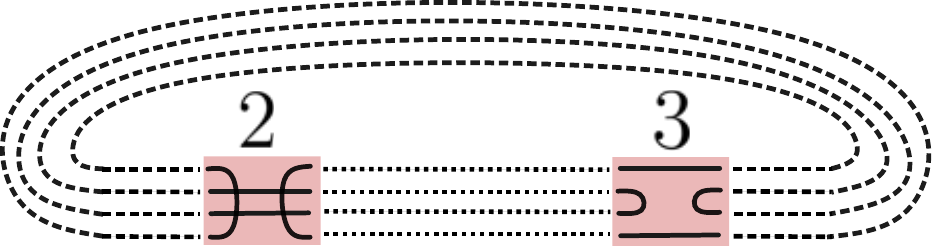}
    \caption{}
    \label{ex3pilef}
    \end{subfigure}
    \caption{Examples of graphs in the effective self-adjoint tensor model 
    \eqref{eq:Z_quart_tm}
    that arises from 
    the complex tensor model \eqref{eq:Z_CT4_app} with quartic pillow interactions \eqref{eq:potquarticpillowcomplex}. 
    The red boxes are the interations in \eqref{eq:V_self_ad} and the black lines are the propagators of this model. For each graph of the complex tensor model, there is a corresponding graph in the effective model: (a) corresponds to the equivalent of the graph in Figure~\ref{ex1pil}, (b) corresponds to the equivalent of the graph in Figure~\ref{ex2pil}, and (c) corresponds to the equivalent of the graph in Figure~\ref{ex3pil}.}
    \label{expilef}
\end{figure}

\noindent Considering the action $\hat S_{\rm 4}[\hat\Phi](\lambda)=\frac{1}{2}\Tr(\hat\Phi^2)-\hat V_{\rm 4}(\lambda)$ with $V_{\rm 4}(\lambda)$ given by \eqref{eq:V_self_ad}, we note this can be re-expressed as
\begin{equation}
\begin{split}
        \hat S_{\rm 4}[\hat\Phi](\lambda)=\frac{N}{2}\sum_{c_{\hat 1},d_{\hat 1}}\sum_{a_{\hat 1},b_{\hat 1}}
        (\hat\Phi^{c_2c_3}_{d_2d_3})^*
        \Bigl[
        (1-\lambda)\delta^{c_2c_3}_{a_2 a_3}\delta_{d_2d_3}^{ b_2b_3}
        - \lambda \delta^{c_2b_3}_{a_2 a_3}\delta_{d_2d_3}^{ b_2c_3}
        - \lambda\delta^{b_2c_3}_{a_2 a_3}\delta_{d_2d_3}^{ c_2b_3}
        \Bigr]\hat\Phi^{a_2a_3}_{b_2b_3}\;.
\end{split} 
\end{equation}
We decompose $\hat\Phi$ into real symmetric and antisymmetric parts under exchange $(a_2a_3)\leftrightarrow(b_2b_3)$,
\begin{equation}
    [\hat\Phi^{(\rm h)}]^{a_2a_3}_{b_2b_3}
    =[\hat\Phi^{(\rm s)}]^{a_2a_3}_{b_2b_3}
    +i[\hat\Phi^{(\rm a)}]^{a_2a_3}_{b_2b_3}\;.
\end{equation}
with $[\hat\Phi^{(\rm s)}]^{a_2a_3}_{b_2b_3}=[\hat\Phi^{(\rm s)}]^{b_2b_3}_{a_2a_3}$ and $[\hat\Phi^{(\rm a)}]^{a_2a_3}_{b_2b_3}=-[\hat\Phi^{(\rm a)}]^{b_2b_3}_{a_2a_3}$. This results in
\begin{equation}
\begin{split}
     {\hat S}_4[\hat\Phi](\lambda)
    &=\frac{N}{2}\sum_{c_{\hat 1},d_{\hat 1}}\sum_{a_{\hat 1},b_{\hat 1}}
    {[\hat\Phi^{(s)}]^{c_2c_3}_{d_2d_3}}
    \Bigl[
    (1-\lambda)\delta^{c_2c_3}_{a_2 a_3}\delta_{d_2d_3}^{ b_2b_3}
    - \lambda \delta^{c_2b_3}_{a_2 a_3}\delta_{d_2d_3}^{ b_2c_3}
    - \lambda\delta^{b_2c_3}_{a_2 a_3}\delta_{d_2d_3}^{ c_2b_3}
    \Bigr][\hat\Phi^{(s)}]^{a_2a_3}_{b_2b_3}\\
    &\quad+ \frac{N}{2}\sum_{c_{\hat 1},d_{\hat 1}}\sum_{a_{\hat 1},b_{\hat 1}}
    {[\hat\Phi^{(a)}]^{c_2c_3}_{d_2d_3}}
    \Bigl[
    (1-\lambda)\delta^{c_2c_3}_{a_2 a_3}\delta_{d_2d_3}^{ b_2b_3}
    - \lambda \delta^{c_2b_3}_{a_2 a_3}\delta_{d_2d_3}^{ b_2c_3}
    - \lambda\delta^{b_2c_3}_{a_2 a_3}\delta_{d_2d_3}^{ c_2b_3}
    \Bigr][\hat\Phi^{(a)}]^{a_2a_3}_{b_2b_3}\;,
\end{split}
\end{equation}
which allows us to computate the partition function \eqref{eq:Z_quart_tm} in terms of standard Gaussian integration methods. With this key insight in hand, we proceed to prove the following proposition for the partition function.
\begin{proposition}
The partition function ${\cal Z}_{\rm \hat HT}(\lambda)$ \eqref{eq:Z_quart_tm} can be computed explicitly in terms of the coupling $\lambda$ and it results in
\begin{equation}
    \frac{{\cal Z}_{\rm \hat HT}(\lambda)}{{\cal Z}_{\rm \hat HT}(0)}
   =(1-\lambda)^{-\frac{1}{2}(N^2-1)^2}
    (1-(1+N)\lambda)^{-(N^2-1)}
    (1-(1+2N)\lambda)^{-\frac{1}{2}}\;.
\end{equation}
\end{proposition}

\begin{proof}
Given that the partition function \eqref{eq:Z_quart_tm} can be expressed in Gaussian form as
\begin{equation}
     {\hat S}_4[\hat\Phi](\lambda)
    =\frac{N}{2}\sum_{c_{\hat 1},d_{\hat 1}}\sum_{a_{\hat 1},b_{\hat 1}}
    \hat\Phi^{c_2c_3}_{d_2d_3}
    C^{c_2c_3;b_2b_3}_{d_2d_3;a_2a_3}
    \hat\Phi^{a_2a_3}_{b_2b_3}\;,
\end{equation}
with\footnote{Not to be confused with $C_k$ defined in \eqref{cnprop}.}
\begin{equation}
\label{eq:Coperator}
   C^{c_2c_3;b_2b_3}_{d_2d_3;a_2a_3}
   =(1-\lambda)\delta^{c_2c_3}_{a_2 a_3}\delta_{d_2d_3}^{ b_2b_3}
   - \lambda \delta^{c_2b_3}_{a_2 a_3}\delta_{d_2d_3}^{ b_2c_3}
   - \lambda\delta^{b_2c_3}_{a_2 a_3}\delta_{d_2d_3}^{ c_2b_3}\;.
\end{equation}
By decomposing this operator $C$ \eqref{eq:Coperator} into four blocks according to whether the index pairs coincide or not, we obtain
\begin{equation}
    C=C_{(\ne,\ne)}\oplus C_{(=,\ne)}\oplus C_{(\ne,=)}\oplus C_{(=,=)}\;,
\end{equation}
where
\begin{itemize}
    \item $j\neq k, i\neq l$ and $j'\neq k', i'\neq l'$ space:
\begin{equation}
    C_{(\ne,\ne)}:=(1-\lambda)I_{N}\otimes I_{N-1}\otimes I_{N}\otimes I_{N-1}\;.
\end{equation}

    \item $j= k, i\neq l$ and $j'=k', i'\neq l'$ space:
\begin{equation}
    C_{(=,\ne)}:=[(1-\lambda)I_{N}-\lambda J_{N}]\otimes I_{N}\otimes I_{N-1}\;.
\end{equation}

    \item $j\ne k, i= l$ and $j'\ne k', i'= l'$ space:
\begin{equation}
    C_{(\ne,=)}:=I_{N}\otimes I_{N-1}\otimes[(1-\lambda)I_{N}-\lambda J_{N}]\;.
\end{equation}

    \item $j= k, i= l$ and $j'= k', i'= l'$ space:
\begin{equation}
    C_{(=,=)} := (1-\lambda)I_{N}\otimes I_{N}
    -I_{N}\otimes \lambda J_{N}-\lambda J_{N}\otimes I_{N}\;.
\end{equation}
\end{itemize}
In these expressions, $I_N$ denotes the $N\times N$ identity and $J_N$ is the $N\times N$ matrix with all entries equal to $1$. The matrix $J_N$ has $N-1$ eigenvalues $0$ and one eigenvalue $N$, so
\begin{equation}\label{detjn}
    \det [(1-\lambda)I_{N}-\lambda J_{N}] =(1-\lambda)^{N-1}\bigl(1-(1+N)\lambda\bigr)\;.
\end{equation}
Taking tensor products of the corresponding eigenvectors, one also finds
\begin{equation}\label{detjnjn}
    \det \bigl[(1-\lambda)I_{N}\otimes I_{N}-I_{N}\otimes \lambda J_{N}- \lambda J_{N}\otimes I_{N}\bigr]
    =(1-\lambda)^{(N-1)^2}
    \bigl(1-(1+N)\lambda\bigr)^{2(N-1)}
    \bigl(1-(1+2N)\lambda\bigr)\;.
\end{equation}
Therefore, the determinants of the blocks are
\begin{equation}
    \det (
    C_{(\ne,\ne)}) = (1-\lambda)^{N^2(N-1)^2}\;,
\end{equation}
\begin{equation}
    \det (C_{(=,\ne)}) = \det (
    C_{(\ne,=)}) = \bigl[(1-\lambda)^{N-1}(1-(1+N)\lambda)\bigr]^{N(N-1)}\;,
\end{equation}
and
\begin{equation}
    \det (
    C_{(=,=)})  =(1-\lambda)^{(N-1)^2}
    \bigl(1-(1+N)\lambda\bigr)^{2(N-1)}
    \bigl(1-(1+2N)\lambda\bigr)\;.
\end{equation}
Combining the blocks, we obtain
\begin{equation}
\begin{aligned}
    \det(C)
    &=\det (
    C_{(\ne,\ne)})\det (
    C_{(=,\ne)})\det (
    C_{(\ne,=)})\det (
    C_{(=,=)})\\
    &=(1-\lambda)^{(N^2-1)^2}
    \bigl(1-(1+N)\lambda\bigr)^{2(N^2-1)}
    \bigl(1-(1+2N)\lambda\bigr)\;.
    \end{aligned}
\end{equation}
Since $\det(C)=1$ when $\lambda=0$, the Gaussian integral yields
\begin{equation}
    \frac{{\cal Z}_{\rm \hat HT}(\lambda)}{{\cal Z}_{\rm \hat HT}(0)}=
    \det(C)^{-1/2}\;,
\end{equation}
hence
\begin{equation}
     \frac{{\cal Z}_{\rm \hat HT}(\lambda)}{{\cal Z}_{\rm \hat HT}(0)}=\frac{{\cal Z}_{\rm \hat HT}[ \hat{V}_{\rm 4},Y[C_2\otimes {\bb 1^{\otimes 2}}]]}{{\cal Z}_{\rm \hat HT}[ 0,Y[C_2\otimes {\bb 1^{\otimes 2}}]]}
     =(1-\lambda)^{-\frac{1}{2}(N^2-1)^2}
     (1-(1+N)\lambda)^{-(N^2-1)}
     (1-(1+2N)\lambda)^{-\frac{1}{2}}\;.
\end{equation}
\end{proof}
We recover exactly the result of the resummation computation.

\bibliographystyle{utphys}
\bibliography{ref}

\end{document}